\documentclass[sn-mathphys,Numbered]{sn-jnl}

\usepackage{graphicx}%
\usepackage{multirow}%
\usepackage{amsmath,amssymb,amsfonts}%
\usepackage{amsthm}%
\usepackage{mathrsfs}%
\usepackage[title]{appendix}%
\usepackage{xcolor}%
\usepackage{textcomp}%
\usepackage{manyfoot}%
\usepackage{booktabs}%
\usepackage{algorithm}%
\usepackage{algorithmicx}%
\usepackage{algpseudocode}%
\usepackage{listings}%

\usepackage[utf8]{inputenc}
\usepackage{amssymb}
\usepackage{amsmath,bm}
\usepackage{mathrsfs}
\usepackage{xcolor}
\usepackage{scalerel}
\usepackage{url}
\usepackage{array,multirow}
\usepackage{pgfplots}
\usetikzlibrary{plotmarks, patterns, tikzmark, positioning}
\usetikzlibrary{matrix,decorations.pathreplacing,calc}
\usepackage{hf-tikz}
\usepackage[nolist]{acronym}

\usepackage{mathtools}

\mathtoolsset{showonlyrefs}

\definecolor{darkgreen}{rgb}{0,0.5,0}

\def\ve#1{{\mathchoice{\mbox{\boldmath$\displaystyle #1$}}%
		{\mbox{\boldmath$\textstyle #1$}}%
		{\mbox{\boldmath$\scriptstyle #1$}}%
		{\mbox{\boldmath$\scriptscriptstyle #1$}}}}

\pgfplotsset{
compat=1.17,
mystyle/.style={
    scale only axis,
    width=0.7\columnwidth,
    height=0.5\columnwidth,
    label style={inner sep=0, font=\normalsize}, 
    tick label style={font=\scriptsize},
    legend style={font=\scriptsize},
    mark size=3,
    major grid style={dashed},
    line width=0.8pt,
    axis line style = thin}
}

\newcolumntype{M}[1]{>{\centering\arraybackslash}m{#1}}

\theoremstyle{thmstyleone}%
\newtheorem{theorem}{Theorem}
\newtheorem{lemma}{Lemma}%
\newtheorem{proposition}[theorem]{Proposition}%

\theoremstyle{thmstyletwo}%
\newtheorem{remark}{Remark}%
\newtheorem{problem}{Problem}%

\theoremstyle{thmstylethree}%
\newtheorem{definition}{Definition}%

\newenvironment{proofsketch}{\noindent\emph{Sketch of the Proof.}}{\hfill$\qed$}

\newcommand{\Z}[1]{\ensuremath{\mathbb{Z}_{#1}}} 
\newcommand{\Fqm}{\ensuremath{\mathbb F_{q^m}}}

\newcommand{\Fqs}{\ensuremath{\mathbb F_{q^s}}}
\newcommand{\Fq}{\ensuremath{\mathbb F_{q}}}

\newcommand{\F}{\ensuremath{\mathbb F}}
\newcommand{\ZZ}{\ensuremath{\mathbb{Z}}}
\newcommand{\set}[1]{\ensuremath{\mathcal{#1}}}

\newcommand{\intervallincl}[2]{\ensuremath{[#1,#2]}}

\newcommand{\aut}{\ensuremath{\sigma}}
\newcommand{\autMat}[3]{\ensuremath{\aut^{#3}({#1}_{#2})}}
\newcommand{\autVec}[3]{\ensuremath{\aut^{#3}({#1}_{#2})}}

\newcommand{\SkewPolyringZeroDer}{\ensuremath{\Fqm[x;\aut]}}

\newcommand{\MultSkewPolyringZeroDer}{\ensuremath{\Fqm[x,y_1,\dots,y_\intOrder;\aut]}}

\newcommand{\opev}[3]{\ensuremath{{#1}(#2)_{#3}}}

\newcommand{\op}[2]{\ensuremath{\mathcal{D}_{#1}\left(#2\right)}}
\newcommand{\opexp}[3]{\ensuremath{\mathcal{D}_{#1}^{#3}(#2)}}

\newcommand{\IPop}[1]{\mathcal{I}_{#1}^{\mathrm{op}}}

\newcommand{\evalMap}[2]{\mathscr{E}_{#1}^{(#2)}}

\newcommand{\algoname}[1]{{\normalfont\textsc{#1}}} 

\newcommand{\defeq}{:=}

\renewcommand{\bar}{\overline}

\DeclareMathOperator{\wt}{wt}
\DeclareMathOperator{\rk}{rk}

\DeclareMathOperator{\diag}{diag}

\newcommand{\rker}{\ensuremath{\ker_r}}
\newcommand{\lker}{\ensuremath{\ker_l}}

\renewcommand{\vec}[1]{\ve{#1}} 
\newcommand{\mat}[1]{\ensuremath{\bm{#1}}}

\newcommand{\genMoore}{{\ensuremath{\aut}-Generalized Moore}}

\newcommand{\opVandermonde}[3]{\ensuremath{\mat{V}_{#1}(#2)_{#3}}}
\newcommand{\genNorm}[2]{\ensuremath{\mathcal{N}_{#1}(#2)}}

\renewcommand{\a}{\vec{a}}
\renewcommand{\b}{\vec{b}}

\newcommand{\e}{\vec{e}}
\newcommand{\f}{\vec{f}}

\renewcommand{\l}{\vec{l}}

\renewcommand{\t}{\vec{t}}
\renewcommand{\u}{\vec{u}}
\renewcommand{\v}{\vec{v}}

\newcommand{\x}{\vec{x}}

\newcommand{\A}{\mat{A}}

\newcommand{\D}{\mat{D}}
\newcommand{\E}{\mat{E}}

\newcommand{\Q}{\mat{Q}}
\newcommand{\R}{\mat{R}}

\newcommand{\U}{\mat{U}}
\newcommand{\W}{\mat{W}}
\newcommand{\X}{\mat{X}}

\renewcommand{\Z}{\mat{Z}}
\newcommand{\0}{\ensuremath{\mathbf 0}}
\newcommand{\vecbeta}{\ensuremath{\boldsymbol{\beta}}}
\newcommand{\vecgamma}{\ensuremath{\boldsymbol{\gamma}}}
\newcommand{\veckappa}{\ensuremath{\boldsymbol{\varkappa}}}
\newcommand{\veczeta}{\ensuremath{\boldsymbol{\zeta}}}
\newcommand{\vecepsilon}{\ensuremath{\boldsymbol{\epsilon}}}
\newcommand{\vecxi}{\ensuremath{\boldsymbol{\xi}}}
\newcommand{\nVec}{\ensuremath{\vec{n}}}
\newcommand{\Nvec}{\ensuremath{\vec{N}}}

\newcommand{\mycode}[1]{\ensuremath{\mathcal{#1}}}

\newcommand{\liftedIntLinRS}[1]{\ensuremath{\mathrm{LILRS}[#1]}}

\newcommand{\SumRankWeight}{\ensuremath{\wt_{\Sigma R}}}

\newcommand{\SubspaceDist}{d_S}
\newcommand{\SumSubspaceDist}{d_{\ensuremath{\Sigma}S}}

\newcommand{\myspace}[1]{\mathcal{#1}}
\newcommand{\myspaceDual}[1]{\myspace{#1}^{\perp}}
\newcommand{\Rowspace}[1]{\ensuremath{{\left\langle #1 \right\rangle}_{q}}}

\newcommand{\RowspaceHuge}[1]{\ensuremath{\scaleleftright[3ex]{\Biggl\langle}{\!\!\!#1\!\!\!}{\Biggr\rangle}\!\!\!\!\raisebox{-21pt}{\scriptsize{q}}}}

\newcommand{\Grassm}[1]{\myspace{G}_q(#1)}

\newcommand{\ProjspaceAny}[1]{\myspace{P}_q(#1)}

\newcommand{\List}{\ensuremath{\mathcal{L}}}

\newcommand{\quadbinom}[2]{\left[\genfrac{}{}{0pt}{}{#1\vphantom{N_N}}{#2\vphantom{N}}\right]_{q}}

\newcommand{\gammaq}{\ensuremath{\kappa_q}}

\newcommand{\oh}[1]{\bnd{O}{#1}}
\newcommand{\softoh}[1]{\bnd{\widetilde{O}}{#1}}
\newcommand{\bnd}[2]{\ensuremath{#1\mathopen{}\left(#2\right)\mathclose{}}}

\newcommand{\OMul}[1]{\mathcal{M}(#1)}

\newcommand{\nTransmit}{\ensuremath{n_t}}
\newcommand{\nTransmitVec}{\ensuremath{\nVec_{t}}}
\newcommand{\nTransmitBarVec}{\ensuremath{\bar{\nVec}_{t}}}
\newcommand{\nReceive}{\ensuremath{n_r}}
\newcommand{\nReceiveVec}{\ensuremath{\nVec_r}}
\newcommand{\insertions}{\ensuremath{\gamma}}
\newcommand{\insertionsmax}{\ensuremath{\insertions}_\mathsf{max}}

\newcommand{\insertionsVec}{\ensuremath{\vecgamma}}
\newcommand{\deletions}{\ensuremath{\delta}}
\newcommand{\deletionsVec}{\vec{\delta}}
\newcommand{\deviations}{\ensuremath{\varkappa}}
\newcommand{\deviationsVec}{\ensuremath{\veckappa}}

\newcommand{\txSpace}{\ensuremath{\myspace{V}}}
\newcommand{\rxSpace}{\ensuremath{\myspace{U}}}
\newcommand{\errSpace}{\ensuremath{\myspace{E}}}
\newcommand{\txSpaceVec}{\ensuremath{\vec{\txSpace}}}
\newcommand{\rxSpaceVec}{\ensuremath{\vec{\rxSpace}}}
\newcommand{\errSpaceVec}{\ensuremath{\vec{\errSpace}}}

\newcommand{\shot}[2]{\ensuremath{{#1}^{(#2)}}}
\newcommand{\nTransmitShot}[1]{\ensuremath{n_t^{(#1)}}}
\newcommand{\nReceiveShot}[1]{\ensuremath{n_r^{(#1)}}}
\newcommand{\insertionsShot}[1]{\ensuremath{\shot{\insertions}{#1}}}
\newcommand{\deletionsShot}[1]{\ensuremath{\shot{\deletions}{#1}}}
\newcommand{\deviationsShot}[1]{\ensuremath{\shot{\deviations}{#1}}}

\newcommand{\txSpaceShot}[1]{\ensuremath{\shot{\txSpace}{#1}}}
\newcommand{\rxSpaceShot}[1]{\ensuremath{\shot{\rxSpace}{#1}}}
\newcommand{\errSpaceShot}[1]{\ensuremath{\shot{\errSpace}{#1}}}
\newcommand{\rankErrShot}[1]{\ensuremath{\shot{\rankErr}{#1}}}

\newcommand{\pe}{\ensuremath{\alpha}}

\newcommand{\degConstraint}{\ensuremath{D}}
\newcommand{\degConstraintUnique}{\ensuremath{D_u}}
\newcommand{\intOrder}{\ensuremath{s}}

\newcommand{\shots}{\ensuremath{\ell}}
\newcommand{\delOp}[1]{\ensuremath{\mathcal{H}_{#1}}}

\newcommand{\IntParam}{\intOrder'}

\newcommand{\rankErr}{\ensuremath{t}}

\newcommand{\liftedLinRS}[1]{\ensuremath{\mathrm{LLRS}[#1]}}

\DeclareMathOperator{\sumDim}{\ensuremath{\dim_{\Sigma}}}

\newcommand{\LODecMat}{{\vec L}}
\newcommand{\RFmat}{\ensuremath{\Q_R}}
\newcommand{\RFvec}{\ensuremath{\f_R}}
\newcommand{\intMat}{\ensuremath{\R_I}}
\newcommand{\h}{\vec{h}}
\newcommand{\T}{\vec{T}}

\newcommand{\Qspace}{\mathcal{Q}}

\newcommand{\myalpha}{\xi}

\newcommand{\Eset}{\mathcal{E}}

\newcommand{\q}{\vec{q}}
\newcommand{\Qbar}{\bar{\Q}}

\renewcommand{\b}{\vec{b}}

\newcommand{\n}{\vec{n}}

\newcommand{\normweight}{\ensuremath{\lambda}}
\newcommand{\avgnormweight}{\ensuremath{\bar{\normweight}}}
\newcommand{\normdist}{\ensuremath{\eta}}
\newcommand{\nTransmitBar}{\ensuremath{\bar{n}_t}}

\newcommand{\LOlike}{Loidreau--Overbeck-like}

\newcommand{\compSet}[1]{\ensuremath{\set{T}_{#1}}}
\newcommand{\numSubTuples}[1]{\ensuremath{\textsf{N}_{\Sigma S}(#1)}}
\newcommand{\assignDet}{\ensuremath{\leftarrow}}
\newcommand{\assignRand}{\ensuremath{\xleftarrow{\$}}}
\newcommand{\Nbar}{\ensuremath{\bar{N}}}
\newcommand{\NbarVec}{\ensuremath{\bar{\Nvec}}}

\begin{acronym}
 \acro{BCH}{Bose--Chaudhuri--Hocquenghem}
 \acro{BMD}{bounded minimum distance}
 \acro{SRS}{skew Reed--Solomon}
 \acro{ISRS}{interleaved skew Reed--Solomon}
 \acro{lclm}{least common left multiple}
 \acro{LRS}{linearized Reed--Solomon}
 \acro{LLRS}{lifted linearized Reed--Solomon}
 \acro{ILRS}{interleaved linearized Reed--Solomon}
 \acro{LILRS}{lifted interleaved linearized Reed--Solomon}
 \acro{MDS}{maximum distance separable}
 \acro{MRD}{maximum rank distance}
 \acro{MSRD}{maximum sum-rank distance}
 \acro{MSD}{maximum skew distance}
\acro{RS}{Reed--Solomon}
\end{acronym}

\begin{document}

\title[Fast Decoding of Lifted Interleaved Linearized Reed--Solomon Codes for Multishot Network Coding]{Fast Decoding of Lifted Interleaved Linearized Reed--Solomon Codes for Multishot Network Coding}


\author*[1]{\fnm{Hannes} \sur{Bartz}}\email{hannes.bartz@dlr.de}

\author[2]{\fnm{Sven} \sur{Puchinger}}\email{sven.puchinger@tum.de}

\affil*[1]{
    \centering
    \orgdiv{Institute of Communications and Navigation},\\ 
    \orgname{German Aerospace Center (DLR)}, 
    \orgaddress{
        \postcode{D-82234},
        \city{Oberpfaffenhofen}, 
        \country{Germany}}}

\affil[2]{
    \centering
    \orgname{Hensoldt Sensors GmbH}, \\
    \orgaddress{
        \postcode{D-89077},
        \city{Ulm},  
        \country{Germany}}}


\abstract{
Mart{\'\i}nez-Pe{\~n}as and Kschischang (IEEE Trans.\ Inf.\ Theory, 2019) proposed lifted linearized Reed--Solomon codes as suitable codes for error control in multishot network coding.
We show how to construct and decode \ac{LILRS} codes.
Compared to the construction by Mart{\'\i}nez-Pe{\~n}as--Kschischang, interleaving allows to increase the decoding region significantly and decreases the overhead due to the lifting (i.e., increases the code rate), at the cost of an increased packet size.
We propose two decoding schemes for \ac{LILRS} that are both capable of correcting insertions and deletions beyond half the minimum distance of the code by either allowing a list or a small decoding failure probability. 
We propose a probabilistic unique {\LOlike} decoder for \ac{LILRS} codes and an efficient interpolation-based decoding scheme that can be either used as a list decoder (with exponential worst-case list size) or as a probabilistic unique decoder.
We derive upper bounds on the decoding failure probability of the probabilistic-unique decoders which show that the decoding failure probability is very small for most channel realizations up to the maximal decoding radius.
The tightness of the bounds is verified by Monte Carlo simulations.
}

\keywords{Multishot network coding, subspace codes, sum-subspace metric, multishot operator channel, interleaving}


\pacs[MSC Classification]{94B35, 94B05}

\acresetall

\maketitle

\section{Introduction}\label{sec:introduction}

Network coding~\cite{ahlswede2000network} is a powerful approach to achieve the capacity of multicast networks. 
Unlike the classical routing schemes, network coding allows to mix (e.g. linearly combine) incoming packets at intermediate nodes. 
Kötter and Kschischang proposed codes in the subpsace metric as a suitable tool for error correction in (random) linear network coding~\cite{koetter2008coding} and defined the corresponding operator channel model.
In an operator channel two type of errors, namely \emph{insertions} and \emph{deletions}, can occur.
Insertions correspond to additional dimensions that are not contained in the transmitted space whereas deletions correspond to dimensions that are removed from the transmitted space.

The single-shot scenario from~\cite{koetter2008coding} was extended to the multishot case, i.e. the transmission over several instances of the operator channel, in~\cite{nobrega2009multishot}.
It was shown in \cite{martinez2019reliable} that \ac{LLRS} codes provide reliable and secure coding schemes for noncoherent multishot network coding, a scenario where the in-network linear combinations are not known or used at the receiver, under an adversarial model in the sum-subspace metric.

An $\intOrder$-interleaved code is a direct sum of $\intOrder$ codes of the same length (called constituent codes).
This means that if the constituent codes are over $\Fq$, then the interleaved code can be viewed as a (not necessarily linear) code over $\Fqs$.
In the Hamming and rank metric, there are various decoders that can significantly increase the decoding radius of a constituent code by collaboratively decoding in an interleaved variant thereof. 
Such decoders are known in the Hamming metric for Reed--Solomon \cite{krachkovsky1997decoding,bleichenbacher2003decoding,coppersmith2003reconstructing,parvaresh2004multivariate,brown2004probabilistic,parvaresh2007algebraic,schmidt2007enhancing,schmidt2009collaborative,cohn2013approximate,nielsen2013generalised,wachterzeh2014decoding,puchinger2017irs,yu2018simultaneous} and in general algebraic geometry codes \cite{brown2005improved,kampf2014bounds,puchinger2019improved}, and in the rank metric for Gabidulin codes \cite{loidreau2006decoding,sidorenko2010decoding,sidorenko2011skew,sidorenko2014fast,wachter2014list,puchinger2017row,puchinger2017alekhnovich,bartz2020fast}.
All of these decoders have in common that they are either list decoders with exponential worst-case and small average-case list size, or probabilistic unique decoders that fail with a very small probability.

Interleaving was suggested in \cite{silva2008rank} as a method to decrease the overhead in lifted Gabidulin codes for error correction in noncoherent (single-shot) network coding, at the cost of a larger packet size while preserving a low decoding complexity. It was later shown \cite{sidorenko2011skew,wachter2014list,bartz2018efficient} that it can also increase the error-correction capability of the code using suitable decoders for interleaved Gabidulin codes.

\subsection{Our Contributions}

In this paper, we define and analyze \ac{LILRS} codes that are obtained by lifting \ac{ILRS} codes as considered in~\cite{bartz2022fast}.
We propose and analyze two decoding schemes that both allow for decoding insertions and deletions beyond the unique decoding region by allowing a (potentially exponential) list or a small decoding failure probability.

First, we propose a {\LOlike} probabilistic unique decoder and derive an upper bound on the decoding failure probability.
The upper bound shows the the decoder succeeds with an overwhelming probability close to one for random realizations of the multishot operator channel that stay within the decoding region.

The second decoding approach is a novel interpolation-based list decoder that is based on the list decoder by Wachter-Zeh and Zeh \cite{wachter2014list} for interleaved Gabidulin codes.
We derive a decoding region for the codes in the sum-subspace metric, analyze the complexity of the decoder and give an exponential upper bound on the list size.
We derive an upper bound on the decoding failure probability of the interpretation as a probabilistic-unique decoder by relating the success conditions to the {\LOlike} decoder. 

Compared to \cite{martinez2019reliable}, we decrease the relative overhead introduced by lifting (or equivalently, increase the rate for the same code length and block size) and at the same time extend the decoding region, especially for insertions, significantly.
These advantages come at the cost of a larger packet size of the packets within the network and
a supposedly small failure probability.
By considering the decoding problem in the complementary code, we show how the proposed \ac{LILRS} coding schemes can be used to improve the decoding region w.r.t. \emph{deletions} significantly.

Moreover, for the case $\intOrder=1$ (no interleaving), our algorithm does not require the assumption from~\cite[Sec.~V.H]{martinez2019reliable} that $\nReceive\leq\nTransmit$, where $\nTransmit$ and $\nReceive$ denotes the sum of the dimensions of the transmitted and received subspaces, respectively.
Hence, the proposed decoder works in cases in which \cite{martinez2019reliable} does not work.

Compared to the conference version~\cite{bartz2021decoding} this work contains several new results, such as e.g. the {\LOlike} decoder, the strict upper bound on the decoding failure probability as well as the improved deletion-correction capability from complementary codes.

The main results of this paper, in particular the improvements upon the existing noninterleaved variants, are illustrated in Table~\ref{tab:overview_decoders}.

\begin{table*}[ht!]
  \small
  \caption{%
    Overview of new decoding regions and computational complexities. Parameters: interleaving parameter $\intOrder$ (usually $\intOrder \ll \nTransmit$), $\nTransmit$ resp.~$\nReceive$ is the dimension of the transmitted resp.~received subspace, $\insertions$ and $\deletions$ the overall number of insertions resp.~deletions.
    $\OMul{n} \in \oh{n^{1.635}}$ is the cost (in operations in $\Fqm$) of multiplying two skew-polynomials of degree at most $n$ and $\omega < 2.373$ is the matrix multiplication exponent.
}\label{tab:overview_decoders}
\begin{minipage}{\textwidth}
\begin{center}
\newcommand{\specialcell}[2][c]{{\def\arraystretch{1.2}\begin{tabular}[#1]{@{}l@{}}#2\end{tabular}}}
\newcommand{\specialcellsmalldistleftalign}[2][c]{{\def\arraystretch{1}\begin{tabular}[#1]{@{}l@{}}#2\end{tabular}}}
\newcommand{\specialcellcenter}[2][c]{{\def\arraystretch{1}\begin{tabular}[#1]{@{}c@{}}#2\end{tabular}}}
\def\arraystretch{2.0}
\begin{tabular}{l|l|l|l|l}
\specialcellsmalldistleftalign{Code Class / \\ (Decoder)} & Decoding Region & \specialcellsmalldistleftalign{List Size $|\List|$ \\ Failure Prob. $P_f$} &\specialcellsmalldistleftalign{Complexity\\(over $\Fqm$)} & Reference(s) \\
\hline \hline
\specialcellsmalldistleftalign{LLRS Codes\footnote{The decoder from~\cite{martinez2019reliable} has the restriction that $\nTransmit=\nReceive$. Therefore, the proposed decoder for $\intOrder=1$ is a more general decoder for \ac{LLRS} codes compared to the decoder in~\cite{martinez2019reliable}.} \\ (unique)} & $\insertions\!+\!\deletions\!<\!\nTransmit\!-\!k\!+\!1$ & --- & $\oh{\nReceive^2}$ & \cite{martinez2019reliable}
\\ \hline
\specialcellsmalldistleftalign{LILRS Codes \\ (list)} & $\insertions\!+\!\intOrder\deletions\!<\!\intOrder(\nTransmit\!-\!k\!+\!1)$ & $|\List| \!\leq\! q^{m(k(\intOrder-1)){}}$ & $\softoh{\intOrder^\omega\!\OMul{\nReceive}}$ & \specialcellsmalldistleftalign{Thm.~\ref{thm:list_dec_LILRS} \\ Sec.~\ref{subsubsec:LILRS_list}}
\\ \hline
\specialcellsmalldistleftalign{LILRS Codes \\ (prob. unique)} & $\insertions\!+\!\intOrder\deletions\leq\intOrder(\nTransmit\!-\!k)$ & $P_f \!\leq\! 3.5^{\ell+1} q^{-m}$ & $\softoh{\intOrder^\omega\!\OMul{\nReceive}}$ & \specialcellsmalldistleftalign{Thms.~\ref{thm:LO_decoder_LILRS} \&~\ref{thm:unique_dec_LILRS} \\ Sec.~\ref{subsec:LO_dec_LILRS} /~\ref{subsec:decodingLILRS}}
\end{tabular}
\end{center}
\end{minipage}
\end{table*}

\subsection{Structure of the paper:}

 In Section~\ref{sec:preliminaries} the notation as well as basic definitions on vector spaces and skew polynomials are introduced.
 Section~\ref{sec:multishot_nwc} is dedicated to a brief introduction to multishot network coding.
 In Section~\ref{sec:LILRS} we consider decoding of \ac{LILRS} codes in the sum-subspace metric for error-control in noncoherent multishot network coding.
 In particular, we derive a {\LOlike} and an interpolation-based decoding scheme for \ac{LILRS} codes, that can correct errors beyond the unique decoding radius in the sum-subspace metric. 
 Section~\ref{sec:conclusion} concludes the paper.

\section{Preliminaries}\label{sec:preliminaries}

We now give some definitions and notation related to multishot network coding and \ac{LLRS} codes.
Since the notation for multishot network coding is quite involved, we tried to stick to common notation as much as possible such that readers familiar with the topic can skip parts of this section.

\subsection{Notation}
A \emph{set} is a collection of distinct elements and is denoted by $\set{S}=\{s_1,s_2,\dots,s_{r}\}$.
The cardinality of $\set{S}$, i.e. the number of elements in $\set{S}$, is denoted by $|\set{S}|$.
By $\intervallincl{i}{j}$ with $i<j$ we denote the set of integers $\{i,i+1,\dots,j\}$.
We denote the set of nonnegative integers by $\ZZ_{\geq0}=\{0,1,2,\dots\}$.

Let $\Fq$ be a finite field of order $q$ and denote by $\Fqm$ the extension field of $\Fq$ of degree $m$ with primitive element $\pe$.
The multiplicative group $\Fqm\setminus\{0\}$ of $\Fqm$ is denoted by $\Fqm^*$.
Matrices and vectors are denoted by bold uppercase and lowercase letters like $\mat{A}$ and $\vec{a}$, respectively, and indexed starting from one.
Under a fixed basis of $\Fqm$ over $\Fq$ any element $a\in\Fqm$ can be represented by a corresponding column vector $\vec{a}\in\Fq^{m\times 1}$. 
For a matrix $\mat{A}\in\Fqm^{M\times N}$ we denote by $\rk_q(\mat{A})$ the rank of the matrix $\mat{A}_q\in\Fq^{Mm\times N}$ obtained by column-wise expanding the elements in $\mat{A}$ over~$\Fq$.
Let $\aut:\Fqm\to\Fqm$ be a finite field automorphism given by $\aut(a)=a^{q^r}$ for all $a\in\Fqm$, where we assume that $1 \leq r \leq m$ and $\gcd(r,m)=1$.
For a matrix $\A$ and a vector $\a$ we use the notation $\aut(\A)$ and $\aut(\a)$ to denote the element-wise application of the automorphism $\aut$, respectively.
For $\mat{A}\in\Fqm^{M\times N}$ we denote by $\Rowspace{\mat{A}}$ the $\Fq$-linear rowspace of the matrix $\mat{A}_q\in\Fq^{M\times Nm}$ obtained by \emph{row-wise} expanding the elements in $\mat{A}$ over $\Fq$.
The \emph{left} and \emph{right} kernel of a matrix $\mat{A}\in\Fqm^{M\times N}$ is denoted by $\lker(\A)$ and $\rker(\A)$, respectively.

For a set $\set{I}\subset\ZZ_{>0}$ we denote by $[\A]_{\set{I}}$ (respectively $[\a]_{\set{I}}$) the matrix (vector) consisting of the columns (entries) of the matrix $\A$ (vector $\a$) indexed by $\set{I}$.

\subsection{Vector Spaces \& Subspace Metrics}

Vector spaces are denoted by calligraphic letters such as e.g. $\myspace{V}$.
By $\ProjspaceAny{N}$ we denote the set of all subspaces of $\Fq^{N}$. 
The \emph{Grassmannian} is the set of all $l$-dimensional subspaces in $\ProjspaceAny{N}$ and is denoted by $\Grassm{N,l}$.
The cardinality of the Grassmannian $|\Grassm{a,b}|$, i.e. the number of $l$-dimensional subspace of $\Fq^N$ is given by the Gaussian binomial $\quadbinom{N}{l}$ which is defined as
\begin{equation*}
\quadbinom{N}{l} = \prod_{i=1}^{l} \frac{q^{N-l+i}-1}{q^i-1}.
\end{equation*}
The Gaussian binomial satisfies (see e.g.~\cite{koetter2008coding})
\begin{equation}
q^{(N-l)l} \leq \textstyle\quadbinom{N}{l} \leq \gammaq q^{(N-l)l}, \label{eq:bounds_gaussian_binomial}
\end{equation}
where $\gammaq := \prod_{i=1}^{\infty} (1-q^{-i})^{-1}$.
Note that $\gammaq$ is monotonically decreasing in $q$ with a limit of $1$, and e.g.~$\kappa_2 \approx 3.463$, $\kappa_3 \approx 1.785$, and $\kappa_4 \approx 1.452$.

The \emph{subspace distance} between two subspaces $\myspace{U},\myspace{V}\in\ProjspaceAny{N_i}$ is defined as (see~\cite{koetter2008coding})
\begin{equation}\label{eq:def_subspace_dist}
  \SubspaceDist(\myspace{U},\myspace{V})\defeq\dim(\myspace{U}+\myspace{V})-\dim(\myspace{U}\cap\myspace{V}).
\end{equation}

For $\Nvec=(N_1,N_2,\dots,N_\shots)$ and $\l=(l_1, l_2, \dots, l_\shots)$ we define the $\shots$-fold Cartesian products
\begin{equation}
 \ProjspaceAny{\Nvec} \defeq \textstyle \prod_{i=1}^{\shots}\ProjspaceAny{N_i}=\ProjspaceAny{N_1}\times\dots\times\ProjspaceAny{N_\shots}
\end{equation}
and
\begin{equation}
 \Grassm{\Nvec, \l} \defeq \textstyle \prod_{i=1}^{\shots}\Grassm{N_i, l_i}=\Grassm{N_1,l_1}\times\dots\times\Grassm{N_\shots, l_\shots}.
\end{equation}
For any tuple of subspaces $\txSpaceVec\in\ProjspaceAny{\Nvec}$ we define its \emph{sum-dimension} as
\begin{equation}
  \sumDim(\txSpaceVec)\defeq\sum_{i=1}^{\shots}\dim(\txSpaceShot{i})
\end{equation}
and call the tuple
\begin{equation*}
    \left(\dim(\txSpaceShot{1}),\dim(\txSpaceShot{2}),\dots,\dim(\txSpaceShot{\shots})\right)
\end{equation*}
the \emph{sum-dimension partition} of $\txSpace$.

We extend fundamental operators on subspaces to tuples of subspaces by considering their application in an element-wise manner, i.e. for $\txSpaceVec,\rxSpaceVec \in \ProjspaceAny{\Nvec, \l}$ we define
\begin{align*}
    \txSpaceVec^\perp &\defeq \left({\txSpaceShot{1}}^\perp, {\txSpaceShot{2}}^\perp, \dots, {\txSpaceShot{\shots}}^\perp\right),
    \\
    \txSpaceVec \cap \rxSpaceVec &\defeq \left(\txSpaceShot{1} \cap \rxSpaceShot{1}, \txSpaceShot{2} \cap \rxSpaceShot{2}, \dots, \txSpaceShot{\shots} \cap \rxSpaceShot{\shots} \right),
    \\
    \txSpaceVec \oplus \rxSpaceVec &\defeq \left(\txSpaceShot{1} \oplus \rxSpaceShot{1}, \txSpaceShot{2} \oplus \rxSpaceShot{2}, \dots, \txSpaceShot{\shots} \oplus \rxSpaceShot{\shots} \right).
\end{align*}

We conclude this subsection considering the extension of the subspace distance to the multishot case~\cite{nobrega2009multishot}.

\begin{definition}[Sum-Subspace Distance~\cite{nobrega2009multishot}]
 Given $\rxSpaceVec=(\rxSpaceShot{1},\rxSpaceShot{2},\dots,\rxSpace^{(\shots)})\in\ProjspaceAny{\Nvec}$ and $\txSpaceVec=(\txSpaceShot{1},\txSpaceShot{2},\dots,\txSpaceShot{\shots})\in\ProjspaceAny{\Nvec}$ the sum-subspace distance between $\rxSpaceVec$ and $\txSpaceVec$ is defined as
 \begin{align}\label{eq:def_sum_subspace_dist}
  \SumSubspaceDist(\rxSpaceVec,\txSpaceVec)\defeq& \sum_{i=1}^{\shots}\SubspaceDist(\rxSpaceShot{i},\txSpaceShot{i}) 
  =\sum_{i=1}^{\shots}\left(\dim(\rxSpaceShot{i}+\txSpaceShot{i})-\dim(\rxSpaceShot{i}\cap\txSpaceShot{i})\right).
 \end{align}
\end{definition}
In~\cite{nobrega2009multishot,nobrega2010multishot} the sum-subspace distance is also called the \emph{extended} subspace distance.
Note, that there exists other metrics for multishot network coding like e.g. the sum-injection distance (see~\cite[Definition~7]{martinez2019reliable}) which are related to the sum-subspace distance and are not considered in this work.

\subsection{Skew Polynomials}\label{subsec:skew_polys}

\emph{Skew polynomials} are a special class of non-commutative polynomials that were introduced by Ore~\cite{ore1933theory}.
A \emph{skew polynomial} is a polynomial of the form 
\begin{equation}
 \textstyle f(x)=\sum_{i}f_i x^i
\end{equation}
with a finite number of coefficients $f_i\in\Fqm$ being nonzero.
The degree $\deg(f)$ of a skew polynomial $f$ is defined as $\max\{i:f_i\neq 0\}$ if $f\neq0$ and $-\infty$ otherwise.

The set of skew polynomials with coefficients in $\Fqm$ together with ordinary polynomial addition and the multiplication rule
\begin{equation}\label{eq:skew_poly_mult_rule}
  xa=\aut(a)x,\qquad a\in\Fqm
\end{equation}
forms a non-commutative ring denoted by $\SkewPolyringZeroDer$.

The set of skew polynomials in $\SkewPolyringZeroDer$ of degree less than $k$ is denoted by $\SkewPolyringZeroDer_{<k}$.
For any $a,b\in\Fqm$ we define the operator
\begin{equation}\label{eq:def_op}
  \op{a}{b} \defeq \aut(b)a.
\end{equation}
For an integer $i\geq0$, we define (see~\cite[Proposition~32]{martinez2018skew})
\begin{equation}\label{eq:def_op_exp}
  \opexp{a}{b}{i+1}
  =\op{a}{\opexp{a}{b}{i}}
\end{equation}
where $\opexp{a}{b}{0}=b$ and $\genNorm{i}{a}=\aut^{i-1}(a)\aut^{i-2}(a)\dots\aut(a)a$ is the generalized power function (see~\cite{lam1988vandermonde}).
The \emph{generalized operator evaluation} of a skew polynomial $f\in\SkewPolyringZeroDer$ at an element $b$ w.r.t. $a$, where $a,b\in\Fqm$, is defined as (see~\cite{leroy1995pseudolinear,martinez2018skew})
\begin{equation}\label{eq:def_gen_op_eval}
  \opev{f}{b}{a}=\sum_{i}f_i\opexp{a}{b}{i}.
\end{equation}
For any fixed evaluation parameter $a \in \Fqm$ the generalized operator evaluation forms an $\Fq$-linear map, i.e. for any $f\in\SkewPolyringZeroDer$, $\beta,\gamma\in\Fq$ and $b,c\in\Fqm$ we have that
\begin{equation}
  \opev{f}{\beta b+\gamma c}{a}=\beta\opev{f}{b}{a}+\gamma\opev{f}{c}{a}.
\end{equation}

For two skew polynomials $f,g\in\SkewPolyringZeroDer$ and elements $a,b\in\Fqm$ the generalized operator evaluation of the product $f\cdot g$ at $b$ w.r.t $a$ is given by (see~\cite{martinez2019private})
\begin{equation}
  \opev{(f\cdot g)}{b}{a} = \opev{f}{\opev{g}{b}{a}}{a}.
\end{equation}

An important notion for the generalized operator evaluation is the concept of conjugacy, which is defined as follows.

\begin{definition}[Conjugacy~\cite{lam1985general}]\label{def:conjugates}
 For any two elements $a\in\Fqm$ and $c\in\Fqm^*$ define
 \begin{equation}
   a^c\defeq \aut(c)ac^{-1}.
 \end{equation}
\begin{itemize}
 \item Two elements $a,b\in\Fqm$ are called $\aut$-conjugates, if there exists an element $c\in\Fqm^*$ such that $b=a^c$.
 \item Two elements that are not $\aut$-conjugates are called $\aut$-distinct.
\end{itemize}
\end{definition}
The notion of $\aut$-conjugacy defines an equivalence relation on $\Fqm$ and thus a partition of $\Fqm$ into conjugacy classes~\cite{lam1988vandermonde}.
The set 
\begin{equation}
  \set{C}(a)\defeq\left\{a^c:c\in\Fqm^*\right\}
\end{equation}
is called \emph{conjugacy class} of $a$.
A finite field $\Fqm$ has at most $\shots\leq q-1$ distinct conjugacy classes.
For $\shots\leq q-1$ the elements $1,\pe,\pe^2,\dots,\pe^{\shots-2}$ are representatives of all (nontrivial) disjoint conjugacy classes of $\Fqm$.

Let $a_1,\dots,a_\shots$ be representatives be from conjugacy classes of $\Fqm$ and let $b_1^{(i)},\dots,b_{n_i}^{(i)}$ be elements from $\Fqm$ for all $i=1,\dots,\shots$.
Then for any nonzero $f\in\SkewPolyringZeroDer$ satisfying
\begin{equation}
  \opev{f}{b_j^{(i)}}{a_i}=0, \qquad \forall i=1,\dots,\shots,j=1,\dots,n_i
\end{equation}
we have that $\deg(f)\leq\sum_{i=1}^{\shots}n_i$ where equality holds if and only if the $b_1^{(i)},\dots,b_{n_i}^{(i)}$ are $\Fq$-linearly independent for each $i=1,\dots,\shots$ (i.e. for each evaluation parameter $a_i$, see~\cite{caruso2019residues}).

The existence of a (generalized operator evaluation) interpolation polynomial is considered in Lemma~\ref{lem:int_poly_op_ev} (see e.g.~\cite{caruso2019residues}).
\begin{lemma}[Lagrange Interpolation (Generalized Operator Evaluation)]\label{lem:int_poly_op_ev}
 Let $b_1^{(i)},\dots,b_{n_i}^{(i)}$ be $\Fq$-linearly independent elements from $\Fqm$ for all $i=1,\dots,\shots$.
 Let $c_1^{(i)},\dots,c_{n_i}^{(i)}$ be elements from $\Fqm$ and let $a_1,\dots,a_\shots$ be representatives for different conjugacy classes of $\Fqm$.
 Define the set of tuples $\set{B} \defeq \{(b_j^{(i)}, c_j^{(i)}, a_i):i=1,\dots,\shots, \, j=1,\dots,n_i\}$.
 Then there exists a unique interpolation polynomial $\IPop{\set{B}}\in\SkewPolyringZeroDer$ such that
 \begin{equation}
  \opev{\IPop{\set{B}}}{b_j^{(i)}}{a_i}=c_j^{(i)},\qquad \forall i=1,\dots,\shots, \, \forall j=1,\dots, n_i,
 \end{equation}
 and $\deg(\IPop{\set{B}})<\sum_{i=1}^{\shots}n_i$.
\end{lemma}

For an element $a\in\Fqm$, a vector $\b\in\Fqm^n$ and a skew polynomial $f\in\SkewPolyringZeroDer$ we define
\begin{equation}
  \opev{f}{\b}{a}\defeq(\opev{f}{b_1}{a},\opev{f}{b_2}{a},\dots,\opev{f}{b_n}{a}).
\end{equation}

The set of all skew polynomials of the form
\begin{equation}\label{eq:def_mult_var_skew_poly}
    Q(x, y_1,\dots, y_\intOrder)=Q_0(x)+Q_1(x)y_1+\dots+Q_\intOrder(x)y_\intOrder,
\end{equation}
where $Q_j\in\SkewPolyringZeroDer$ for all $j=0,\dots,\shots$ is denoted by $\MultSkewPolyringZeroDer$.

\begin{definition}[$\vec{w}$-weighted Degree]
  Given a vector $\vec{w}\in\ZZ_{\geq0}^{\intOrder+1}$, the $\vec{w}$-weighted degree of a multivariate skew polynomial from $Q\in\MultSkewPolyringZeroDer$ is defined as 
  \begin{equation}
    \deg_{\vec{w}}(Q)=\max_j\{\deg(Q_j)+w_j\}.
  \end{equation}
\end{definition}

For an element $a\in\Fqm$ and a vector $\vec{b}=(b_1,b_2,\dots,b_n)\in\Fqm^n$ we define the vector
\begin{equation*}
  \opexp{a}{\b}{j} \defeq \left(\opexp{a}{b_1}{j}, \opexp{a}{b_2}{j}, \dots, \opexp{a}{b_n}{j}\right)
\end{equation*}
and the matrix
\begin{equation}
  \opVandermonde{d}{\vec{b}}{a}
  \defeq
  \begin{pmatrix}
    \b 
    \\ 
    \opexp{a}{\b}{1}
    \\ 
    \opexp{a}{\b}{2}
    \\ 
    \vdots 
    \\
    \opexp{a}{\b}{d-1}
  \end{pmatrix}
  =
  \begin{pmatrix}
   b_1 & b_2 & \dots & b_n
   \\
   \opexp{a}{b_1}{1} & \opexp{a}{b_2}{1} & \dots & \opexp{a}{b_n}{1} 
   \\
   \opexp{a}{b_1}{2} & \opexp{a}{b_2}{2} & \dots & \opexp{a}{b_n}{2} 
   \\
   \vdots & \vdots & \ddots & \vdots
   \\
   \opexp{a}{b_1}{d-1} & \opexp{a}{b_2}{d-1} & \dots & \opexp{a}{b_n}{d-1} 
  \end{pmatrix}
  \in\Fqm^{d\times n}.
\end{equation}

For a vector $\vec{x}=\left(\vec{x}^{(1)} \mid \vec{x}^{(2)} \mid \dots \mid \vec{x}^{(\shots)}\right)\in\Fqm^n$ with $\vec{x}^{(i)}\in\Fqm^{n_i}$ for all $i=1,\dots,\shots$, a length partition $\n=(n_1,n_2,\dots,n_\shots)\in\ZZ_{\geq0}^\shots$ such that $\sum_{i=1}^{\shots}n_i=n$ and a vector $\a = (a_1,a_2,\dots,a_\shots)\in\Fqm^\shots$ we define the vector\footnote{To simplify the notation we omit the length partition $\n$ from the vector operator $\opexp{\a}{\x}{i}$ since it will be always clear from the context (i.e. as the length partition of the vector $\x$).}
\begin{align*}
\opexp{\a}{\x}{i}
:= \begin{pmatrix}
& \opexp{a_1}{\x^{(1)}}{i} & \big\lvert & \opexp{a_2}{\x^{(2)}}{i}  & \big\lvert & \dots & \big\lvert &  \opexp{a_\ell}{\x^{(\ell)}}{i} &
\end{pmatrix} \in \Fqm^n.
\end{align*}
By the properties of the operator $\opexp{a}{\cdot}{i}$, we have that
\begin{equation}\label{eq:rho_i+j_property}
\opexp{\a}{\x}{i+j} = \opexp{\a}{\opexp{\a}{\x}{i}}{j}
\end{equation}
and
\begin{equation}\label{eq:rho_linearity_property}
\opexp{\a}{\myalpha\x}{i} = \aut^i(\myalpha) \opexp{\a}{\x}{i} \quad \forall \, \myalpha \in \Fqm. 
\end{equation}

For a matrix 
\begin{equation*}
  \X=
  \begin{pmatrix}
   \x_1
   \\ 
   \x_2 
   \\ 
   \vdots 
   \\ 
   \x_d
  \end{pmatrix}
  \in\Fqm^{d\times n},
\end{equation*}
and integer $j$ and a vector $\a = (a_1,a_2,\dots,a_\shots)\in\Fqm^\shots$  we define $\opexp{\a}{\cdot}{j}$ applied to $\X$ as
\begin{equation}
  \opexp{\a}{\X}{j}\defeq
  \begin{pmatrix}
   \opexp{\a}{\x_1}{j}
   \\ 
   \opexp{\a}{\x_2}{j}
   \\ 
   \vdots 
   \\ 
   \opexp{\a}{\x_d}{j}
  \end{pmatrix}.
\end{equation}

The element-wise application of the operator to matrices does not affect the rank, i.e. we have that $\rk_{q^m}(\opexp{\a}{\X}{j}) = \rk_{q^m}(\X)$ (see~\cite[Lemma~3]{bartz2022fast}).

\begin{definition}[{\genMoore} Matrix]
  For an integer $d \in \ZZ_{>0}$, a length partition $\n=(n_1,n_2,\dots,n_\shots)\in\ZZ_{\geq0}^\shots$ such that $\sum_{i=1}^{\shots}n_i=n$ and the vectors $\a=(a_1,a_2,\dots,a_\shots)\in\Fqm^\shots$ and $\vec{x}=\left(\vec{x}^{(1)} \mid \vec{x}^{(2)} \mid \dots \mid \vec{x}^{(\shots)}\right)\in\Fqm^n$ with $\vec{x}^{(i)}\in\Fqm^{n_i}$ for all $i=1,\dots,\shots$, the {\genMoore} matrix is defined as
  \begin{align*}
  \Lambda_d(\x)_{\a}
  \defeq
  \begin{pmatrix}
  \opexp{\a}{\x}{0} \\
  \opexp{\a}{\x}{1} \\
  \vdots \\
  \opexp{\a}{\x}{d-1}
  \end{pmatrix}
  =
  \left(
  \begin{array}{c|c|c|c}
  \opVandermonde{d}{\x^{(1)}}{a_1} & \opVandermonde{d}{\x^{(2)}}{a_2} & \cdots & \opVandermonde{d}{\x^{(\shots)}}{a_\shots}
  \end{array}\right) 
  \in \Fqm^{d \times n}.
  \end{align*}
\end{definition}

Similar as for ordinary polynomials and Vandermonde matrices, there is a relation between the generalized operator evaluation and the product with a {\genMoore} matrix. 
In particular, for a skew polynomial $f(x)=\sum_{i=0}^{k-1} f_i x^i \in \SkewPolyringZeroDer_{<k}$ and vectors $\a=(a_1,a_2,\dots,a_\shots)\in\Fqm^\shots$ and $\vec{x}=\left(\vec{x}^{(1)} \mid \vec{x}^{(2)} \mid \dots \mid \vec{x}^{(\shots)}\right)\in\Fqm^n$ we have that
\begin{equation}\label{eq:relation_op_ev_moore}
  \opev{f}{\a}{\x} = (f_0, f_1, \dots, f_{k-1}) \cdot \Lambda_k(\x)_{\a}.
\end{equation}

The rank of a {\genMoore} matrix satisfies $\rk_{q^m}(\Lambda_d(\x)_{\a})=\min\{d,n\}$ if and only if $\SumRankWeight(\x)=n$ (see e.g.~\cite[Theorem~4.5]{lam1988vandermonde}).

\begin{remark}
 To simplify the notation we omit the rank partition $\n$ in $\Lambda_j(\cdot)_\a$ since it will be always clear from the context (i.e. the length partition of the considered vector).
\end{remark}

\section{Multishot Network Coding}\label{sec:multishot_nwc}

As a channel model we consider the \emph{multishot operator channel} from~\cite{nobrega2009multishot} which consists of multiple independent channel uses of the operator channel from~\cite{koetter2008coding}. 
The operator channel is a discrete channel that relates the input $\txSpace \in \ProjspaceAny{N}$ with $\nTransmit \defeq \dim(\txSpace)$ to the output $\rxSpace \in \ProjspaceAny{N}$ by
\begin{equation}\label{eq:def_op_channel}
   \rxSpace = \delOp{\nTransmit-\deletions}(\txSpace) \oplus \errSpace
\end{equation}
where $\delOp{\nTransmit-\deletions}(\txSpace)$ is an erasure operator that returns an $(\nTransmit-\deletions)$-dimensional subspace of $\txSpace$ and $\errSpace \in \Grassm{N,\insertions}$ is a $\insertions$-dimensional subspace with $\txSpace \cap \errSpace = \{\0\}$.
The dimension of the received space $\nReceive \defeq \dim(\rxSpace)$ is then
\begin{equation*}
    \nReceive = \nTransmit - \deletions + \insertions
\end{equation*}
where $\deletions$ is called the number of \emph{deletions} and $\insertions$ is called the number of \emph{insertions}.
Observe, that the subspace distance between the input $\txSpace$ and the output $\rxSpace$ is $\SubspaceDist(\txSpace,\rxSpace)=\insertions+\deletions$.

\subsection{Multishot Operator Channel}

A multishot (or $\shots$-shot) operator channel~\cite{nobrega2009multishot} with overall $\insertions$ insertions and $\deletions$ deletions is a discrete channel with input and output alphabet $\ProjspaceAny{\Nvec}$. 
Consider the partitions of insertions $\insertionsVec=(\insertionsShot{1},\insertionsShot{2},\dots,\insertionsShot{\shots})$ and deletions $\deletionsVec=(\deletionsShot{1},\deletionsShot{2},\dots,\deletionsShot{\shots})$ such that $\insertions = \sum_{i=1}^{\shots} \insertionsShot{i}$ and $\deletions = \sum_{i=1}^{\shots} \deletionsShot{i}$.
The input is a tuple of subspaces $\txSpaceVec\in\ProjspaceAny{\Nvec}$ with sum-dimension $\sumDim(\txSpaceVec) = \nTransmit$ and sum-dimension partition $\nTransmitVec$.
The output $\rxSpaceVec\in\ProjspaceAny{\Nvec}$ is related to the input $\txSpaceVec\in\ProjspaceAny{\Nvec}$ by
\begin{equation}\label{eq:def:multishot_op_channel}
  \rxSpaceVec=\delOp{\nTransmit-\deletions}(\txSpaceVec)\oplus\errSpaceVec,
\end{equation}
where $\delOp{\nTransmit-\deletions}(\txSpaceVec)$ returns a tuple with sum-dimension $\nTransmit-\deletions$ and sum-dimension partition $\nTransmitVec-\deletionsVec$ and $\errSpaceVec \in \Grassm{\Nvec,\insertionsVec}$ is a tuple of error spaces with $\sumDim(\errSpaceVec)=\insertions$ and $\txSpaceVec \cap \errSpaceVec=\{\0\}$. 
The multishot operator channel can be considered as $\shots$ instances of the (single-shot) operator channel defined in~\eqref{eq:def_op_channel} such that an overall number of $\insertions$ insertions and $\deletions$ deletions occurs.

The output $\rxSpaceVec$ of the multishot operator channel has sum-dimension
\begin{equation*}
  \nReceive=\nTransmit-\deletions+\insertions
\end{equation*}
with sum-dimension partition
\begin{equation}\label{eq:part_rec_space}
    \nReceiveVec = \nTransmitVec - \deletionsVec + \insertionsVec.
\end{equation}
The multishot operator channel is illustrated in Figure~\ref{fig:ms_op_channel}.
\begin{figure}[ht!]
  \centering
  \begin{tikzpicture}[auto]
	\node (a) at (0,0) {};
	
	\node [draw, inner sep=15pt,thick] (channel) at (150pt,0) {$\myspace{H}_{\nTransmit-\deletions}(\txSpaceVec)\oplus\errSpaceVec$};
	
	\node [coordinate,thick] (end) at (295pt,0){};
	
	\path[->,thick] (a) edge node[midway, above, xshift=-5pt] {$\txSpaceVec = (\txSpaceShot{1}\!, \dots, \txSpaceShot{\shots})$} node[midway, below, xshift=-5pt] {$\in \ProjspaceAny{\Nvec}$} (channel);
	
	\path[->,thick] (channel) edge node[midway, above, xshift=+5pt] {$\rxSpaceVec = (\rxSpaceShot{1}\!, \dots, \rxSpaceShot{\shots})$} node[midway, below, xshift=+5pt] {$\in \ProjspaceAny{\Nvec}$} (end);

\end{tikzpicture}
  \caption{Illustration of the $\shots$-shot operator channel.}
  \label{fig:ms_op_channel}
\end{figure}
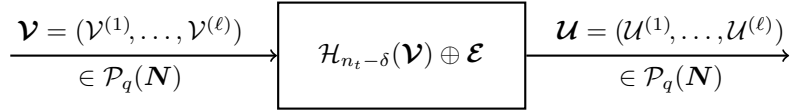

In~\cite{nobrega2010multishot,martinez2019reliable} a sum-rank-metric representation of the multishot operator channel in the spirit of~\cite{silva2008rank,Silva_PhD_ErrorControlNetworkCoding} was considered.
This equivalent channel representation is more suitable for decoding of \ac{LILRS} codes in the sum-rank metric. 
In this work, we consider the interpretation as a multishot operator channel that is more closely related to the sum-subspace metric.

\begin{remark}\label{rem:rand_instance_op_channel}
  By the term ``random instance of the $\shots$-shot operator channel with overall $\insertions$ insertions and $\deletions$ deletions'' we mean, that we draw uniformly at random an instance from all instances of the $\shots$-shot operator channel~\eqref{eq:def:multishot_op_channel}, i.e. we dram uniformly at random form all partitions of the insertions $\insertionsVec$ and deletions $\deletionsVec$, and for fixed $\txSpaceVec$, $\insertionsVec$ and $\deletionsVec$, the error space is chosen uniformly at random from the set
  \begin{equation}\label{eq:def_err_space_set}
    \Eset^{(\insertionsVec)}\defeq\left\{\errSpaceVec\in\Grassm{\Nvec, \insertionsVec}:\txSpaceShot{i}\cap\errSpaceShot{i}=\{\vec{0}\},\forall i \wedge \sumDim(\errSpaceVec)=\insertions\right\}.
  \end{equation}
\end{remark}

In Appendix~\ref{app:draw_uniform_errors} we propose an efficient procedure to implement random instances of the multishot operator channel for parameters $\NbarVec=(\Nbar,\dots,\Nbar)$ and $\nTransmitBarVec=(\nTransmitBar,\dots,\nTransmitBar)$ by adapting the dynamic-programming routine in~\cite[Appendix~A]{puchinger2020generic} for drawing an error of given sum-rank weight uniformly at random to the sum-subspace case (see Algorithm~\ref{alg:multi_shot_op_channel}).

We now extend the definition of $(\insertions,\deletions)$ reachability for the operator channel~\cite{bartz2020fast} to the multishot operator channel.

\begin{definition}[$(\insertions,\deletions)$ Reachability]\label{def:ins_del_reachability}
 Given two tuples of subspaces $\rxSpaceVec,\txSpaceVec \in \ProjspaceAny{\Nvec}$ we say that $\txSpaceVec$ is $(\insertions,\deletions)$-reachable from $\rxSpaceVec$ if there exists a realization of the multishot operator channel~\eqref{eq:def:multishot_op_channel} with $\insertions$ insertions and $\deletions$ deletions that transforms the input $\txSpaceVec$ to the output $\rxSpaceVec$.
\end{definition}

Next, we relate the $(\insertions,\deletions)$-reachability with the sum-subspace distance.

\begin{proposition}
 Consider $\rxSpaceVec \in \ProjspaceAny{\Nvec}$ and $\txSpaceVec \in \ProjspaceAny{\Nvec}$.
 If $\txSpaceVec$ is $(\insertions,\deletions)$-reachable from $\rxSpaceVec$, then we have that $\SumSubspaceDist(\rxSpaceVec,\txSpaceVec)=\insertions+\deletions$.
\end{proposition}

Similar as in~\cite{koetter2008coding} we now define normalized parameters for codes in the sum-subspace metric.
The normalized weight $\normweight$, the code rate $R$ and the normalized minimum distance $\normdist$ of a sum-subspace code $\mycode{C}$ with parameters $N=\sum_{i=1}^{\shots}N_i$ and $\nTransmit=\sum_{i=1}^{\shots}\nTransmitShot{i}$ is defined as
\begin{equation}\label{eq:normalized_parameters}
 \normweight\defeq\sum_{i=1}^{\shots}\frac{\nTransmitShot{i}}{N}=\frac{\nTransmit}{N},
 R\defeq\frac{\log_q(\mycode{C})}{\sum_{i=1}^{\shots}\nTransmitShot{i}N_i}
 \quad\text{and}\quad
 \normdist\defeq\frac{\SumSubspaceDist(\mycode{C})}{2\nTransmit}=\frac{\SumSubspaceDist(\mycode{C})}{2\normweight N},
\end{equation}
respectively.
The normalized parameters $\normweight$, $R$ and $\normdist$ defined in~\eqref{eq:normalized_parameters} lie naturally within the interval $[0,1]$.
Define $\nTransmitBar\defeq\nTransmit/\shots$. 
For $\nTransmitShot{i}=\nTransmitBar$ for all $i=1,\dots,\shots$ we can write the code rate as
\begin{equation*}
  R=\frac{\log_q(|\mycode{C}|)}{\sum_{i=1}^{\shots}\nTransmitBar N_i}
  =\frac{\shots\log_q(|\mycode{C}|)}{\nTransmit N}
  =\frac{\shots\log_q(|\mycode{C}|)}{\normweight N^2}.
\end{equation*}

A \emph{sum-subspace} code $\mycode{C}$ is a non-empty subset of $\ProjspaceAny{\Nvec}$, and has minimum subspace distance $\SumSubspaceDist(\mycode{C})$ when all subspaces in the code have distance at least $\SumSubspaceDist(\mycode{C})$ and there is at least one pair of subspaces with distance exactly $\SumSubspaceDist(\mycode{C})$.
In the following we consider \emph{constant-shot-dimension} codes\footnote{In~\cite{martinez2019reliable} these codes are called \emph{sum-constant-dimension} codes.}, i.e. codes that inject the same number of (linearly independent) packets $\nTransmitShot{i}$ in a given shot.
In this setup, we transmit a tuple of subspaces
\begin{equation}
  \txSpaceVec=\left(\txSpaceShot{1},\txSpaceShot{2},\dots,\txSpaceShot{\shots}\right)
  \in \Grassm{\Nvec, \nTransmitVec}
\end{equation}
and receive a tuple of subspaces
\begin{equation}
 \rxSpaceVec=\left(\rxSpaceShot{1},\rxSpaceShot{2},\dots,\rxSpaceShot{\shots}\right)
  \in \Grassm{\Nvec, \nReceiveVec},
\end{equation}
where
\begin{equation}\label{eq:comp_rec_sapces}
  \rxSpaceShot{i} =  
  \Rowspace{
  \begin{pmatrix}
   \vecxi^{(i)\top} & \u_1^{(i)\top} & \u_2^{(i)\top} & \dots & \u_\intOrder^{(i)\top}
  \end{pmatrix}
  } 
  =
  \RowspaceHuge{
  \begin{pmatrix}
   \xi_1^{(i)} & u_{1,1}^{(i)} & u_{1,2}^{(i)} & \dots & u_{1,\intOrder}^{(i)}
   \\[-2pt] 
   \vdots & \vdots & \vdots & \ddots & \vdots
   \\ 
   \xi_{\nReceiveShot{i}}^{(i)} & u_{\nReceiveShot{i},1}^{(i)} & u_{\nReceiveShot{i},2}^{(i)} & \dots & u_{\nReceiveShot{i},\intOrder}^{(i)}
  \end{pmatrix}
  }
\end{equation}
for all $i=1,\dots,\shots$.

Similar as for subspace codes in~\cite{koetter2008coding} we now define \emph{complementary} sum-subspace codes.
For any $\myspace{V}\in\ProjspaceAny{N}$ the \emph{dual} space $\myspaceDual{V}$ is defined as
\begin{equation*}
  \myspaceDual{V}\defeq\{\u\in\Fq^{N}:\u\v^\top=0,\,\forall \v\in\myspace{V}\}
\end{equation*}
where $\dim(\myspaceDual{V})=N-\dim(\myspace{V})$.
For a tuple $\txSpaceVec=(\txSpaceShot{1},\txSpaceShot{2},\dots,\txSpaceShot{\shots})\in\ProjspaceAny{\Nvec}$ we define the dual tuple as
\begin{equation*}
  \txSpaceVec^\perp\defeq\left((\txSpaceShot{1})^\perp,(\txSpaceShot{2})^\perp,\dots,(\txSpaceShot{\shots})^\perp\right)\in\ProjspaceAny{\Nvec}.
\end{equation*}
Note, that if $\txSpaceVec \in \Grassm{\Nvec,\nTransmitVec}$, then we have that $\txSpaceVec^\perp \in \Grassm{\Nvec,\Nvec-\nTransmitVec}$.
By applying~\cite[Equation~4]{koetter2008coding} to the subspace distance between each component space of two tuples $\rxSpaceVec,\txSpaceVec\in\ProjspaceAny{\Nvec}$ in~\eqref{eq:def_sum_subspace_dist} we get that 
\begin{equation}
  \SumSubspaceDist(\txSpaceVec^\perp,\rxSpaceVec^\perp)=\SumSubspaceDist(\txSpaceVec,\rxSpaceVec).
\end{equation}

Consider a constant-shot-dimension sum-subspace code $\mycode{C} \subseteq \Grassm{\Nvec,\nTransmitVec}$.
Then the complementary constant-shot-dimension sum-subspace code $\mycode{C}^\perp$ is defined as
\begin{equation}
  \mycode{C}^\perp\defeq\{\txSpaceVec^\perp:\txSpaceVec\in\mycode{C}\} 
  \subseteq \Grassm{\Nvec,\Nvec-\nTransmitVec}.
\end{equation}
The complementary code $\mycode{C}^\perp$ has cardinality $|\mycode{C}^\perp|=|\mycode{C}|$, minimum sum-subspace distance $\SumSubspaceDist(\mycode{C}^\perp)=\SumSubspaceDist(\mycode{C})$ and code rate
\begin{equation*}
  R^\perp=\frac{\log_q(|\mycode{C}^\perp|)}{\sum_{i=1}^{\shots}(N_i-\nTransmitShot{i})N_i}.
\end{equation*}

\subsection{Lifted Linearized Reed--Solomon Codes}

\emph{Lifted} \ac{LRS} codes~\cite{martinez2019reliable} are constant-shot-dimension multishot network codes for error-control in noncoherent multishot network coding.
The main idea behind the construction of \ac{LLRS} codes is to lift codewords of an \ac{LRS} code in a block-wise manner by augmenting each (transposed) codeword block by the corresponding $\Fq$-linearly independent code locators.
For the special case of $\shots=1$ the construction coincides with the Kötter--Kschischang subspace codes~\cite{koetter2008coding}.

Let $\a=(a_1,a_2,\dots,a_\shots)$ be a vector containing representatives from different conjugacy classes of $\Fqm$.
Let the vectors $\vecbeta^{(i)}=(\beta_1^{(i)},\beta_2^{(i)},\dots,\beta_{\nTransmitShot{i}}^{(i)})\in\Fqm^{\nTransmitShot{i}}$ contain $\Fq$-linearly independent elements from $\Fqm$ for all $i=1,\dots,\shots$ and define $\vecbeta=\left(\vecbeta^{(1)}\mid\vecbeta^{(2)}\mid\dots\mid\vecbeta^{(\shots)}\right)\in\Fqm^{\nTransmit}$ and $\nTransmitVec=(\nTransmitShot{1},\nTransmitShot{2},\dots,\nTransmitShot{\shots})$.
Then an \ac{LLRS} code $\liftedLinRS{\vecbeta,\a,\shots;\nTransmitVec,k}$ of sum-subspace dimension $\nTransmit=\nTransmitShot{1}+\nTransmitShot{2}+\dots+\nTransmitShot{\shots}$, sum-dimension partition $\nTransmitVec$ and dimension $k\leq\nTransmit$ is defined as

\begin{align*}
\Big\{\txSpaceVec(f) \defeq \left(\txSpaceShot{1}(f), \dots, \txSpaceShot{\shots}(f)\right) \, : \, f \in \SkewPolyringZeroDer_{<k}\Big\}
\subseteq \Grassm{\Nvec,\nTransmitVec}
\end{align*}
where $\Nvec=(N_1,\dots,N_\shots)$ with $N_i=\nTransmitShot{i}+m$, and, for $f \in \SkewPolyringZeroDer_{<k}$, we have
\begin{equation*}
\txSpaceShot{i}(f) \defeq
  \Rowspace{
  \begin{pmatrix}
   \vecbeta^{(i)\top} & \opev{f}{\vecbeta^{(i)}}{a_i}^\top
  \end{pmatrix}
  }\
  \in \Grassm{N_i,\nTransmitShot{i}},
  \qquad \forall i=1,\dots,\shots.
\end{equation*}
The \emph{lifting} operation corresponds to augmenting each transposed codeword block of an \ac{LRS} codeword by the corresponding (transposed) code locators and considering the $\Fq$-linear rowspace thereof (see~\cite{silva2008rank,nobrega2010multishot,martinez2019reliable}). 
The lifting operation causes a rate-loss since the code locators do not carry information since they are common for all codewords.

The minimum sum-subspace distance of $\liftedLinRS{\vecbeta,\a,\shots;\nTransmitVec,k}$ equals (see~\cite{martinez2019reliable}) 
\begin{equation*}
     \SumSubspaceDist(\liftedLinRS{\vecbeta,\a,\shots;\nTransmitVec,k}) = 2(\nTransmit - k + 1)
\end{equation*} 
and the code rate is
\begin{equation*}
    R = \frac{mk}{\sum_{i=1}^{\shots}\nTransmitShot{i}(\nTransmitShot{i} + m)}.
\end{equation*}

In~\cite{martinez2019reliable} and efficient interpolation-based decoding algorithm that can correct an overall number of $\insertions$ insertions and $\deletions$ deletions up to
\begin{equation}
    \insertions + \deletions < \nTransmit - k + 1
\end{equation}
was presented. 
However, the decoder from~\cite{martinez2019reliable} has the restriction that the dimension of the received spaces and the dimension of the transmitted spaces must be the same (c.f.~\cite[Section~V.H]{martinez2019reliable}).

\section{Decoding of Lifted Interleaved LRS Codes for Error-Control in Multishot Network Coding}\label{sec:LILRS}

In this section, we consider the application of \emph{lifted} \ac{ILRS} codes for error control in multishot network coding. In particular, we focus on noncoherent transmissions, where the network topology and/or the coefficients of the in-network linear combinations at the intermediate nodes are not known (or used) at the transmitter and the receiver.
Therefore, we define and analyze \acf{LILRS} codes. 
We derive a {\LOlike} decoder~\cite{loidreau2006decoding,overbeck2007public,overbeck2008structural} for \ac{LILRS} codes which is capable of correcting insertions and deletions beyond the unique decoding region at the cost of a (very) small decoding failure probability.
Although the {\LOlike} decoder is not the most efficient decoder in terms of computational complexity, it allows to analyze the decoding failure probability and gives insights about the decoding procedure.
We derive a tight upper bound on the decoding failure probability of the {\LOlike} decoder for \ac{LILRS} codes, which, unlike simple heuristic bounds, considers the distribution of the error spaces caused by insertions.

We propose an efficient interpolation-based decoding scheme, which can correct insertions and deletions beyond the unique decoding region and which be used as a list decoder or as a probabilistic unique decoder.
We derive upper bounds on the worst-case list size and use the relation between the interpolation-based decoder and the {\LOlike} decoder to derive an upper bound on the decoding failure probability for the interpolation-based probabilistic unique decoding approach.
Unlike the interpolation-based decoder in~\cite[Section~V.H]{martinez2019reliable}, the proposed decoding schemes \emph{do not} have the restriction that the dimension of the received spaces and the dimension of the transmitted spaces must be the same.

\subsection{Lifted Interleaved Linearized Reed--Solomon Codes}\label{subsec:LILRS}

In this section we consider \ac{LILRS} codes for transmission over a multishot operator channel~\eqref{eq:def:multishot_op_channel}. 
We generalize the ideas from~\cite{martinez2019reliable} to obtain multishot subspace codes by \emph{lifting} the \ac{ILRS} codes defined in~\cite{bartz2022fast}.

\begin{definition}[Lifted Interleaved Linearized Reed--Solomon Code]\label{def:LILRS}
Let $\a=(a_1,a_2,\dots,a_\shots)$ be a vector containing representatives from different conjugacy classes of $\Fqm$.
Let the vectors $\vecbeta^{(i)}=(\beta_1^{(i)},\beta_2^{(i)},\dots,\beta_{\nTransmitShot{i}}^{(i)})\in\Fqm^{\nTransmitShot{i}}$ contain $\Fq$-linearly independent elements from $\Fqm$ for all $i=1,\dots,\shots$ and define $\vecbeta=\left(\vecbeta^{(1)}\mid\vecbeta^{(2)}\mid\dots\mid\vecbeta^{(\shots)}\right)\in\Fqm^{\nTransmit}$ and $\nTransmitVec=(\nTransmitShot{1},\nTransmitShot{2},\dots,\nTransmitShot{\shots})$.
A lifted $\intOrder$-interleaved linearized Reed--Solomon (LILRS) code $\liftedIntLinRS{\vecbeta,\a,\shots,\intOrder;\nTransmitVec,k}$ of sum-subspace dimension $\nTransmit=\nTransmitShot{1}+\nTransmitShot{2}+\dots+\nTransmitShot{\shots}$ and dimension $k\leq\nTransmit$ is defined as

\begin{align*}
\Big\{\txSpaceVec(\f) \defeq \left(\txSpaceShot{1}(\f), \dots, \txSpaceShot{\shots}(\f)\right) \, : \, \f \in \SkewPolyringZeroDer_{<k}^\intOrder\Big\}
\subseteq \Grassm{\Nvec,\nTransmitVec}
\end{align*}
where $\Nvec=(N_1,\dots,N_\shots)$ with $N_i=\nTransmitShot{i}+\intOrder m$, and, for $\f = (f_1, \dots, f_\intOrder)$, we have
\begin{equation*}
\txSpaceShot{i}(\f) \defeq
   \Rowspace{
  \begin{pmatrix}
   \vecbeta^{(i)\top} & \opev{f_1}{\vecbeta^{(i)}}{a_i}^\top & \dots & \opev{f_\intOrder}{\vecbeta^{(i)}}{a_i}^\top
  \end{pmatrix}
  }
  \in \Grassm{N_i,\nTransmitShot{i}}.
\end{equation*}
\end{definition}

Observe, that compared to \ac{LLRS} codes the relative overhead due to lifting decreases in $\intOrder$ since the evaluations are performed at the same code locators and thus have to be appended only once.
The reduction of the relative overhead comes at the cost of an increased packet size $N_i$ for each shot $i=1,\dots,\shots$.

The definition of~\ac{LILRS} codes generalizes several code families.  
For $\intOrder=1$ we obtain the lifted linearized Reed--Solomon codes from~\cite[Section~V.III]{martinez2019reliable}. 
For $\shots=1$ we obtain lifted interleaved Gabidulin codes as considered in e.g.~\cite{wachter2014list, bartz2018efficient} with Kötter--Kschischang codes~\cite{koetter2008coding} as special case for $\intOrder=1$.

Proposition~\ref{prop:min_dist_lilrs} shows that interleaving does not increase the minimum sum-subspace distance of the code.

\begin{proposition}[Minimum Distance]\label{prop:min_dist_lilrs}
 The minimum sum-subspace distance of a \ac{LILRS} code $\liftedIntLinRS{\vecbeta,\a,\shots,\intOrder;\nTransmitVec,k}$ as in Definition~\ref{def:LILRS} is
 \begin{equation}
  \SumSubspaceDist\left(\liftedIntLinRS{\vecbeta,\a,\shots,\intOrder;\nTransmitVec,k}\right)=2\left(\nTransmit-k+1\right).
 \end{equation}
\end{proposition}

\begin{proof}
    Consider an \ac{LLRS} code $\liftedLinRS{\vecbeta,\a,\shots;\nTransmitVec,k}$ with minimum subspace distance $2(\nTransmit-\deletions+1)$.
    Let $\txSpace_1,\txSpace_2 \in \liftedLinRS{\vecbeta,\a,\shots;\nTransmitVec,k}$ be two subspaces having the minimum distance to each other, i.e. we have that $\SubspaceDist(\txSpace_1,\txSpace_2) = 2(\nTransmit-\deletions+1)$.
    Define the subspaces 
    \begin{equation}\label{eq:def_neutral_space}
        \shot{\myspace{Z}}{i}
        =\Rowspace{
            \begin{pmatrix}
                {\shot{\vecbeta}{i}}^\top & \0^\top
            \end{pmatrix}
        },
        \qquad \forall i=1,\dots,\shots
    \end{equation}
    and consider without loss of generality the tuples
    \begin{equation*}
        \txSpaceVec_1 = \left(\txSpace_1, \shot{\myspace{Z}}{2}, \dots, \shot{\myspace{Z}}{\shots}\right)
        \qquad \text{and} \qquad
        \txSpaceVec_2 = \left(\txSpace_2, \shot{\myspace{Z}}{2}, \dots, \shot{\myspace{Z}}{\shots}\right).
    \end{equation*}
    By definition we have that $\txSpaceVec_1,\txSpaceVec_2 \in \liftedIntLinRS{\vecbeta,\a,\shots,\intOrder;\nTransmitVec,k}$.
    The component spaces $\shot{\myspace{Z}}{2}, \dots, \shot{\myspace{Z}}{\shots}$ do not contribute to the sum-subspace distance and thus we have that $\SumSubspaceDist(\txSpaceVec_1,\txSpaceVec_2) = \SubspaceDist(\txSpace_1,\txSpace_2)=2(\nTransmit-k+1)$.
    Hence, we have that $\SumSubspaceDist(\liftedIntLinRS{\vecbeta,\a,\shots,\intOrder;\nTransmitVec,k}) \leq 2(\nTransmit-k+1)$.
    Now suppose that $\SumSubspaceDist(\liftedIntLinRS{\vecbeta,\a,\shots,\intOrder;\nTransmitVec,k}) < 2(\nTransmit-k+1)$.
    Then there must exists two tuples $\txSpaceVec_1,\txSpaceVec_2 \in \liftedIntLinRS{\vecbeta,\a,\shots,\intOrder;\nTransmitVec,k}$ that consist of $(\shots-1)$ component spaces of the form~\eqref{eq:def_neutral_space} and two spaces $\txSpace_1,\txSpace_2 \in \liftedLinRS{\vecbeta,\a,\shots;\nTransmitVec,k}$ (in the same shot position), which contracts that $\liftedLinRS{\vecbeta,\a,\shots;\nTransmitVec,k}$ has minimum subspace distance $2(\nTransmit-k+1)$.
    Therefore, we conclude that $\SumSubspaceDist\left(\liftedIntLinRS{\vecbeta,\a,\shots,\intOrder;\nTransmitVec,k}\right)=2\left(\nTransmit-k+1\right)$.
\end{proof}

For a LILRS code $\mycode{C}=\liftedIntLinRS{\vecbeta,\a,\shots,\intOrder;\nTransmitVec,k}$ we have that $N_i=\nTransmitShot{i}+\intOrder m$ for all $i=1,\dots,\shots$ and therefore the code rate is
\begin{equation}\label{def:code_rate_LILRS}
 R
 =\frac{\log_q(|\mycode{C}|)}{\sum_{i=1}^{\shots}\nTransmitShot{i}N_i}
 =\frac{\intOrder mk}{\sum_{i=1}^{\shots}\nTransmitShot{i}(\nTransmitShot{i}+\intOrder m)}.
\end{equation}
Note, that there exist other definitions of the code rate for multishot codes that are not considered in this paper (see e.g.~\cite[Section~IV.A]{nobrega2009multishot}).

The normalized weight $\normweight$ and the normalized distance $\normdist$ of an \ac{LILRS} code $\liftedIntLinRS{\vecbeta,\a,\shots,\intOrder;\nTransmitVec,k}$ is
\begin{equation}
    \normweight=\frac{\nTransmit}{N}
    \quad\text{and}\quad
    \normdist=\frac{\nTransmit-k+1}{\nTransmit}.
\end{equation} 
For $\nTransmitShot{i}=\nTransmitBar$ for all $i=1,\dots,\shots$ the code rate of an \ac{LILRS} code $\liftedIntLinRS{\vecbeta,\a,\shots,\intOrder;\nTransmitVec,k}$ in~\eqref{def:code_rate_LILRS} becomes
\begin{equation*}
  R=\frac{\shots\intOrder m k}{\nTransmit N}.
\end{equation*}

For $\nTransmitShot{i}=\nTransmitBar$ for all $i=1,\dots,\shots$ the Singleton-like upper bound for constant-shot-dimension sum-subspace codes~\cite[Theorem~7]{martinez2019reliable} evaluated for the parameters of \ac{LILRS} codes becomes
\begin{equation*}
  |\mycode{C}|\leq \gammaq^\shots q^{\intOrder m\left(\nTransmit-\left(\SumSubspaceDist(\mycode{C})/2-1\right)\right)}
  =\gammaq^\shots q^{\intOrder mk}
\end{equation*}
which shows, that a Singleton-like bound achieving code can be at most $\gammaq^\shots<3.5^\shots$ times larger than the corresponding \ac{LILRS} code.
Equivalently, the code rate is therefore upper-bounded by
\begin{equation*}
  R\leq\frac{\shots(\log_q(\gammaq)+\intOrder mk)}{\nTransmit N}.
\end{equation*}

The benefit of the decreased relative overhead due to interleaving is illustrated in Figure~\ref{fig:norm_dist_rate}.
The figure shows, that the rate loss due to the overhead introduced by the lifting is reduced significantly, even for small interleaving orders.
Further, we see that \ac{LILRS} codes approach the Singleton-like bound for sum-subspace codes (see~\cite{martinez2019reliable}) with increasing interleaving order while preserving the extension field degree $m$\footnote{The Singleton-like depends on the sum-dimension $N$ of the ambient space and thus changes in $\intOrder$.}.
\begin{figure}[ht!]
  \centering
  \definecolor{mycolor1}{rgb}{1.00000,0.00000,1.00000}%
\definecolor{mycolor2}{rgb}{0.00000,1.00000,1.00000}%

\newcommand{\setthr}[2]{\draw[thin,black] (axis cs:#1,1e-12) -- (axis cs:#1,3e-12);\node at (axis cs: #1,5e-12) {\tiny \textcolor{black}{#2}};}

\begin{tikzpicture}

\begin{axis}[%
xmin=0,
xmax=1,
xlabel={Code rate $R$},
compat=newest,
xmajorgrids,
ymin=0,
ymax=1.2,
yminorticks=false,
label style={anchor=near ticklabel, font=\footnotesize},
label style={inner sep=0}, 
ylabel={Normalized distance $\eta$},
ymajorgrids,
yminorgrids,
tick label style={font=\scriptsize},
legend style={at={(.98,.98)},anchor=north east,legend cell align=left,align=left,draw=white!15!black, font=\footnotesize},
mystyle]

\addplot [color=blue!80!black,
solid,
]
table[row sep=crcr]{%
0.004122 	 1.062500 \\
0.035372 	 1.000000 \\
0.066622 	 0.937500 \\
0.097872 	 0.875000 \\
0.129122 	 0.812500 \\
0.160372 	 0.750000 \\
0.191622 	 0.687500 \\
0.222872 	 0.625000 \\
0.254122 	 0.562500 \\
0.285372 	 0.500000 \\
0.316622 	 0.437500 \\
0.347872 	 0.375000 \\
0.379122 	 0.312500 \\
0.410372 	 0.250000 \\
0.441622 	 0.187500 \\
0.472872 	 0.125000 \\
0.504122 	 0.062500 \\
};
\addlegendentry{Singleton bound $\intOrder = 1$};

\addplot [color=green!80!black,
solid,
]
table[row sep=crcr]{%
0.002061 	 1.062500 \\
0.048936 	 1.000000 \\
0.095811 	 0.937500 \\
0.142686 	 0.875000 \\
0.189561 	 0.812500 \\
0.236436 	 0.750000 \\
0.283311 	 0.687500 \\
0.330186 	 0.625000 \\
0.377061 	 0.562500 \\
0.423936 	 0.500000 \\
0.470811 	 0.437500 \\
0.517686 	 0.375000 \\
0.564561 	 0.312500 \\
0.611436 	 0.250000 \\
0.658311 	 0.187500 \\
0.705186 	 0.125000 \\
0.752061 	 0.062500 \\
};
\addlegendentry{Singleton bound $\intOrder = 3$};

\addplot [color=orange!80!black,
solid,
]
table[row sep=crcr]{%
0.000749 	 1.062500 \\
0.057568 	 1.000000 \\
0.114386 	 0.937500 \\
0.171204 	 0.875000 \\
0.228022 	 0.812500 \\
0.284840 	 0.750000 \\
0.341658 	 0.687500 \\
0.398477 	 0.625000 \\
0.455295 	 0.562500 \\
0.512113 	 0.500000 \\
0.568931 	 0.437500 \\
0.625749 	 0.375000 \\
0.682568 	 0.312500 \\
0.739386 	 0.250000 \\
0.796204 	 0.187500 \\
0.853022 	 0.125000 \\
0.909840 	 0.062500 \\
};
\addlegendentry{Singleton bound $\intOrder = 10$};


\addplot [color=blue!80!black,
dashed,
mark=diamond,
mark options={solid}
]
table[row sep=crcr]{%
0.000000 	 1.062500 \\
0.031250 	 1.000000 \\
0.062500 	 0.937500 \\
0.093750 	 0.875000 \\
0.125000 	 0.812500 \\
0.156250 	 0.750000 \\
0.187500 	 0.687500 \\
0.218750 	 0.625000 \\
0.250000 	 0.562500 \\
0.281250 	 0.500000 \\
0.312500 	 0.437500 \\
0.343750 	 0.375000 \\
0.375000 	 0.312500 \\
0.406250 	 0.250000 \\
0.437500 	 0.187500 \\
0.468750 	 0.125000 \\
0.500000 	 0.062500 \\
};
\addlegendentry{\ac{LLRS} ($\intOrder = 1$)~\cite{martinez2019reliable}};

\addplot [color=green!80!black,
dashed,
mark=asterisk,
mark options={solid}
]
table[row sep=crcr]{%
0.000000 	 1.062500 \\
0.046875 	 1.000000 \\
0.093750 	 0.937500 \\
0.140625 	 0.875000 \\
0.187500 	 0.812500 \\
0.234375 	 0.750000 \\
0.281250 	 0.687500 \\
0.328125 	 0.625000 \\
0.375000 	 0.562500 \\
0.421875 	 0.500000 \\
0.468750 	 0.437500 \\
0.515625 	 0.375000 \\
0.562500 	 0.312500 \\
0.609375 	 0.250000 \\
0.656250 	 0.187500 \\
0.703125 	 0.125000 \\
0.750000 	 0.062500 \\
};
\addlegendentry{\ac{LILRS} $\intOrder = 3$};

\addplot [color=orange!80!black,
dashed,
mark=x,
mark options={solid}
]
table[row sep=crcr]{%
0.000000 	 1.062500 \\
0.056818 	 1.000000 \\
0.113636 	 0.937500 \\
0.170455 	 0.875000 \\
0.227273 	 0.812500 \\
0.284091 	 0.750000 \\
0.340909 	 0.687500 \\
0.397727 	 0.625000 \\
0.454545 	 0.562500 \\
0.511364 	 0.500000 \\
0.568182 	 0.437500 \\
0.625000 	 0.375000 \\
0.681818 	 0.312500 \\
0.738636 	 0.250000 \\
0.795455 	 0.187500 \\
0.852273 	 0.125000 \\
0.909091 	 0.062500 \\
};
\addlegendentry{\ac{LILRS} $\intOrder = 10$};

\node (lambda_1) at (0.41,0.12) {$\avgnormweight$=0.5};
\node (lambda_2) at (0.63,0.12) {$\avgnormweight$=0.25};
\node (lambda_3) at (0.95,0.12) {$\avgnormweight$=0.09};

\end{axis}
\end{tikzpicture}%

  \caption{Normalized distance $\normdist$ over the code rate $R$ for an \ac{LILRS} code $\liftedIntLinRS{\vecbeta,\a,\shots=2,\intOrder;\nTransmitVec=(8,8),k=4}$ over $\F_{3^8}$ for interleaving orders $\intOrder\in\{1,3,10\}$ with the corresponding (average) normalized weight $\avgnormweight$. The case $\intOrder=1$ corresponds to the \ac{LLRS} codes from~\cite{martinez2019reliable}.}
  \label{fig:norm_dist_rate}
\end{figure}

\subsection{{\LOlike} Decoder for LILRS Codes}\label{subsec:LO_dec_LILRS}

Loidreau and Overbeck proposed the first efficient decoder for interleaved Gabidulin codes in the rank metric~\cite{loidreau2006decoding,overbeck2007public,overbeck2008structural}.
The main idea behind the Loidreau--Overbeck decoder is to compute an $\Fq$-linear transformation matrix from a decoding matrix (which depends on the code and the received word) that allows to transform the received word into a \emph{corrupted} part and a \emph{noncorrupted} part.
The noncorrupted part is then used to recover the message polynomials e.g. via Lagrange interpolation.

The concept of the Loidreau--Overbeck decoder was generalized to decoding \ac{ILRS} codes in the sum-rank metric~\cite{bartz2022fast}.
In the sum-rank-metric case an $\Fq$-linear transformation matrix is obtained for each block.

Based on the previous decoders for the rank and sum-rank metric, we derive a {\LOlike} decoder for \ac{LILRS} codes.
Similar to the original decoder and its sum-rank-metric analogue we set up a decoding matrix that allows to compute $\Fq$-linear transformation matrices for each shot $\rxSpaceShot{i}$. 
The obtained transformation matrices allow to compute particular bases for the received subspaces $\rxSpaceShot{i}$ that can be split into a basis for the \emph{corrupted} part (corresponding to the error space $\errSpaceShot{i}$) and a \emph{noncurrupted} part (i.e. a basis for $\txSpaceShot{i} \cap \rxSpaceShot{i}$) for each shot.  
The noncorrupted part is then used to reconstruct the message polynomials via Lagrange interpolation.
The qualitative structure of the tuple $\hat{\U}=(\hat{\U}^{(1)},\hat{\U}^{(2)},\dots,\hat{\U}^{(\shots)})$ containing transformed basis matrices is illustrated in Figure~\ref{fig:transformed_rxSpaces}.

\begin{figure}[ht!]
  \centering
  \pgfkeys{tikz/mymatrixenv/.style={decoration={brace},every left delimiter/.style={xshift=8pt},every right delimiter/.style={xshift=-8pt}}}

\pgfkeys{tikz/mymatrix/.style={matrix of math nodes,nodes in empty cells,left delimiter={(},right delimiter={)},inner sep=1pt,outer sep=1.5pt,column sep=8pt,row sep=8pt,nodes={minimum width=20pt,minimum height=10pt,anchor=center,inner sep=0pt,outer sep=0pt}}}

\pgfkeys{tikz/myarray/.style={matrix of math nodes,nodes in empty cells,inner sep=3pt,outer sep=3pt,column sep=8pt,row sep=8pt,nodes={minimum width=20pt,minimum height=10pt,anchor=center,inner sep=0pt,outer sep=0pt}}}

\pgfkeys{tikz/myarray_left/.style={matrix of math nodes,nodes in empty cells,left delimiter={(},inner sep=1pt,outer sep=2pt,column sep=8pt,row sep=8pt,nodes={minimum width=20pt,minimum height=10pt,anchor=center,inner sep=0pt,outer sep=0pt}}}

\pgfkeys{tikz/myarray_mid/.style={matrix of math nodes,nodes in empty cells,left delimiter={|},right delimiter={|},inner sep=3pt,outer sep=3pt,column sep=8pt,row sep=8pt,nodes={minimum width=20pt,minimum height=10pt,anchor=center,inner sep=0pt,outer sep=0pt}}}

\pgfkeys{tikz/myarray_right/.style={matrix of math nodes,nodes in empty cells,right delimiter={)},inner sep=3pt,outer sep=3pt,column sep=8pt,row sep=8pt,nodes={minimum width=20pt,minimum height=10pt,anchor=center,inner sep=0pt,outer sep=0pt}}}

\pgfkeys{tikz/mymatrixbrace/.style={decorate,thick}}

\tikzset{
	style green/.style={
    set fill color=green!50!black!60,draw opacity=0.4,
    set border color=green!50!black!60,fill opacity=0.1,
  },
  style cyan/.style={
    set fill color=cyan!90!blue!60, draw opacity=0.4,
    set border color=blue!70!cyan!30,fill opacity=0.1,
  },
  style orange/.style={
    set fill color=orange!90, draw opacity=0.8,
    set border color=orange!90, fill opacity=0.3,
  },
  style brown/.style={
    set fill color=brown!70!orange!40, draw opacity=0.4,
    set border color=brown, fill opacity=0.3,
  },
  style purple/.style={
    set fill color=violet!90!pink!20, draw opacity=0.5,
    set border color=violet, fill opacity=0.3,    
  },
  style red/.style={
    set fill color=red!90!pink!20, draw opacity=0.5,
    set border color=red, fill opacity=0.3,    
  },
  kwad/.style={
    above left offset={-0.3, 0.28},
    below right offset={0.3, -0.28},
    #1
  },
  pion/.style={
    above left offset={-0.3, 0.4},
    below right offset={0.3, -0.01},
    #1
  },
  set fill color/.code={\pgfkeysalso{fill=#1}},
  set border color/.style={draw=#1}
}
\[
	\hat{\U}=
    \scalebox{0.85}{
    \begin{tikzpicture}[baseline={-0.5ex}, mymatrixenv]
    \matrix [myarray_left,inner sep=8pt, ampersand replacement=\&] (RX1) {
	    \tikzmarkin[kwad=style green]{corr1}\phantom{1} \& \phantom{1} \& \phantom{1} \& \phantom{1}
      \\
      \phantom{1} \& \phantom{1} \& \phantom{1} \& \phantom{1}\tikzmarkend{corr1}
      \\ 
      \tikzmarkin[kwad=style red]{noncorr1}\phantom{1} \& \phantom{1} \& \phantom{1} \& \phantom{1}\tikzmarkend{noncorr1}
	    \\    
    };

    \matrix [myarray, inner sep=8pt, right=of RX1, xshift=-35pt, ampersand replacement=\&] (RX2) {
      \tikzmarkin[kwad=style green]{corr2}\phantom{1} \& \phantom{1} \& \phantom{1} \& \phantom{1}\tikzmarkend{corr2}
      \\
      \tikzmarkin[kwad=style red]{noncorr2}\phantom{1} \& \phantom{1} \& \phantom{1} \& \phantom{1}\tikzmarkend{noncorr2}
      \\ 
    };

    \matrix [myarray_right, inner sep=8pt, right=of RX2, ampersand replacement=\&] (RX3) {
      \tikzmarkin[kwad=style green]{corr3}\phantom{1} \& \phantom{1} \& \phantom{1} \& \phantom{1}
      \\
      \phantom{1} \& \phantom{1} \& \phantom{1} \& \phantom{1}\tikzmarkend{corr3}
      \\ 
      \tikzmarkin[kwad=style red]{noncorr3}\phantom{1} \& \phantom{1} \& \phantom{1} \& \phantom{1}\tikzmarkend{noncorr3}
      \\    
    };

    \draw ([xshift=-5pt,yshift=5pt]RX1.south east) node {,};
    \draw ([xshift=-35pt,yshift=5pt]RX3.south west) node (comma) {,};
    \draw node[right=of comma, xshift=1pt] {,};
    \draw node[right=of RX2, xshift=-25pt] {$\dots$};

    \end{tikzpicture}
    }
\]
  \caption{Qualitative illustration of the structure of the tuple $\hat{\U}$ containing transformed basis matrices. The green parts form a basis for the non-corrupted spaces whereas the red parts indicate a basis for the erroneous spaces.
  The green part is used to reconstruct the message polynomials.
  }
  \label{fig:transformed_rxSpaces}
\end{figure} 

The main motivation to derive the {\LOlike} decoder is to obtain an upper bound on the decoding failure probability that incorporates the distribution of the error spaces in $\errSpaceVec$.
In Section~\ref{subsec:decodingLILRS} we will reduce the interpolation-based decoder for \ac{LILRS} codes to the {\LOlike} decoder in order to obtain an upper bound on the decoding failure probability of the interpolation-based probabilistic unique decoder.

Up to our knowledge, this is the first {\LOlike} decoding scheme in the (sum-) subspace metric. 
It includes lifted interleaved Gabidulin (or interleaved Kötter--Kschischang) codes~\cite{bartz2018efficient,bartz2017algebraic} as special case for $\shots=1$.
Hence, the results give a strict upper bound\footnote{In~\cite{bartz2018efficient,bartz2017algebraic} a heuristic upper bound on the decoding failure probability, which does \emph{not} incorporate the distribution of the error spaces, was derived.} on the decoding failure probability of the decoders in~\cite{bartz2018efficient,bartz2017algebraic}.

Suppose we transmit the tuple of subspaces 
\begin{equation}
 \txSpaceVec(\f) = \left(\txSpaceShot{1}(\f), \dots, \txSpaceShot{\shots}(\f)\right)\in
 \liftedIntLinRS{\vecbeta,\a,\shots,\intOrder;\nTransmitVec,k}
\end{equation}
over an $\shots$-shot operator channel with overall $\insertions$ insertions and $\deletions$ deletions and receive the tuple of subspaces
\begin{equation}
 \rxSpaceVec = \left(\rxSpaceShot{1}, \dots, \rxSpaceShot{\shots}\right) \in \Grassm{\Nvec, \nReceiveVec},
\end{equation}
where the received subspaces $\rxSpaceShot{i}$ are as defined in~\eqref{eq:comp_rec_sapces} and $\nReceiveVec = \nTransmitVec + \insertionsVec - \deletionsVec$ (see~\eqref{eq:part_rec_space}). 
Define the vectors
\begin{equation}\label{eq:def_sum_sub_rec_vecs}
  \vecxi=\left(\vecxi^{(1)} \mid \vecxi^{(2)} \mid\dots\mid \vecxi^{(\shots)}\right)\in\Fqm^{\nReceive}
  \quad\text{and}\quad
  \u_j=\left(\u_j^{(1)} \mid \u_j^{(2)} \mid\dots\mid \u_j^{(\shots)}\right)\in\Fqm^{\nReceive},
\end{equation}
for all $j=1,\dots,\intOrder$ and consider the matrix
\begin{equation}\label{eq:LO_matrix_LILRS}
\LODecMat := 
\begin{pmatrix}
\Lambda_{\nTransmit-\deletions-1}(\vecxi)_{\a} \\
\Lambda_{\nTransmit-\deletions-k}\left(\u_1\right)_{\a} \\
\vdots \\
\Lambda_{\nTransmit-\deletions-k}\left(\u_\intOrder\right)_{\a} \\
\end{pmatrix}
\in \Fqm^{((\intOrder+1)(\nTransmit-\deletions)-\intOrder k -1) \times \nReceive}.
\end{equation}

\begin{lemma}[Transformed of Decoding Matrix]\label{lem:properties_LODecMat_LILRS}
  Consider the transmission of a tuple of subspaces $\txSpaceVec(\f) \in \liftedIntLinRS{\vecbeta,\a,\shots,\intOrder;\nTransmitVec,k}$ over an $\shots$-shot operator channel with overall $\insertions$ insertions and $\deletions$ deletions and receive the tuple of subspaces $\rxSpaceVec$.
  Let $\LODecMat$ be as in~\eqref{eq:LO_matrix_LILRS}.
  Then there exist invertible matrices $\W^{(i)}\in\Fq^{\nReceiveShot{i}\times\nReceiveShot{i}}$ such that for $\W=\diag(\W^{(1)},\W^{(2)},\dots,\W^{(\shots)})$ we have that
   \begin{equation}\label{eq:def_L_bar}
     \bar{\LODecMat}=\left(\bar{\LODecMat}^{(1)},\dots,\bar{\LODecMat}^{(\shots)}\right)
     =\LODecMat\cdot\W
    \in \Fqm^{((\intOrder+1)(\nTransmit-\deletions)-\intOrder k -1) \times \nReceive}
   \end{equation}
   consists of component matrices of the form
   \begin{equation*}
     \bar{\LODecMat}^{(i)}=
     \begin{pmatrix}
      \opVandermonde{\nTransmit-\deletions-1}{\bar{\vecxi}_1^{(i)}}{a_i} & \opVandermonde{\nTransmit-\deletions-1}{\bar{\vecxi}_2^{(i)}}{a_i} & \0
      \\ 
      \0 & \opVandermonde{\nTransmit-\deletions-k}{\widetilde{\e}_1^{(i)}}{a_i} & \opVandermonde{\nTransmit-\deletions-k}{\hat{\e}_1^{(i)}}{a_i}
      \\
      \vdots & \vdots & \vdots
      \\ 
      \0 & \opVandermonde{\nTransmit-\deletions-k}{\widetilde{\e}_\intOrder^{(i)}}{a_i} & \opVandermonde{\nTransmit-\deletions-k}{\hat{\e}_\intOrder^{(i)}}{a_i}
     \end{pmatrix}
     \in \Fqm^{((\intOrder+1)(\nTransmit-\deletions)-\intOrder k -1) \times \nReceiveShot{i}}
   \end{equation*}
   where $\bar{\vecxi}_1^{(i)}\in\Fqm^{\nTransmitShot{i}-\deletionsShot{i}}$, $\bar{\vecxi}_2^{(i)},\widetilde{\e}_l^{(i)}\in\Fqm^{\rankErrShot{i}},\hat{\e}_l^{(i)}\in\Fqm^{\deviationsShot{i}}$ and
   \begin{alignat*}{3}
    \rk_{q}(\bar{\vecxi}_1^{(i)})&=\nTransmitShot{i}-\deletionsShot{i}, \qquad &&\rk_{q}(\bar{\vecxi}_2^{(i)})&&=\rankErrShot{i}
    \\ 
    \rk_{q}(\widetilde{\e}_l^{(i)})&=\rankErrShot{i}, &&\rk_{q}(\hat{\e}_l^{(i)})&&=\deviationsShot{i}
    \\ 
    \rk_{q}((\widetilde{\e}_l^{(i)}\mid\hat{\e}_l^{(i)}))&=\rankErrShot{i}+\deviationsShot{i}, &&
   \end{alignat*}
   for all $i=1,\dots,\shots$ such that $\rk_{q^m}(\LODecMat)=\rk_{q^m}(\bar{\LODecMat})$.
\end{lemma}

The proof of Lemma~\ref{lem:properties_LODecMat_LILRS} is based on particular bases for the received spaces in $\rxSpaceVec$ and properties of the intersection and error spaces in $\txSpaceVec \cap \rxSpaceVec$ and $\errSpaceVec$, respectively, and can be found in Appendix~\ref{app:proof_properties_LODecMat_LILRS}.

\begin{lemma}[Properties of Decoding Matrix]\label{lem:right_kernel_prop_LILRS}
  Consider the notation and definitions as in Lemma~\ref{lem:properties_LODecMat_LILRS} and define the vectors
\begin{equation*}
  \bar{\e}_l\defeq\left((\widetilde{\e}_{l}^{(1)}\mid\hat{\e}_{l}^{(1)})\mid\dots\mid(\widetilde{\e}_{l}^{(\shots)}\mid\hat{\e}_{l}^{(\shots)})\right)\in\Fqm^{\insertions},
  \quad\forall l=1,\dots\intOrder,
\end{equation*}
and the matrix 
\begin{equation}
 \bar{\Z}=\left(\bar{\Z}^{(1)}\mid\dots\mid\bar{\Z}^{(\shots)}\right)\defeq
   \left(
   \begin{array}{c}
    \Lambda_{\nTransmit-\deletions-k}(\bar{\e}_1)_{\vec{a}}    
    \\ 
    \vdots 
    \\ 
    \Lambda_{\nTransmit-\deletions-k}(\bar{\e}_\intOrder)_{\vec{a}}
   \end{array}
   \right)\in\Fqm^{\intOrder(\nTransmit-\deletions-k)\times\insertions}.
   \label{eq:def_Z}
\end{equation}
Let $\h = (\h^{(1)} \mid \h^{(2)} \mid \dots \mid \h^{(\shots)}) \in \Fqm^{\nReceive}$ with $\h^{(i)} \in \Fqm^{\nReceiveShot{i}}$ for all $i=1,\dots,\shots$ be a nonzero vector in the right kernel of the decoding matrix $\LODecMat$ and suppose that $\bar{\Z}$ has $\Fqm$-rank $\insertions$. Then:
\begin{enumerate}
  \item \label{itm:LO_LILRS_lemma_rank_L} We have $\rk_{q^m}(\LODecMat) = \nReceive-1$.
  
  \item \label{itm:LO_LILRS_lemma_rank_of_kernel_element} We have $\rk_{q}(\h^{(i)})=\nTransmitShot{i}-\deletionsShot{i}$ for all $i=1,\dots,\ell$, i.e., $\h$ has sum-rank weight $\SumRankWeight^{(\nReceiveVec)}(\h) = \nTransmit-\deletions$. 
  
  \item \label{itm:LO_LILRS_lemma_Ti_matrices} There are invertible matrices $\T^{(i)} \in \Fq^{\nReceiveShot{i} \times \nReceiveShot{i}}$, for all $i=1,\dots,\ell$, such that the last (rightmost) $\insertionsShot{i}$ positions of $\h^{(i)} \T^{(i)}$ are zero.
  
  \item \label{itm:LO_LILRS_lemma_EDi_zero} The first (upper) $\nTransmitShot{i}-\deletionsShot{i}$ rows of $\hat{\U}^{(i)} = \left(\T^{(i)}\right)^{-1}\U^{(i)}$ form a basis for the non-corrupted received space $\rxSpaceShot{i} \cap \txSpaceShot{i}$ for all $i=1,\dots,\shots$.
  
  \item \label{itm:LO_LILRS_reconstruct_polynomials} The $l$-th message polynomial $f_l$ can be uniquely reconstructed from the transformed basis $\hat{\U}^{(i)}$ for the received space $\rxSpaceShot{i}$ by Lagrange interpolation on the first $\nTransmitShot{i} - \deletionsShot{i}$ rows of $\hat{\U}^{(i)}$ for all $l=1,\dots,\intOrder$ and $i=1,\dots,\shots$.  
\end{enumerate}
\end{lemma}

We now provide a sketch of the proof.
The full proof of Lemma~\ref{lem:right_kernel_prop_LILRS} can be found in Appendix~\ref{app:proof_right_kernel_prop_LILRS}.

\begin{proofsketch}
    \begin{itemize}
      \item[--] Ad~\ref{itm:LO_LILRS_lemma_rank_L}):
        The matrix $\bar{\LODecMat}$ can be rearranged into an upper block-triangular matrix whose rank is determined by the two blocks on the diagonal, which have rank $\nTransmit-\deletions-1$ and $\insertions$ and thus imply that the $\Fqm$-rank of the whole matrix equals $\nTransmit-\deletions-1+\insertions = \nReceive-1$. 
        The statement follows since by Lemma~\ref{lem:properties_LODecMat_LILRS} we have that $\rk_{q^m}(\LODecMat)=\rk_{q^m}(\bar{\LODecMat})$.

      \item[--] Ad~\ref{itm:LO_LILRS_lemma_rank_of_kernel_element}):
        By assumption the $\Fqm$-rank of $\bar{\Z}$ equals $\insertions$ which implies that $\rk_{q^m}(\shot{\bar{\Z}}{i})=\insertionsShot{i}$ for all $i=1,\dots,\shots$.
        Thus, for any $\bar{\h} \in \rker(\bar{\LODecMat}) \setminus \{\0\}$ the $\insertionsShot{i}$ rightmost entries of $\shot{\bar{\h}}{i}$ must be zero which implies that $\rk_q(\shot{\bar{\h}}{i}) \leq \nTransmitShot{i}-\deletionsShot{i}$ for all $i=1,\dots,\shots$.
        On the other hand $\bar{\h}$ is contained in a code with minimum sum-rank distance $\nTransmit-\deletions$ which is the dual of the code spanned by the first $\nTransmit-\deletions-1$ rows of $\bar{\LODecMat}$.
        The statement follows by combining these two facts.

      \item[--] Ad~\ref{itm:LO_LILRS_lemma_Ti_matrices}): 
        By~\ref{itm:LO_LILRS_lemma_rank_of_kernel_element}) the $\Fq$-rank of $\shot{\h}{i} \in \Fqm^{\nReceiveShot{i}}$ equals $\nTransmitShot{i}-\deletionsShot{i}$ for all $i=1,\dots,\shots$.
        Hence there exist matrices $\shot{\T}{i} \in \Fq^{\nReceiveShot{i} \times \nReceiveShot{i}}$ such that the $\nReceiveShot{i}-(\nTransmitShot{i}-\deletionsShot{i})=\insertionsShot{i}$ rightmost entries of $\shot{\h}{i} \shot{\T}{i}$ are equal to zero.

      \item[--] Ad~\ref{itm:LO_LILRS_lemma_EDi_zero}): 
         Define the matrices $\D^{(i)}=\left(\T^{(i)-1}\right)^\top$ for all $i=1,\dots,\shots$ and observe that $\h\T\in\rker(\LODecMat\cdot\diag(\D^{(1)},\dots,\D^{(\shots)}))$.
         By using the $\Fqm$-rank condition on $\bar{\Z}$ one can show that the span of the $\insertionsShot{i}$ rightmost columns of the matrices $\shot{\bar{\LODecMat}}{i}$ and $\shot{\LODecMat}{i}\shot{\D}{i}$ coincides.
         These columns correspond to the insertions which in turn implies that the last $\insertionsShot{i}$ rows of $\shot{\hat{\U}}{i}=(\shot{\D}{i})^\top \shot{\U}{i}$ form a basis for $\errSpaceShot{i}$.
         The statement follows since by the definition of the operator channel we have that $\txSpaceShot{i} \cap \errSpaceShot{i} = \{\0\}$ for all $i=1,\dots,\shots$.

      \item[--] Ad~\ref{itm:LO_LILRS_reconstruct_polynomials}):
        By~\ref{itm:LO_LILRS_lemma_EDi_zero}) the first $\nTransmitShot{i}-\deletionsShot{i}$ rows of the transformed basis $\shot{\hat{\U}}{i}$ form a basis for the noncorrupted intersection space $\txSpaceShot{i} \cap \rxSpaceShot{i}$ for all $i=1,\dots,\shots$.
        Due to the $\Fq$-linearity of the generalized operator evaluation for a fixed evaluation parameter (i.e. per shot), the message polynomials can be reconstructed by constructing the corresponding Lagrange interpolation polynomials (see Figure~\ref{fig:transformed_rxSpace}).
    \end{itemize}
\end{proofsketch}

The complete procedure for the {\LOlike} decoder for \ac{LILRS} codes is given in Algorithm~\ref{alg:LO_LILRS}.
The structure of the transformed basis matrices $\hat{\U}^{(i)}$ for all $i=1,\dots,\shots$ is illustrated in Figure~\ref{fig:transformed_rxSpace}.

\begin{algorithm}[ht]
    \caption{
    \algoname{{\LOlike} Decoder for \ac{LILRS} Codes}
    }
  \label{alg:LO_LILRS}

  \begin{algorithmic}[1]
  \Statex \textbf{Input:} 
   A tuple containing the basis matrices $\U=(\shot{\U}{1},\shot{\U}{2}, \dots, \shot{\U}{\shots})\in\prod_{i=1}^{\shots}\Fqm^{\nReceiveShot{i}\times(\intOrder+1)}$ for the output $\rxSpaceVec=(\rxSpaceShot{1},\rxSpaceShot{2},\dots,\rxSpaceShot{\shots}) \in \Grassm{\Nvec,\nReceiveVec}$ of an $\shots$-shot operator channel with overall $\insertions$ insertions and $\deletions$ deletions for input $\txSpaceVec(\f)\in\liftedIntLinRS{\vecbeta,\a,\shots,\intOrder;\nTransmitVec,k}$
  
  \Statex

  \Statex \textbf{Output:} 
  Message polynomial vector $\f=(f_1,\dots,f_\intOrder)\in\SkewPolyringZeroDer_{<k}^\intOrder$ or ``decoding failure''
  
  \Statex
  \State Set up the matrix $\LODecMat$ as in~\eqref{eq:LO_matrix_LILRS}
  \State Compute right kernel $\myspace{H}=\rker(\LODecMat)$
  \If{$\dim(\myspace{H})>1$} \label{line:LILRS_check_if_kernel_dimension>1}
    \State \Return ``decoding failure''
  \Else
     \State Compute an element $\h=\left(\vec{h}^{(1)} \mid \dots \mid \vec{h}^{(\shots)}\right)\in\myspace{H}\setminus\{\vec{0}\}$ \label{line:LO_LILRS_h}
  %
    \For{$i=1,\dots,\shots$}
        \State Compute $\nTransmitShot{i}-\deletionsShot{i}\gets\rk_{q}\left(\h^{(i)}\right)$ \label{line:LO_LILRS_rk_hi}
        \State {Compute full-rank matrix $\T^{(i)}\in\Fq^{\nReceiveShot{i}\times\nReceiveShot{i}}$ s.t. \newline \hspace*{.96cm} the rightmost $\insertionsShot{i}=\nReceiveShot{i}-(\nTransmitShot{i}-\deletionsShot{i})$ entries of $\vec{h}^{(i)}\T^{(i)}$ are zero.} \label{line:LO_LILRS_Ti}
        \State $\hat{\U}^{(i)}\gets\left(\T^{(i)}\right)^{-1}\U^{(i)}
      =
       \begin{pmatrix}
          \hat{\vecxi}^{(i)\top} & \hat{\u}_1^{(i)\top} & \dots & \hat{\u}_\intOrder^{(i)\top} 
       \end{pmatrix}
       \in\Fqm^{\nReceiveShot{i}\times (\intOrder+1)}$ \label{line:LO_LILRS_R_trans}
    \EndFor
    \For{$l=1,\dots,\intOrder$}
      \State{$f_l\gets\IPop{\set{B}_l}$ where $\set{B}_l = \left\{(\hat{\xi}^{(i)}_\mu,\hat{u}^{(i)}_{l,\mu}, a_i): i=1,\dots,\shots,\mu=1,\dots,\nTransmitShot{i}-\deletionsShot{i}\right\}$} \label{line:LO_LILRS_IP}
    \EndFor
    \State \Return $\f=(f_1,\dots,f_\intOrder)$
  \EndIf
\end{algorithmic}
\end{algorithm}

\begin{figure}[ht!]
  \centering
  \pgfkeys{tikz/mymatrixenv/.style={decoration={brace},every left delimiter/.style={xshift=2pt},every right delimiter/.style={xshift=-0pt}}}

\pgfkeys{tikz/mymatrix/.style={matrix of math nodes,nodes in empty cells,left delimiter={(},right delimiter={)},inner sep=1pt,outer sep=1.5pt,column sep=8pt,row sep=8pt,nodes={minimum height=10pt,anchor=center,inner sep=0pt,outer sep=0pt}}}

\pgfkeys{tikz/mymatrixbrace/.style={decorate,thick}}

\tikzset{
	style green/.style={
    set fill color=green!50!black!60,draw opacity=0.4,
    set border color=green!50!black!60,fill opacity=0.1,
  },
  style cyan/.style={
    set fill color=cyan!90!blue!60, draw opacity=0.4,
    set border color=blue!70!cyan!30,fill opacity=0.1,
  },
  style orange/.style={
    set fill color=orange!90, draw opacity=0.8,
    set border color=orange!90, fill opacity=0.3,
  },
  style brown/.style={
    set fill color=brown!70!orange!40, draw opacity=0.4,
    set border color=brown, fill opacity=0.3,
  },
  style purple/.style={
    set fill color=violet!90!pink!20, draw opacity=0.5,
    set border color=violet, fill opacity=0.3,    
  },
  style red/.style={
    set fill color=red!90!pink!20, draw opacity=0.5,
    set border color=red, fill opacity=0.3,    
  },
  kwad/.style={
    above left offset={-0.15,0.35},
    below right offset={0.52,-0.5},
    #1
  },
  pion/.style={
    above left offset={-0.7,0.5},
    below right offset={0.22,-0.5},
    #1
  },
  set fill color/.code={\pgfkeysalso{fill=#1}},
  set border color/.style={draw=#1}
}
\[
	\hat{\U}^{(i)}\!=\!\big(\T^{(i)}\big)^{-1}\U^{(i)}\!=\!
  \scalebox{0.9}{
    \begin{tikzpicture}[baseline={-0.5ex}, mymatrixenv]
    \matrix [mymatrix, inner sep=5pt, ampersand replacement=\&, column 3/.append style={minimum width=10pt}, column 4/.append style={ minimum width=60pt}] (m) {
      \tikzmarkin[pion=style green]{noncorr} \hat{\xi}^{(i)}_1  \&  \opev{f_1}{\hat{\xi}^{(i)}_1}{a_i} \&  \dots  \& \opev{f_\intOrder}{\hat{\xi}^{(i)}_1}{a_i} 
      \\[-5pt]
      \vdots  \& \vdots \& \ddots \& \vdots   
      \\[-5pt]
      \hat{\xi}^{(i)}_{\nTransmitShot{i}\!-\deletionsShot{i}}  \& \opev{f_1}{\hat{\xi}^{(i)}_{\nTransmitShot{i}\!-\deletionsShot{i}}}{a_i}  \& \dots \& \opev{f_\intOrder}{\hat{\xi}^{(i)}_{\nTransmitShot{i}\!-\deletionsShot{i}}}{a_i}  \tikzmarkend{noncorr}   
      \\
      \tikzmarkin[kwad=style red]{corr}\hat{\xi}^{(i)}_{\nTransmitShot{i}\!-\deletionsShot{i}+1}   \&\hat{u}_{1,\nTransmitShot{i}\!-\deletionsShot{i}+1}^{(i)} \& \dots \& \hat{u}_{\intOrder,\nTransmitShot{i}\!-\deletionsShot{i}+1}^{(i)} 
      \\[-5pt]
      \vdots  \& \vdots \& \ddots \& \vdots  
      \\[-5pt]
      \hat{\xi}^{(i)}_{\nReceiveShot{i}}   \&\hat{u}_{1,\nReceiveShot{i}}^{(i)} \& \dots \& \phantom{12} \hat{u}_{\intOrder,\nReceiveShot{i}}^{(i)}\hspace*{7.5pt} \tikzmarkend{corr}   
      \\    
    };
    \draw[<->, thick, color=gray]  ([xshift=20pt, yshift=5pt]m-1-4.north east) -- ([xshift=20pt,yshift=-5pt]m-3-4.south east) node[midway,xshift=25pt] {$\nTransmitShot{i}\!\!-\deletionsShot{i}$};
    \draw[<->, thick, color=gray]  ([xshift=20pt, yshift=2pt]m-4-4.north east) -- ([xshift=20pt,yshift=-5pt]m-6-4.south east) node[midway,xshift=25pt] {$\insertionsShot{i}$};
    \end{tikzpicture}
    }
\]
  \caption{Illustration of the structure of the transformed basis matrices $\hat{\U}^{(i)}$. The green part forms a basis for the non-corrupted space $\rxSpaceShot{i}\cap\txSpaceShot{i}$ whereas the red part forms a basis for the error space $\errSpaceShot{i}$.
  }
  \label{fig:transformed_rxSpace}
\end{figure}

\begin{lemma}[Decoding Failure Probability]\label{lem:LO_LILRS_probability_full_Fqm_rank}
Suppose that a tuple of subspaces
\begin{equation}
 \txSpaceVec(\f) = \left(\txSpaceShot{1}(\f), \dots, \txSpaceShot{\shots}(\f)\right)\in
 \liftedIntLinRS{\vecbeta,\a,\shots,\intOrder;\nTransmitVec,k}
\end{equation}
is transmitted over a random instance of the $\shots$-shot operator channel (see Remark~\ref{rem:rand_instance_op_channel}) with overall $\insertions$ insertions and $\deletions$ deletions, where $\insertions$ and $\deletions$ satisfy $\insertions \leq \insertionsmax := \intOrder(\nTransmit-\deletions-k)$.
Let $\bar{\Z}$ be defined as in~\eqref{eq:def_Z} (see Lemma~\ref{lem:right_kernel_prop_LILRS}).
Then, we have
\begin{align*}
\Pr\!\left(\rk_{q^m}(\Lambda_{\nTransmit-\deletions-k}(\bar{\Z})_{\a})<\insertions\right) \leq \gammaq^{\ell+1} q^{-m(\insertionsmax-\insertions+1)}.
\end{align*}
\end{lemma}

\begin{proof}
 Let $\insertionsVec=(\insertionsShot{1},\dots,\insertionsShot{\shots})\in\ZZ_{\geq0}^{\shots}$ and $\deletionsVec=(\deletionsShot{1},\dots,\deletionsShot{\shots})\in\ZZ_{\geq0}^{\shots}$ be the partition of the insertions and deletions of the $\shots$-shot operator channel, respectively. 
 By assumption, the tuple of error spaces $\errSpaceVec \in \Grassm{\Nvec, \insertionsVec}$ is chosen uniformly at random from the set $\Eset^{(\insertionsVec)}$ as defined in~\eqref{eq:def_err_space_set}.
 Each instance of the $\shots$-shot operator channel yields a decoding matrix $\bar{\LODecMat}$ of the form~\eqref{eq:def_L_bar} for some $\t=(t_1,\dots,t_\shots)\in\ZZ_{\geq0}^{\shots}$, $\deviationsVec=(\deviations_1,\dots,\deviations_\shots)\in\ZZ_{\geq0}^{\shots}$.
 Now let us fix $\insertionsVec$ and $\deletionsVec$, which yields a particular instance of $\t$ and $\deviationsVec$.
 By Lemma~\ref{lem:right_kernel_prop_LILRS} the decoder succeeds if the $\Fqm$-rank of the matrix 
 \begin{equation*}
 \bar{\Z}\defeq
   \left(
   \begin{array}{c}
    \Lambda_{\nTransmit-\deletions-k}(\bar{\e}_1)_{\vec{a}}    
    \\ 
    \vdots 
    \\ 
    \Lambda_{\nTransmit-\deletions-k}(\bar{\e}_\intOrder)_{\vec{a}}
   \end{array}
   \right)\in\Fqm^{\intOrder(\nTransmit-\deletions-k)\times\insertions},
 \end{equation*}
 equals $\insertions$ given that
 \begin{equation}\label{eq:sum_rank_condition}
  \SumRankWeight\left(\shot{\widetilde{\E}}{1}\shot{\hat{\E}}{1}\mid\dots\mid \shot{\widetilde{\E}}{\shots}\shot{\hat{\E}}{\shots}\right)=t+\deviations=\insertions.
 \end{equation}
 Drawing $\errSpaceVec$ uniformly at random with a fixed rank partition $\insertionsVec$ from $\Eset^{(\insertionsVec)}$ corresponds to drawing $\widetilde{\E}$ and $\hat{\E}$ uniformly at random from all matrices that satisfy~\eqref{eq:sum_rank_condition}.
 This implies that the probability 
 \begin{equation*}
  \Pr(\rk_{q^m}(\Lambda_{\nTransmit-\deletions-1}(\bar{\Z})_\a)<\insertions)
 \end{equation*}
 depends on $\insertionsVec$ and \emph{not} on $\t$ and $\deviationsVec$.
 Hence, we can use~\cite[Lemma~7]{bartz2022fast} with $t=\insertions$, $\t=\insertionsVec$, $n=\nReceive$ and get 

\begin{equation*}
\Pr(\rk_{q^m}(\Lambda_{\nTransmit-\deletions-1}(\bar{\Z})_\a)<\insertions) \leq \gammaq^{\shots+1} q^{-m\left(\insertionsmax-\insertions+1\right)}.
\end{equation*}
Note that this expression is independent of the rank partition $\insertionsVec$, so it is also an upper bound for $\Pr\!\left(\rk_{q^m}(\Lambda_{\nTransmit-\deletions-k}(\bar{\Z})_{\a})<\insertions\right)$ with $\errSpaceVec$ drawn according to Remark~\ref{rem:rand_instance_op_channel}.
\end{proof}

\begin{theorem}[{\LOlike} Decoder for~\ac{LILRS} Codes]\label{thm:LO_decoder_LILRS}
Suppose we transmit the tuple of subspaces 
\begin{equation}
 \txSpaceVec(\f) = \left(\txSpaceShot{1}(\f), \dots, \txSpaceShot{\shots}(\f)\right)\in
 \liftedIntLinRS{\vecbeta,\a,\shots,\intOrder;\nTransmitVec,k}
\end{equation}
over a random instance of the $\shots$-shot operator channel (see Remark~\ref{rem:rand_instance_op_channel}) 
\begin{equation*}
  \rxSpaceVec=\delOp{\nTransmit-\deletions}(\txSpaceVec)\oplus\errSpaceVec,
\end{equation*}
with overall $\insertions$ insertions and $\deletions$ deletions, where
\begin{align}\label{eq:dec_region_LO_LILRS}
\insertions \leq \insertionsmax := \intOrder(\nTransmit-\deletions-k).
\end{align} 
Then, Algorithm~\ref{alg:LO_LILRS} with input $\txSpaceVec(\f)$ returns the correct message polynomial vector $\f$ with success probability at least
\begin{align}\label{eq:succ_prob_LO_LILRS}
\Pr(\text{success}) \geq 1 - \gammaq^{\ell+1} q^{-m(\insertionsmax-\insertions+1)}.
\end{align}
Furthermore, the algorithm has complexity $O(\intOrder \nReceive^\omega)$ operations in $\Fqm$ plus $O(m \nReceive^{\omega-1})$ operations in $\Fq$.
\end{theorem}

\begin{proof}
Due to Proposition~\ref{lem:right_kernel_prop_LILRS}, the algorithm returns the correct message polynomial vector $\f$ if the $\Fqm$-rank of $\Lambda_{\nTransmit-\deletions-k}(\bar{\Z})_\a$ is at least $\insertions$. Hence, the success probability is lower bounded by the probability that $\rk_{q^m}(\Lambda_{\nTransmit-\deletions-k}(\bar{\Z})_{\a})=\insertions$, which is given in Lemma~\ref{lem:LO_LILRS_probability_full_Fqm_rank}.

The lines of the algorithm have the following complexities:
\begin{itemize}
  \item Lines~\ref{line:LILRS_check_if_kernel_dimension>1} and~\ref{line:LO_LILRS_h}: 
  This can be done by solving the linear system of equations $\LODecMat \h^\top = \0$. 
  Since $\LODecMat \in \Fqm^{((\intOrder+1)(\nTransmit-\deletions)-\intOrder k -1) \times \nReceive}$, it costs $\oh{\intOrder \nReceive^\omega}$ operations in $\Fqm$.
  \item Line~\ref{line:LO_LILRS_rk_hi} can be implemented by transforming the matrix representation of $\h^{(i)}$, which is an $m\times\nReceiveShot{i}$ matrix over $\Fq$, into column echelon form. For each $i$, this costs $\oh{m {\nReceiveShot{i}}^{\omega-1}}$ operations in $\Fq$. In total, all $\shots$ calls of this line cost $\oh{\shots m \sum_i {\nReceiveShot{i}}^{\omega-1}} \subseteq \oh{m \nReceive^{\omega-1}}$ operations in $\Fq$.
  \item Line~\ref{line:LO_LILRS_Ti} can be implemented by transforming the matrix representation of $\h^{(i)}$ into column echelon form, which was already accomplished in Line~\ref{line:LO_LILRS_rk_hi}. 
  \item Line~\ref{line:LO_LILRS_R_trans} requires $\oh{\intOrder {\nReceiveShot{i}}^2}$ multiplications over $\Fqm$ and thus $\oh{\intOrder\sum_i{\nReceiveShot{i}}^2}\subseteq\oh{\intOrder \nReceive^2}$ operations in $\Fqm$ in total.
  \item Line~\ref{line:LO_LILRS_IP} computes $\intOrder$ interpolation polynomials of degree less than $k\leq \nTransmit$ point tuples. This costs in total $\softoh{\intOrder \OMul{\nTransmit}}$ operations in $\Fqm$ (c.f.~\cite{PuchingerWachterzeh-ISIT2016,caruso2017fast}).
\end{itemize}
This proves the complexity statement.
\end{proof}

The decoding region of the {\LOlike} decoder for \ac{LILRS} codes is illustrated in Figure~\ref{fig:decodingRegion}.

\subsection{An Interpolation-Based Decoding Approach}\label{subsec:decodingLILRS}

We now derive an interpolation-based decoding approach for~\ac{LILRS} codes.
The decoding principle consists of an interpolation step and a root-finding step.
In~\cite{martinez2019reliable}, (lifted) linearized Reed--Solomon codes are decoded using the isometry between the sum-rank and the skew metric.
In this work we consider an interpolation-based decoding scheme in the generalized operator evaluation domain.
The new decoder is a generalization of \cite{wachter2014list} (interleaved Gabidulin codes in the rank metric) and \cite{bartz2018efficient} (lifted interleaved Gabidulin codes in the subspace metric).
Compared to the {\LOlike} decoder from Section~\ref{subsec:LO_dec_LILRS}, which requires $\oh{\intOrder^\omega\nReceive^2}$ operations in $\Fqm$, the proposed interpolation based decoder has a reduced computational complexity in the order of $\softoh{\intOrder^\omega\OMul{\nReceive}}$ operations in $\Fqm$.

\subsubsection{Interpolation Step}

Suppose we transmit the tuple of subspaces 
\begin{equation}
 \txSpaceVec(\f) = \left(\txSpaceShot{1}(\f), \dots, \txSpaceShot{\shots}(\f)\right) \in
 \liftedIntLinRS{\vecbeta,\a,\shots,\intOrder;\nTransmitVec,k}
\end{equation}
over an $\shots$-shot operator channel~\eqref{eq:def:multishot_op_channel} with $\insertions$ insertions and $\deletions$ deletions and receive the tuple of subspaces
\begin{equation}
 \rxSpaceVec = \left(\rxSpaceShot{1},\rxSpaceShot{2}, \dots, \rxSpaceShot{\shots}\right)\in \Grassm{\Nvec,\nReceiveVec}
\end{equation}
where the received subspaces $\rxSpaceShot{i}$ are represented as in~\eqref{eq:comp_rec_sapces}.
We describe $\rxSpaceVec$ by the tuple containing the basis matrices of the received subspaces as
\begin{equation}\label{eq:basis_rec_spaces}
    \U\defeq\left(\shot{\U}{1},\shot{\U}{2},\dots,\shot{\U}{\shots}\right)\in\prod_{i=1}^{\shots}\Fqm^{\nReceiveShot{i}\times (\intOrder+1)}
\end{equation}
where
\begin{equation}
    \shot{\U}{i}\defeq
    \begin{pmatrix}
   \vecxi^{(i)\top} & \u_1^{(i)\top} & \u_2^{(i)\top} & \dots & \u_\intOrder^{(i)\top}
  \end{pmatrix}
  \in\Fqm^{\nReceiveShot{i}\times (\intOrder+1)}
\end{equation}
has $\Fq$-rank $\rk_q(\shot{\U}{i})=\nReceiveShot{i}$ and satisfies $\rxSpaceShot{i}=\Rowspace{\shot{\U}{i}}$ for all $i=1,\dots,\shots$.
\begin{remark}\label{rem:rest_rx_spaces}
 In contrast to~\cite[Section~V.H]{martinez2019reliable} we do not need the assumption that the $\Fq$-rank of $\vecxi^{(i)}$ equals $\nReceiveShot{i}$ for all $i=1,\dots,\shots$, which is not the case in general (see also~\cite[Section~5.1.2]{Silva_PhD_ErrorControlNetworkCoding}).
\end{remark}
For a multivariate skew polynomial of the form
\begin{equation}\label{eq:mult_var_skew_poly}
  Q(x, y_1,\dots, y_\intOrder)=Q_0(x)+Q_1(x)y_1+\dots+Q_\intOrder(x)y_\intOrder
\end{equation}
where $Q_l(x)\in\SkewPolyringZeroDer$ for all $l\in\intervallincl{0}{\intOrder}$ define the $\nReceive$ generalized operator evaluation maps
\begin{align}\label{eq:defDiGenOpLILRS}
\MultSkewPolyringZeroDer \times \Fqm^{\intOrder+1} &\to \Fqm \notag \\
  \left(Q, (\xi_j^{(i)},u_{1,l}^{(i)},\dots,u_{\intOrder,l}^{(i)})\right) 
  &\mapsto\evalMap{j}{i}(Q)\defeq \opev{Q_0}{\xi_j^{(i)}}{a_i}+\textstyle\sum_{l=1}^{\intOrder}\opev{Q_l}{u_{j,l}^{(i)}}{a_i}.
\end{align}
for all $j=1,\dots,\nReceiveShot{i}$ and $i=1,\dots,\shots$.
Now consider the following interpolation problem in $\SkewPolyringZeroDer$. 

\begin{problem}[\ac{LILRS} Interpolation Problem]\label{prob:skewIntProblemGenOpLILRS}
 Given the integers $\degConstraint,\intOrder,\shots\in\ZZ_{\geq0}$, a set  
 \begin{equation}
  \mathscr{E}=\left\{\evalMap{j}{i}:i=1,\dots,\shots, j=1,\dots,n_i\right\}
 \end{equation}
 containing the generalized operator evaluation maps defined in~\eqref{eq:defDiGenOpLILRS} and a vector $\vec{w}=(0,k-1,\dots,k-1)\in\ZZ_{\geq0}^{\intOrder+1}$, find a nonzero polynomial of the form 
 \begin{equation}
  Q(x, y_1,\dots, y_\intOrder)=Q_0(x)+Q_1(x)y_1+\dots+Q_\intOrder(x)y_\intOrder
 \end{equation}
 with $Q_l(x)\in\SkewPolyringZeroDer$ for all $l\in\intervallincl{0}{\intOrder}$ that satisfies:
 \begin{enumerate}
  \item $\evalMap{j}{i}(Q)=0, \qquad\forall i=1,\dots,\shots$, $j=1,\dots,\nReceiveShot{i}$,
  \item $\deg_{\vec{w}}(Q(x,y_1,\dots,y_\intOrder))<\degConstraint$.
 \end{enumerate}
\end{problem}

Define the skew polynomials 
\begin{equation}
  Q_0(x)=\sum_{i=0}^{\degConstraint-1}q_{0,i}x^i
  \qquad\text{and}\qquad
  Q_j(x)=\sum_{i=0}^{\degConstraint-k}q_{j,i}x^i,
\end{equation}
and the vectors
\begin{equation}
  \vecxi=\left(\vecxi^{(1)} \mid \vecxi^{(2)} \mid\dots\mid \vecxi^{(\shots)}\right)\in\Fqm^{\nReceive}
  \quad\text{and}\quad
  \u_j=\left(\u_j^{(1)} \mid \u_j^{(2)} \mid\dots\mid \u_j^{(\shots)}\right)\in\Fqm^{\nReceive}
\end{equation}
for all $j=1,\dots,\intOrder$.
Then a solution of Problem~\ref{prob:skewIntProblemGenOpLILRS} can be found by solving the $\Fqm$-linear system
\begin{equation}\label{eq:intSystemLILRS}
  \intMat\vec{q}=\vec{0}
\end{equation}
for
\begin{equation}\label{eq:defIntVec_LILRS}
  \vec{q} = \left(q_{0,0},q_{0,1},\dots,q_{0,\degConstraint-1}\mid q_{1,0},q_{1,1},\dots,q_{1,\degConstraint-k}\mid \dots \mid  q_{\intOrder,0},q_{\intOrder,1},\dots,q_{\intOrder,\degConstraint-k}\right)
\end{equation}
where the interpolation matrix $\intMat\in\Fqm^{\nReceive\times \degConstraint(\intOrder+1)-\intOrder(k-1)}$ is given by
\begin{equation}\label{eq:intMatrixLILRS}
  \intMat=
  \begin{pmatrix}
   \Lambda_{\degConstraint}(\vecxi)_\a^\top & \Lambda_{\degConstraint-k+1}(\u_1)_\a^\top & \dots & \Lambda_{\degConstraint-k+1}(\u_\intOrder)_\a^\top
  \end{pmatrix}.
\end{equation}

Problem~\ref{prob:skewIntProblemGenOpLILRS} can be solved by the skew Kötter interpolation~\cite{liu2014kotter} with the generalized operator evaluation maps $\evalMap{j}{i}$ as defined in~\eqref{eq:defDiGenOpLILRS} requiring $\oh{\intOrder^2 \nReceive^2}$ operations in $\Fqm$.
A solution of Problem~\ref{prob:skewIntProblemGenOpLILRS} can be found efficiently requiring only $\softoh{\intOrder^\omega \OMul{\nReceive}}$ operations in $\Fqm$ using a variant of the minimal approximant bases approach from~\cite{bartz2022fast}.
Another approach yielding the same computational complexity of $\softoh{\intOrder^\omega \OMul{\nReceive}}$ operations in $\Fqm$ is given by the fast divide-and-conquer Kötter interpolation from~\cite{bartz2021fastSkewKNH}.

\begin{lemma}[Existence of Solution]
 A nonzero solution of Problem~\ref{prob:skewIntProblemGenOpLILRS} exists if $\degConstraint=\big\lceil\tfrac{\nReceive+\intOrder(k-1)+1}{\intOrder+1}\big\rceil$.
\end{lemma}

\begin{proof}
 Problem~\ref{prob:skewIntProblemGenOpLILRS} corresponds to a system of $\nReceive$ $\Fqm$-linear equations in $\degConstraint(\intOrder+1)-\intOrder(k-1)$ unknowns (see~\eqref{eq:intSystemLILRS}) which has a nonzero solution if the number of equations is less than the number of unknowns, i.e. if
 \begin{equation}
  \nReceive<\degConstraint(\intOrder+1)-\intOrder(k-1)\label{eq:existenceCondLILRS}
  \quad\Longleftrightarrow\quad
  \degConstraint\geq\tfrac{\nReceive+\intOrder(k-1)+1}{\intOrder+1}.
 \end{equation} 
\end{proof}

The $\Fqm$-linear solution space $\Qspace$ of Problem~\ref{prob:skewIntProblemGenOpLILRS} is defined as
\begin{equation}\label{eq:def_sol_space_LILRS}
 \Qspace\defeq\{Q\in\MultSkewPolyringZeroDer:\vec{q}(Q)\in\rker(\intMat)\}  
\end{equation}
where $\vec{q}(Q)\in\Fqm^{\degConstraint(\intOrder+1)-\intOrder(k-1)}$ denotes the coefficient vector of $Q$ as defined in~\eqref{eq:defIntVec_LILRS}.
The dimension of the $\Fqm$-linear solution space $\Qspace$ of Problem~\ref{prob:skewIntProblemGenOpLILRS} (i.e. the dimension of the right kernel of $\intMat$ in~\eqref{eq:intMatrixLILRS}) is denoted by 
\begin{equation}
  d_I\defeq\dim(\Qspace)=\dim(\rker(\intMat)).
\end{equation}

\subsubsection{Root-Finding Step}\label{subsubsec:RF_LILRS}

The goal of the root-finding step is to recover the message polynomials $f_1,\dots,f_\intOrder\in\SkewPolyringZeroDer_{<k}$ from the multivariate polynomial constructed in the interpolation step. 
We now derive a condition for the recovery of the message polynomials.

\begin{lemma}[Roots of Polynomial]\label{lem:decConditionLILRS}
 Let
 \begin{equation}
  P(x)\defeq Q_0(x)+Q_1(x)f_1(x)+\dots+Q_\intOrder(x)f_\intOrder(x).
 \end{equation}
 Then there exist elements $\zeta_1^{(i)},\dots,\zeta_{\nTransmitShot{i}-\deletionsShot{i}}^{(i)}$ in $\Fqm$ that are $\Fq$-linearly independent for each $i=1,\dots,\shots$ such that
 \begin{equation}
  \opev{P}{\zeta_j^{(i)}}{a_i}=0
 \end{equation}
 for all $i=1,\dots,\shots$ and $j=1,\dots,\nTransmitShot{i}-\deletionsShot{i}$.
\end{lemma}

\begin{proof}
 In each shot the non-corrupted intersection space has dimension $\dim(\rxSpaceShot{i}\cap\txSpaceShot{i})=\nTransmitShot{i}-\deletionsShot{i}$ for all $i=1,\dots,\shots$.
 A basis for each intersection space $\rxSpaceShot{i}\cap\txSpaceShot{i}$ can be represented as
 \begin{equation}
    \left\{\left(\zeta_j^{(i)},\opev{f_1}{\zeta_j^{(i)}}{a_i},\!\dots,\opev{f_\intOrder}{\zeta_j^{(i)}}{a_i}\right)\!:\!j\!\in\!\intervallincl{1}{\nTransmitShot{i}\!-\!\deletionsShot{i}}\right\}
 \end{equation}
 where $\zeta_1^{(i)},\dots,\zeta_{\nTransmitShot{i}-\deletionsShot{i}}^{(i)}$ are $(\nTransmitShot{i}-\deletionsShot{i})$ $\Fq$-linearly independent elements from $\Fqm$ for all $i=1,\dots,\shots$.
 Since each intersection space $\rxSpaceShot{i}\cap\txSpaceShot{i}$ is a subspace of the received space $\rxSpaceShot{i}$ we have that
 \begin{equation}\label{eq:decConditionLILRS}
  \opev{P}{\zeta_j^{(i)}}{a_i}\defeq\opev{Q_0}{\zeta_j^{(i)}}{a_i}+\sum_{l=1}^{\intOrder}\opev{Q_l}{\opev{f_l}{\zeta_j^{(i)}}{a_i}}{a_i}
  =0
 \end{equation}
 for all $i=1,\dots,\shots, j=1,\dots,\nTransmitShot{i}-\deletionsShot{i}$.
\end{proof} 

\begin{theorem}[Decoding Region]\label{thm:dec_region_lilrs}
 Let $\rxSpaceVec \in \Grassm{\Nvec,\nReceiveVec}$ be the tuple containing the received subspaces and let $Q(x,y_1,\dots,y_\intOrder)\neq0$ fulfill the constraints in Problem~\ref{prob:skewIntProblemGenOpLILRS}. 
 Then for all codewords $\txSpaceVec(\f)\in\liftedIntLinRS{\vecbeta,\a,\shots,\intOrder;\nTransmitVec,k}$ that are $(\insertions,\deletions)$-reachable from $\rxSpaceVec$, where $\insertions$ and $\deletions$ satisfy 
 \begin{equation}\label{eq:listDecRegion}
  \insertions+\intOrder \deletions<\intOrder(\nTransmit-k+1),
 \end{equation}
 we have that
 \begin{equation}\label{eq:rootFindingEquationLILRS}
  P(x)\!=\!Q_0(x)+Q_1(x)f_1(x)+\!\dots\!+Q_\intOrder(x)f_\intOrder(x)\!=\!0.
 \end{equation}
\end{theorem}

\begin{proof}
 By Lemma~\ref{lem:decConditionLILRS} there exist elements $\zeta_1^{(i)},\dots,\zeta_{\nTransmitShot{i}-\deletionsShot{i}}$ in $\Fqm$ that are $\Fq$-linearly independent for each $i=1,\dots,\shots$ such that
 \begin{equation}
  \opev{P}{\zeta_j^{(i)}}{a_i}=0
 \end{equation}
 for all $i=1,\dots,\shots$ and $j=1,\dots,\nTransmitShot{i}-\deletionsShot{i}$.
 By choosing 
 \begin{equation}\label{eq:decDegreeConstraintLILRS}
  \degConstraint\leq \nTransmit-\deletions
 \end{equation}
 the degree of $P(x)$ exceeds the degree bound from~\cite{caruso2019residues}) which is possible only if $P(x)=0$.
 Combining~\eqref{eq:existenceCondLILRS} and~\eqref{eq:decDegreeConstraintLILRS} we get
 \begin{align*}
  \nReceive+\intOrder(k-1)&<\degConstraint(\intOrder+1)\leq(\intOrder+1)(\nTransmit-\deletions)
  \\\Longleftrightarrow\qquad
  \insertions+\intOrder\deletions&<\intOrder(\nTransmit-k+1).
 \end{align*}
\end{proof}

The decoding region in~\eqref{eq:listDecRegion} shows and improved insertion-correction performance due to interleaving. The resulting improvement is illustrated in Figure~\ref{fig:decodingRegion}.

In the root-finding step, all polynomials $f_1,\dots,f_\intOrder \in \SkewPolyringZeroDer_{<k}$ that satisfy~\eqref{eq:rootFindingEquationLILRS} need to be found. 
Instead of using only one solution of Problem~\ref{prob:skewIntProblemGenOpLILRS} to set up the root-finding system we use a basis for the $d_I$-dimensional $\Fqm$-linear solution space $\Qspace$ (see also~\cite{wachter2014list,bartz2017algebraic}.
Alternatively, a degree-restricted subset of a Gröbner basis for the interpolation module of cardinality at most $\intOrder$ can be used to set up the root-finding system and find the minimal number of solutions (see~\cite{bartz2022fast}). 

To set up the root-finding system set up with a basis for $\Qspace$ define the matrices 
 \begin{equation}\label{eq:def_RF_blocks}
  \autMat{\Q}{j}{i}\defeq
  \begin{pmatrix}
   \aut^i\left(q_{1,j}^{(1)}\right) & \aut^i\left(q_{2,j}^{(1)}\right) & \dots & \aut^i\left(q_{\intOrder,j}^{(1)}\right)
   \\
   \vdots & \vdots & \ddots & \vdots
   \\ 
   \aut^i\left(q_{1,j}^{(d_I)}\right) & \aut^i\left(q_{2,j}^{(d_I)}\right) & \dots & \aut^i\left(q_{\intOrder,j}^{(d_I)}\right)
  \end{pmatrix}\in\Fqm^{d_I \times s}
 \end{equation}
 and the vectors
 \begin{equation}
  \autVec{\f}{j}{i}
  \defeq\left(\aut^i\left(f_j^{(1)}\right),\dots,\aut^i\left(f_j^{(\intOrder)}\right)\right)\in\Fqm^{\intOrder}
 \end{equation}
 and
 \begin{equation}\label{eq:def_RF_vecs}
  \autVec{\q}{0,j}{i}
  \defeq\left(\aut^i\left(q_{0,j}^{(1)}\right),\dots,\aut^i\left(q_{0,j}^{(d_I)}\right)\right)\in\Fqm^{d_I}.
 \end{equation}

Defining the root-finding matrix
 \begin{equation}\label{eq:rootFindingMatrixLILRS}
 \RFmat\defeq
 \begin{pmatrix}
 \phantom{\aut^{-}(}\Q_0        &     &  &              \\
 \autMat{\Q}{1}{-1}         & \autMat{\Q}{0}{-1}  &        &              \\[-3pt]
 \vdots             & \autMat{\Q}{1}{-2}  & \ddots &              \\[-3pt]
 \autMat{\Q}{\degConstraint-k}{-(\degConstraint-k)} & \vdots    & \ddots & \autMat{\Q}{0}{-(k-1)}       \\[-3pt]
              & \autMat{\Q}{\degConstraint-k}{-(\degConstraint-k-1)}    & \ddots & \autMat{\Q}{1}{-k\phantom{(-1)}}           \\[-5pt]
                                          &     & \ddots & \vdots          \\
              &           & & \autMat{\Q}{\degConstraint-k}{-(\degConstraint-1)}
 \end{pmatrix}\in\Fqm^{\degConstraint d_I \times \intOrder k}
\end{equation}
and the vectors
\begin{equation*}
 \RFvec\defeq\left(\vec{f}_0,\autVec{\f}{1}{-1},\dots,\autVec{\f}{k-1}{-(k-1)}\right)^\top\in\Fqm^{\intOrder k}
\end{equation*}
and
\begin{equation*}
 \vec{q}_{0}\defeq\left(\vec{q}_{0,0},\autVec{\vec{q}_{0,1}}{}{-1},\dots, \autVec{\vec{q}_{0,\degConstraint-1}}{}{-(\degConstraint-1)}\right)^\top\in\Fqm^{\degConstraint d_I}
\end{equation*}
as in~\cite{bartz2022fast} we can write the root-finding system~\eqref{eq:rootFindingEquationLILRS} as
\begin{equation}\label{eq:rootFindingSystemFqmLILRS}
 \RFmat\cdot\RFvec=-\vec{q}_0.
\end{equation}

The root-finding system can be solved efficiently by the minimal approximant bases method in~\cite[Algorithm~7]{bartz2020fast} (see also~\cite[Section~IV.C]{bartz2022fast}) requiring at most $\softoh{\intOrder^\omega \OMul{n}}$ operations in $\Fqm$.

\subsubsection{List Decoding}\label{subsubsec:LILRS_list}

In general, the root-finding matrix $\RFmat$ in~\eqref{eq:rootFindingSystemFqmLILRS} can be rank deficient.
In this case we obtain a \emph{list} of potential message polynomials $f_1,\dots,f_\intOrder$.
By~\cite[Proposition~4]{bartz2022fast} the root-finding system in~\eqref{eq:rootFindingEquationLILRS} has at most $q^{m(k(\intOrder-1))}$ solutions $f_,\dots,f_\intOrder \in\SkewPolyringZeroDer_{<k}$.
In general, we have that $k\leq\nTransmit$, where $\nTransmit\leq\shots m$.
Hence, for $m\approx\nTransmit/\shots$ we get a worst-case list size of $q^{\frac{\nTransmit}{\shots}(k(\intOrder-1))}$.

\begin{algorithm}[ht]
    \caption{\algoname{List Decoding of \ac{LILRS} Codes}}
    \label{alg:list_dec_LILRS}

    \begin{algorithmic}[1]

        \Statex \textbf{Input:} A tuple containing the basis matrices $\U=(\shot{\U}{1},\shot{\U}{2}, \dots, \shot{\U}{\shots})\in\prod_{i=1}^{\shots}\Fqm^{\nReceiveShot{i}\times(\intOrder+1)}$ for the output $\rxSpaceVec=(\rxSpaceShot{1},\rxSpaceShot{2},\dots,\rxSpaceShot{\shots}) \in \Grassm{\Nvec,\nReceiveVec}$ of an $\shots$-shot operator channel with overall $\insertions$ insertions and $\deletions$ deletions for input $\txSpaceVec(\f)\in\liftedIntLinRS{\vecbeta,\a,\shots,\intOrder;\nTransmitVec,k}$

        \Statex

        \Statex \textbf{Output:} A list $\List$ containing message polynomial vectors $\f=(f_1,\dots,f_\intOrder) \in \SkewPolyringZeroDer_{<k}^\intOrder$ that satisfy~\eqref{eq:rootFindingEquationLILRS}

        \Statex

        \State Find left $\SkewPolyringZeroDer$-linearly independent $Q^{(1)},\dots,Q^{(\IntParam)} \in \Qspace \setminus \{0\}$ whose left $\SkewPolyringZeroDer$-span contains the $\Fqm$-linear solution space $\Qspace$ of Problem~\ref{prob:skewIntProblemGenOpLILRS}

        \State Using $Q^{(1)},\dots,Q^{(\IntParam)}$, find the list $\List\subseteq\SkewPolyringZeroDer^{\intOrder}_{<k}$ of all $\f=(f_1,\dots,f_\intOrder) \in \SkewPolyringZeroDer_{<k}^\intOrder$ that satisfy~\eqref{eq:rootFindingEquationLILRS}

        \State \Return $\List$
    
    \end{algorithmic}
\end{algorithm}

\begin{theorem}[List Decoding of \ac{LILRS} Codes]\label{thm:list_dec_LILRS}
 Let $\rxSpaceVec \in \Grassm{\Nvec,\nReceiveVec}$ be a tuple of received subspaces of a transmission of a codeword $\txSpaceVec\in\liftedIntLinRS{\vecbeta,\a,\shots,\intOrder;\nTransmitVec,k}$ over an $\shots$-shot operator channel with overall $\insertions$ insertions and $\deletions$ deletions.
 If the number of overall insertions $\insertions$ and deletions $\deletions$ satisfy 
 \begin{equation}
  \insertions+\intOrder \deletions<\intOrder(\nTransmit-k+1),
 \end{equation}
 then a list $\List$ of size 
 \begin{equation}
    |\List|\leq q^{m(k(\intOrder-1))}
 \end{equation}
 containing all message polynomial vectors $\f\in\SkewPolyringZeroDer_{<k}^\intOrder$ corresponding to codewords $\txSpaceVec(\f)\in\liftedIntLinRS{\vecbeta,\a,\shots,\intOrder;\nTransmitVec,k}$ that are $(\insertions,\deletions)$-reachable from $\rxSpaceVec$ can be found requiring at most $\softoh{\intOrder^\omega\OMul{\nReceive}}$ operations in $\Fqm$.
\end{theorem}

\begin{proof}
    The proof follows directly from Lemma~\ref{lem:decConditionLILRS}, Theorem~\ref{thm:dec_region_lilrs} and the discussion above.
\end{proof}

\subsubsection{Probabilistic Unique Decoding}\label{subsubsec:LILRS_unique}

We now consider the interpolation-based decoder from Section~\ref{subsec:decodingLILRS} as a probabilistic unique decoder which either returns a unique solution (if the list size is equal to one) or a decoding failure. 
The main idea is to use a basis for the $d_I$-dimensional $\Fqm$-linear solution space $\Qspace$ of the interpolation system~\eqref{eq:intSystemLILRS} to set up the root-finding matrix~\eqref{eq:rootFindingMatrixLILRS} which in turn facilitates that the root-finding matrix $\RFmat$ can have full rank.

Using similar arguments as in~\cite{bartz2018efficient,bartz2017algebraic,bartz2022fast} we can lower bound the dimension $d_I$ of the $\Fqm$-linear solution space $\Qspace$ of Problem~\ref{prob:skewIntProblemGenOpLILRS}.

\begin{lemma}[Dimension of Solution Space]\label{lem:dim_Qspace_LILRS}
 Let $\insertions$ and $\deletions$ satisfy~\eqref{eq:listDecRegion}.
 Then the dimension $d_I=\dim(\Qspace)$ of the $\Fqm$-linear solution space $\Qspace$ of Problem~\ref{prob:skewIntProblemGenOpLILRS} satisfies
 \begin{equation}
  d_I\geq \intOrder(\degConstraint+1)-\intOrder k-\insertions.
 \end{equation}
\end{lemma}

\begin{proof}
 Let
 \begin{equation}
    \left\{\left(\zeta_j^{(i)},\opev{f_1}{\zeta_j^{(i)}}{a_i},\dots,\opev{f_\intOrder}{\zeta_j^{(i)}}{a_i}\right):j\in\intervallincl{1}{\nTransmitShot{i}-\deletionsShot{i}}\right\}
 \end{equation}
 be a basis for each non-corrupted intersection space $\rxSpaceShot{i}\cap\txSpaceShot{i}$ where $\zeta_1^{(i)},\dots,\zeta_{\nTransmitShot{i}-\deletionsShot{i}}^{(i)}$ are $(\nTransmitShot{i}-\deletionsShot{i})$ $\Fq$-linearly independent elements from $\Fqm$ for all $i=1,\dots,\shots$. 
 Define the vector $\vec{\zeta}\defeq(\veczeta^{(1)} \mid \veczeta^{(2)} \mid \dots \mid \veczeta^{(\shots)})\in\Fqm^{\nTransmit-\deletions}$ where $\veczeta^{(i)}=(\zeta_1^{(i)},\dots,\zeta_{\nTransmitShot{i}-\deletionsShot{i}}^{(i)})$ for all $i=1,\dots,\shots$.
 Let 
 \begin{equation}
   \left\{\left(\epsilon_j^{(i)},\tilde{e}_{j,1}^{(i)},\dots,\tilde{e}_{j,\intOrder}^{(i)},\right):j\in\intervallincl{1}{\insertionsShot{i}}\right\}
 \end{equation}
 be a basis for the error space $\errSpaceShot{i}$ for all $i=1,\dots,\shots$ and define
 \begin{align*}
   \vecepsilon&=(\vecepsilon^{(1)} \mid \vecepsilon^{(2)} \mid \dots \mid \vecepsilon^{(\shots)})\in\Fqm^{\insertionsShot{i}}
   \quad\text{with}\quad
   \vecepsilon^{(i)}=(\epsilon_1^{(i)},\epsilon_2^{(i)},\dots,\epsilon_{\insertionsShot{i}}^{(i)}),
   \\
   \tilde{\vec{e}}_j&=(\tilde{\vec{e}}_j^{(1)} \mid \tilde{\vec{e}}_j^{(2)} \mid \dots \mid \tilde{\vec{e}}_j^{(\shots)})\in\Fqm^{\insertionsShot{i}}
   \quad\text{with}\quad
   \tilde{\vec{e}}_j^{(i)}=(\tilde{e}_{1,j}^{(i)},\tilde{e}_{2,j}^{(i)},\dots,\tilde{e}_{\insertionsShot{i},j}^{(i)}).
 \end{align*}
 Then the matrix
 \begin{equation}
   \widetilde{\R}_I=
   \begin{pmatrix}
    \Lambda_{\degConstraint}(\veczeta)_\a^\top & \0 & \dots & \0
    \\ 
    \Lambda_{\degConstraint}(\vecepsilon)_\a^\top & \Lambda_{\degConstraint-k+1}(\tilde{\vec{e}}_1)_\a^\top & \dots & \Lambda_{\degConstraint-k+1}(\tilde{\vec{e}}_\intOrder)_\a^\top
   \end{pmatrix}
   \in\Fqm^{\nReceive\times \degConstraint(\intOrder+1)-\intOrder(k-1)}
 \end{equation}
 has the same column space as the matrix $\intMat$ in~\eqref{eq:intMatrixLILRS}.
 Since $\SumRankWeight(\veczeta) = \nTransmit-\deletions$ and $\degConstraint\leq\nTransmit-\deletions$ (see~\eqref{eq:decDegreeConstraintLILRS}) we have that the matrix $\Lambda_{\degConstraint}(\veczeta)_\a^\top$ has $\Fqm$-rank $\degConstraint$.
 The last $\insertions$ rows of $\widetilde{\R}_I$ can increase the $\Fqm$-rank of $\widetilde{\R}_I$ by at most $\insertions$.
 Thus we have that $\rk_{q^m}(\intMat)=\rk_{q^m}(\widetilde{\R}_I)\leq \degConstraint+\insertions$.
 Hence, the dimension $d_I$ of the $\Fqm$-linear solution space $\Qspace$ of Problem~\ref{prob:skewIntProblemGenOpLILRS} satisfies
 \begin{align*}
  d_I=\dim(\rker(\intMat))&=\degConstraint(\intOrder+1)-\intOrder(k-1)-\rk_{q^m}(\intMat)
  \\ 
  &\geq \intOrder(\degConstraint+1)-\intOrder k -\insertions.
 \end{align*}
\end{proof}

The rank of the root-finding matrix $\RFmat$ can be full if and only if the dimension of the solution space of the interpolation problem $d_I$ is at least $\intOrder$, i.e. if
\begin{align}
 d_I\geq\intOrder\qquad\Longleftrightarrow\qquad
 \insertions&\leq\intOrder\degConstraint-\intOrder k\nonumber
 \\\Longleftrightarrow\qquad
 \insertions+\intOrder\deletions&\leq \intOrder(\nTransmit-k). \label{eq:decRegionLILRSprob}
\end{align}
The probabilistic unique decoding region in~\eqref{eq:decRegionLILRSprob} is only sightly smaller than the list decoding region in~\eqref{eq:listDecRegion}.
The improved decoding region for~\ac{LILRS} codes is illustrated in Figure~\ref{fig:decodingRegion}.
Recall from Remark~\ref{rem:rest_rx_spaces} that unlike the proposed decoder, the decoder in~\cite{martinez2019reliable} has the restriction that $\nReceive=\nTransmit$, which corresponds to the case that $\insertions = \deletions$.
\begin{figure}[ht!]
\centering
%
%
\definecolor{mycolor1}{rgb}{0.00000,0.44700,0.74100}%
\definecolor{mycolor2}{rgb}{0.85000,0.32500,0.09800}%
\begin{tikzpicture}

\begin{axis}[%
xmin=0,
xmax=13,
xtick={ 0,  1,  2,  3,  4,  5,  6,  7,  8,  9, 10, 11, 12, 13},
xlabel={$\insertions\text{ insertions}$},
xmajorgrids,
ymin=0,
ymax=4,
ytick={0, 1, 2, 3, 4},
ylabel={$\deletions\text{ deletions}$},
ymajorgrids,
legend style={at={(0.99,0.965)},anchor=north east,legend cell align=left,align=left, draw=white!15!black},
mystyle,
height=0.28\columnwidth,
width=0.85\columnwidth,
reverse legend
]


\addplot[const plot, color=gray,  dotted, line width=1.5pt] plot table[row sep=crcr]{%
-2	-5\\
-1	3\\
0.02	3\\
1	3\\
2	3\\
3	2\\
4	2\\
5	2\\
6	2\\
7	1\\
8	1\\
9	1\\
10	1\\
11	0.02\\
12	0.02\\
13	0.02\\
14	0.02\\
15	-1\\
16	-1\\
17	-3\\
};
\addlegendentry{List decoding~\eqref{eq:listDecRegion} $(\intOrder=4)$};

\addplot[const plot, fill=cyan,color=cyan,  pattern=north east lines, pattern color = cyan, solid, line width=1.5pt] plot table[row sep=crcr] {%
-3	-3 \\
-2	3\\
-1	3\\
0.01	2\\
1	2\\
2	2\\
3	2\\
4	1\\
5	1\\
6	1\\
7	1\\
8	0.02\\
9	0.02\\
10	0.02\\
11	0.02\\
12	-3\\
};
\addlegendentry{Probabilistic unique decoding~\eqref{eq:dec_region_LO_LILRS},~\eqref{eq:decRegionLILRSprob} $(\intOrder=4)$};

\addplot[const plot, fill = blue, color=blue, pattern=north west lines, pattern color = blue,  dashed, line width=1.5pt] plot table[row sep=crcr] {%
-2	-2\\
-1	3\\
0.02	2\\
1	1\\
2	0.02\\
3	-1\\
4	-2\\
5	-2\\
6	-2\\
7	-2\\
8	-2\\
9	-2\\
10	-2\\
11	-2\\
12	-2\\
13	-2\\
};
\addlegendentry{Mart{\'\i}nez-Pe{\~n}as--Kschischang \cite{martinez2019reliable} $(\intOrder=1)$};

\node[anchor=east] (A) at (axis cs: 0.01, 1.45){};
\node[anchor=west] (B) at (axis cs: 3.965,1.45){};
\draw [<->, line width=1pt] (A) edge (B);
\node (T) at (axis cs: 2,1.65){\footnotesize{$\intOrder$}};

\node[anchor=east] (C) at (axis cs: 0.01, 0.45){};
\node[anchor=west] (D) at (axis cs: 7.965,0.45){};
\draw [<->, line width=1pt] (C) edge (D);
\node (T) at (axis cs: 4,0.65){\footnotesize{$2\intOrder$}};

\node[anchor=east] (A) at (axis cs: 0.05, 2.45){};
\node[anchor=west] (B) at (axis cs: 2.97,2.45){};
\draw [<->, line width=1pt,color=gray] (A) edge (B);
\node[color=gray] (T) at (axis cs: 1.5,2.65){\footnotesize{$\intOrder-1$}};

\node[anchor=east] (A) at (axis cs: 4.035, 1.45){};
\node[anchor=west] (B) at (axis cs: 6.97,1.45){};
\draw [<->, line width=1pt,color=gray] (A) edge (B);
\node[color=gray] (T) at (axis cs: 5.5,1.65){\footnotesize{$\intOrder-1$}};

\node[anchor=east] (A) at (axis cs: 8.035, 0.45){};
\node[anchor=west] (B) at (axis cs: 10.97,0.45){};
\draw [<->, line width=1pt,color=gray] (A) edge (B);
\node[color=gray] (T) at (axis cs: 9.5,0.65){\footnotesize{$\intOrder-1$}};

\end{axis}  

\end{tikzpicture}%
\caption{Decoding region for Mart{\'\i}nez-Pe{\~n}as--Kschischang~\cite{martinez2019reliable} codes $(\intOrder=1)$ and for decoding of lifted $(\intOrder=4)$-interleaved linearized Reed--Solomon codes.
The decoding region for insertions increases with the interleaving order $\intOrder$.
}\label{fig:decodingRegion}
\end{figure}
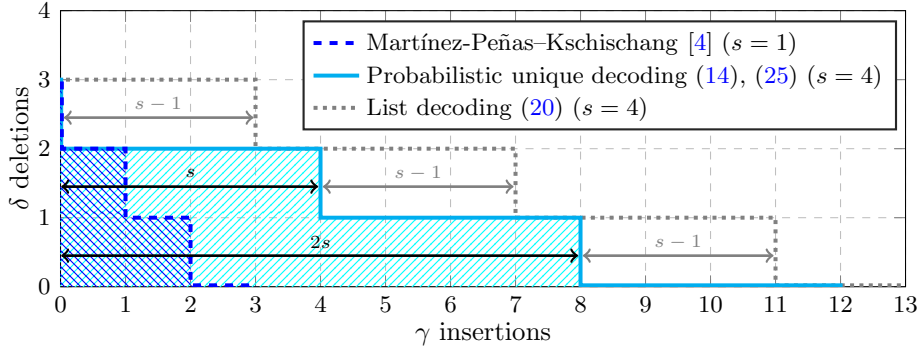

Combining~\eqref{eq:decDegreeConstraintLILRS} and~\eqref{eq:decRegionLILRSprob}  we get the degree constraint for the probabilistic unique decoder (see~\cite{bartz2017algebraic})
\begin{equation}
  \degConstraintUnique=\left\lceil\frac{\nReceive+\intOrder k}{\intOrder+1}\right\rceil.
\end{equation} 

In order to get an estimate of the probability of successful decoding, we use similar assumptions as in~\cite{wachter2014list,bartz2017algebraic} to derive a heuristic upper bound on the decoding failure probability $P_f$.

The root-finding matrix $\RFmat$ in~\eqref{eq:rootFindingMatrixLILRS} contains a lower block-diagonal matrix with $\Q_0, \autMat{\Q}{0}{-1},\dots,\autMat{\Q}{0}{-(k-1)}$ on the diagonal.
Since $\rk_{q^m}(\autMat{\Q}{0}{i})=\rk_{q^m}(\Q_0)$ for all $i$, this implies that $\rk_{q^m}(\RFmat) = \intOrder k$ if $\rk_{q^m}(\Q_0) = \intOrder$.

Under the assumption that the coefficients $q_{i,j}^{(r)}$ are uniformly distributed over $\Fqm$ (see~\cite[Lemma~9]{wachter2014list,bartz2017algebraic}), we can upper bound the decoding failure probability $P_f \defeq \Pr(\rk_{q^m}(\RFmat) < \intOrder k)$ by the probability that the $(d_I \times \intOrder)$ matrix $\Q_0$ with uniformly distributed elements from $\Fqm$ has $\Fqm$-rank less than $\intOrder$ and get
\begin{equation}\label{eq:heur_upper_bound_Pfail_LILRS}
  P_f
  \leq \gammaq q^{-m(d_I-\intOrder+1)}
  \leq \gammaq q^{-m\left(\intOrder\left(\left\lceil\frac{\nReceive+\intOrder k}{\intOrder+1}\right\rceil-k\right)-\insertions+1\right)}.
\end{equation}
Note, that for $\degConstraint=\nTransmit-\deletions$ (see~\eqref{eq:decDegreeConstraintLILRS}) we get
\begin{equation*}
  P_f
  \leq \gammaq q^{-m(\insertionsmax-\insertions+1)}.
\end{equation*}

Note, that the assumption that the coefficients $q_{i,j}^{(r)}$ are uniformly distributed over $\Fqm$ does not reflect the distribution of the error space tuple $\errSpaceVec$.
Although there is evidence that this assumption is reasonable (see e.g.~\cite{hoermann2022folded} for folded \ac{LRS} codes), it does not reflect the actual error model of the multishot operator channel.

Similar as in~\cite[Lemma~8]{wachter2014list} for interleaved Gabidulin codes and~\cite[Theorem~4]{bartz2022fast} for \ac{ILRS} codes, the conditions of successful decoding of the interpolation-based decoder can be reduced to the conditions of the {\LOlike} decoder from Section~\ref{subsec:LO_dec_LILRS}.
This reduction allows to obtain an upper bound on the decoding failure probability since the distribution of the error space tuple $\errSpaceVec$ is considered in the derivation.
The results of the interpolation-based probabilistic unique decoder are summarized in Algorithm~\ref{alg:unique_dec_LILRS} and Theorem~\ref{thm:unique_dec_LILRS}.

\begin{algorithm}[ht]
    \caption{\algoname{Probabilistic Unique Decoding of \ac{LILRS} Codes}}\label{alg:unique_dec_LILRS}

    \begin{algorithmic}[1]
        \Statex \textbf{Input:} A tuple containing the basis matrices $\U=(\shot{\U}{1},\shot{\U}{2}, \dots, \shot{\U}{\shots})\in\prod_{i=1}^{\shots}\Fqm^{\nReceiveShot{i}\times(\intOrder+1)}$ for the output $\rxSpaceVec=(\rxSpaceShot{1},\rxSpaceShot{2},\dots,\rxSpaceShot{\shots})\in \Grassm{\Nvec,\nReceiveVec}$ of an $\shots$-shot operator channel with overall $\insertions$ insertions and $\deletions$ deletions for input $\txSpaceVec(\f)\in\liftedIntLinRS{\vecbeta,\a,\shots,\intOrder;\nTransmitVec,k}$

        \Statex
    
        \Statex \textbf{Output:} Message polynomial vector $\f=(f_1,\dots,f_\intOrder) \in \SkewPolyringZeroDer_{<k}^\intOrder$ or ``decoding failure''

        \Statex

        \State Define the generalized operator evaluation maps as in~\eqref{eq:defDiGenOpLILRS}

        \State Find left $\SkewPolyringZeroDer$-linearly independent $Q^{(1)},\dots,Q^{(\IntParam)} \in \Qspace \setminus \{0\}$ whose left $\SkewPolyringZeroDer$-span contains the $\Fqm$-linear solution space $\Qspace$ of Problem~\ref{prob:skewIntProblemGenOpLILRS}

        \If{$\IntParam=\intOrder$}

            \State Use $Q^{(1)},\dots,Q^{(\intOrder)}$ to find the unique vector $\f= (f_1,\dots,f_\intOrder) \in \SkewPolyringZeroDer_{<k}^\intOrder$ that satisfies~\eqref{eq:rootFindingEquationLILRS}

        \Else
            \State \Return ``decoding failure''
        \EndIf

  \\ 
    \Return Message polynomial vector $\f=(f_1,\dots,f_\intOrder) \in \SkewPolyringZeroDer_{<k}^\intOrder$

    \end{algorithmic}
\end{algorithm}

\begin{theorem}[Probabilistic Unique Decoding of \ac{LILRS} Codes]\label{thm:unique_dec_LILRS}
 Let $\rxSpaceVec \in \Grassm{\Nvec,\nReceiveVec}$ be a tuple of received subspaces of a transmission of a codeword $\txSpaceVec\in\liftedIntLinRS{\vecbeta,\a,\shots,\intOrder;\nTransmitVec,k}$ over random instance of the $\shots$-shot operator channel (see Remark~\ref{rem:rand_instance_op_channel}) with overall $\insertions$ insertions and $\deletions$ deletions.
 If the number of overall insertions $\insertions$ and deletions $\deletions$ satisfy 
 \begin{equation}
  \insertions+\intOrder \deletions\leq\intOrder(\nTransmit-k),
 \end{equation}
 then a the unique message polynomial vector $\f\in\SkewPolyringZeroDer_{<k}^\intOrder$ corresponding to the codeword $\txSpaceVec(\f)\in\liftedIntLinRS{\vecbeta,\a,\shots,\intOrder;\nTransmitVec,k}$ satisfying $\SumSubspaceDist(\txSpaceVec(\f),\rxSpaceVec)=\insertions+\deletions$ can be found with probability at least
 \begin{equation}\label{eq:lower_bound_Psucc_LILRS}
    \Pr(\text{success}) \geq 1 - \gammaq^{\ell+1} q^{-m(\insertionsmax-\insertions+1)}
 \end{equation}
 requiring at most $\softoh{\intOrder^\omega\OMul{\nReceive}}$ operations in $\Fqm$.
\end{theorem}

\begin{proof}
For the purpose of the proof (but not algorithmically), we consider the root-finding problem set up with an $\Fqm$-basis $Q^{(1)},\dots,Q^{(d_I)}$ of $\Qspace$.  
The unique decoder fails if there are at least two distinct roots $\f$ and $\f'$.
In this case, the $\Fqm$-linear system $\RFmat\cdot\RFvec=-\vec{q}_0$ in~\eqref{eq:rootFindingSystemFqmLILRS} set up with the $\Fqm$-basis $\widetilde{Q}^{(r)}\in\Qspace$ for $r=1,\dots,d_I$ has at least two solutions.
This means that $\RFmat \in \Fq^{\degConstraint d_I\times \intOrder k}$ must have rank $<\intOrder k$.

The matrix $\RFmat$ contains a lower block triangular matrix with matrices $\Q_0,\aut^{-1}(\Q_0),\dots,\aut^{-(k-1)}(\Q_0)$ on the upper diagonal, which have all $\Fqm$-rank $\rk_{q^m}(\Q_0)$ (see~\cite{bartz2022fast}).
Thus, if $\rk_{q^m}(\Q_0)=\intOrder$ the matrix $\RFmat$ has full $\Fqm$-rank $\intOrder k$. Therefore, $\rk_{q^m}(\RFmat)<\intOrder k$ implies that $\Q_0$ has rank $<s$.

Since the root-finding system~\eqref{eq:rootFindingSystemFqmLILRS} has at least one solution $\RFvec$, there is a vector $\f_0 \in \Fqm^\intOrder$ such that
\begin{equation*}
\Q_0 \f_0 = - \q_{0,0}^\top.
\end{equation*}
Thus, the matrix
\begin{equation*}
\Qbar_0 := 
\begin{pmatrix}
\Q_0 &  \q_{0,0}^\top
\end{pmatrix} \in \Fqm^{d_I\times (s+1)}
\end{equation*}
has rank $\rk_{q^m}(\Qbar_0) = \rk_{q^m}(\Q_0) < s$.
Hence, there are at least $d_I-s+1$ $\Fqm$-linearly independent polynomials $\widetilde{Q}^{(1)},\dots,\allowbreak\widetilde{Q}^{(d_I-s+1)}\in\Qspace$ such that their zeroth coefficients $\widetilde{q}_{l,0}^{\,(1)},\dots,\widetilde{q}_{l,0}^{\,(d_I-\intOrder+1)}$ are zero for all $l=0,\dots,\intOrder$ (obtained by suitable) $\Fqm$-linear combinations of the original basis polynomials $Q^{(1)},\dots,Q^{(d_I)}$, such that the corresponding $\Fqm$-linear row operations on $\Qbar_0$ give a $(d_I-\intOrder+1)\times(\intOrder+1)$ zero matrix (recall that $\Qbar_0$ has $d_I$ rows, but rank at most $s-1$).

The $d_I-s+1$ $\Fqm$-linearly independent coefficient vectors of $\widetilde{Q}^{(1)},\dots,\widetilde{Q}^{(d_I-s+1)}$ of the form~\eqref{eq:defIntVec_LILRS} are in the left kernel of the matrix
\begin{equation*}
\intMat^\top=
\begin{bmatrix}
\Lambda_{\degConstraint}(\vecxi)_\a \\
\Lambda_{\degConstraint-k+1}(\u_1)_\a \\
\vdots \\
\Lambda_{\degConstraint-k+1}(\u_\intOrder)_\a
\end{bmatrix}\in\Fqm^{\degConstraint(\intOrder+1)-\intOrder(k-1)\times \nReceive}.
\end{equation*}
Since the zeroth components $\widetilde{q}_{l,0}^{\,(r)}$ of all $\widetilde{Q}^{(r)}$ are zero for all $l=0,\dots,\intOrder$ and $r=1,\dots,d_I-\intOrder+1$, this means that the left kernel of the matrix
\begin{equation*}
\widetilde{\R}_I^\top=
\begin{bmatrix}
\op{\a}{\Lambda_{\degConstraint-1}(\vecxi)_\a} \\
\op{\a}{\Lambda_{\degConstraint-k}(\u_1)_\a} \\
\vdots \\
\op{\a}{\Lambda_{\degConstraint-k}(\u_\intOrder)_\a}
\end{bmatrix}\Fqm^{\degConstraint(\intOrder+1)-\intOrder k -1 \times n}
\end{equation*}
has dimension at least $d_I-s+1$. 
The maximum decoding region corresponds to the degree constraint $\degConstraint=\nReceive-\insertionsmax=\nTransmit-\deletions$ (see~\eqref{eq:decDegreeConstraintLILRS}) and thus 
\begin{align*}
 \dim(\lker(\widetilde{\R}_I^\top))
 &\geq\intOrder(\nTransmit-\deletions+1)-\intOrder k-\insertionsmax-\intOrder+1
 \geq1.
\end{align*}
Therefore, we have that
\begin{align*}
 \rk_{q^m}(\widetilde{\R}_I^\top)
 &\leq\degConstraint(\intOrder+1)-\intOrder k -1 -\dim(\lker(\widetilde{\R}_I^\top))
 \\ 
 &<(\nTransmit-\deletions)(\intOrder+1)-\intOrder k-1
 \\
 &=\nReceive-1. 
\end{align*}
Observe, that for $\degConstraint=\nTransmit-\deletions$ we have that
\begin{equation*}
  \widetilde{\R}_I^\top=\op{\a}{\LODecMat}
\end{equation*}
where $\LODecMat$ is the Loidreau--Overbeck decoding matrix from~\eqref{eq:LO_matrix_LILRS}.
By~\cite[Lemma~3]{bartz2022fast} the $\Fqm$-rank of $\LODecMat$ and $\op{\a}{\LODecMat}$ is the same and thus we have that
\begin{align*}
  \rk_{q^m}(\LODecMat)=\rk_{q^m}(\widetilde{\R}_I)< \nReceive-1
\end{align*}
which shows that in this case the {\LOlike} decoder fails as well. 
Therefore, we conclude that 
\begin{equation}
  \Pr(\rk_{q^m}(\RFmat)<\intOrder k)\leq\Pr(\rk_{q^m}(\Q_0)<\intOrder)\leq\Pr(\rk_{q^m}(\LODecMat)<\nReceive-1)
\end{equation}
and thus the lower bound on the probability of successful decoding follows from Theorem~\ref{thm:LO_decoder_LILRS}.
The complexity statement follows from~\cite[Corollary~1]{bartz2022fast} and~\cite[Corollary~2]{bartz2022fast}.
\end{proof}

The lower bound on the probability of successful decoding in~\eqref{eq:lower_bound_Psucc_LILRS} yields an upper bound on the decoding failure probability $P_f$, i.e. we have that
\begin{equation}\label{eq:upper_bound_Pfail_LILRS}
  P_f\leq\gammaq^{\ell+1} q^{-m(\insertionsmax-\insertions+1)}.
\end{equation}

The simulations results in Section~\ref{subsec:sim_results_LILRS} show that the upper bound on the decoding failure probability in~\eqref{eq:upper_bound_Pfail_LILRS} gives a good estimate of the performance of the probabilistic unique decoder. 


\subsection{Insertion/Deletion-Correction with the Complementary Code}\label{subsec:complementaryLILRS}
In~\cite[Section~4.4]{bartz2018efficient} it was shown, that the complementary of an interleaved (single-shot) subspace code is capable of correcting more deletions than insertions.
We will now briefly describe how to extend the concept from~\cite{bartz2018efficient} to the multishot scenario. 
In particular, we show that the complementary code of a \ac{LILRS} code is more resilient against \emph{deletions} than \emph{insertions}.
By using the arguments from~\cite[Lemma~14]{bartz2018efficient} and~\cite[Theorem~4]{bartz2018efficient} on each of the components of $\rxSpaceVec^\perp$ (and $\rxSpaceVec)$ we obtain the following result.

\begin{proposition}\label{prop:complementary_LILRS_decoding}
  Consider a \ac{LILRS} code $\mycode{C}=\liftedIntLinRS{\vecbeta,\a,\shots,\intOrder;\nTransmitVec,k}$ and the corresponding complementary code $\mycode{C}^\perp$.
  Suppose we transmit a tuple $\txSpaceVec^\perp\in\mycode{C}^\perp$ over a multishot operator channel~\eqref{eq:def:multishot_op_channel} with overall $\insertions$ insertions and $\deletions$ deletions and receive
  \begin{equation*}
    \rxSpaceVec^\perp=\delOp{N-\nTransmit-\deletions}(\txSpaceVec^\perp)\oplus\errSpaceVec_\insertions
  \end{equation*}
  where $\sumDim(\errSpaceVec_\insertions)=\insertions$. Then we have that
  \begin{align}
     \rxSpaceVec&=\delOp{\nTransmit-\insertions}(\txSpaceVec)\oplus\errSpaceVec_\deletions \label{eq:compl_rec_space}
  \end{align} 
  where $\sumDim(\errSpaceVec_\deletions)=\deletions$.
\end{proposition}

The proof of Proposition~\ref{prop:complementary_LILRS_decoding} can be found in Appendix~\ref{app:proof_complementary_LILRS_decoding}.

Proposition~\ref{prop:complementary_LILRS_decoding} shows, that the dual of the received tuple $\rxSpaceVec^\perp$, which is $\rxSpaceVec$, is a codeword $\txSpaceVec\in\mycode{C}$ that is corrupted by $\deletions$ insertions and $\insertions$ deletions.
Therefore we can use the decoder from Section~\ref{subsubsec:LILRS_list} on $\rxSpaceVec$ to perform list decoding of $\insertions$ insertions and $\deletions$ deletions up to 
\begin{equation}
  \deletions+\intOrder\insertions<\intOrder(\nTransmit-k+1)
\end{equation}
or the decoders from Sections~\ref{subsec:LO_dec_LILRS} \& \ref{subsubsec:LILRS_unique} to perform probabilistic unique decoding up to 
\begin{equation}
  \deletions+\intOrder\insertions\leq\intOrder(\nTransmit-k).
\end{equation}

The decoding steps can be summarized as follows:

\begin{enumerate}
  \item Transmit a tuple $\txSpaceVec^\perp\in\mycode{C}^\perp$ over a multishot operator channel with overall $\insertions$ insertions and $\deletions$ deletions.

  \item Compute the dual of the received tuple $(\rxSpaceVec^\perp)^\perp=\rxSpaceVec$.
  
  \item Use a decoder from Sections~\ref{subsubsec:LILRS_list},~\ref{subsec:LO_dec_LILRS} \& \ref{subsubsec:LILRS_unique} to recover the tuple $\txSpaceVec\in\mycode{C}$.
  
  \item Compute the dual tuple of $\txSpaceVec$ to obtain $\txSpaceVec^\perp$.
\end{enumerate}

\subsection{Simulation Results}\label{subsec:sim_results_LILRS}

In order to verify the upper bound on the decoding failure probability in~\eqref{eq:lower_bound_Psucc_LILRS} we performed a Monte Carlo simulation ($100$ errors) of a code $\liftedIntLinRS{\vecbeta,\a,\shots=2,\intOrder=3;\nTransmitVec=(3,3),k=3}$ over $\F_{3^3}$ over a multishot operator channel with overall $\deletions=1$ deletion and $\insertions\in\{4,5,6\}$ insertions.

The channel realization is chosen uniformly at random from all possible realizations of the multishot operator channel with exactly this number of deletions and insertions (see Remark~\ref{rem:rand_instance_op_channel}). 
The drawing procedure was implemented using the adapted dynamic-programming routine in Appendix~\ref{app:draw_uniform_errors}.

The results in Figure~\ref{fig:simLILRS} show, that the upper bound in~\eqref{eq:lower_bound_Psucc_LILRS} gives a good estimate of the decoding failure probability.
Although the heuristic upper bound from~\eqref{eq:heur_upper_bound_Pfail_LILRS} looks tighter for the considered parameters, it is \emph{not} a strict upper bound for the considered multishot operator channel.
\begin{figure}[ht!]
  \centering
%
%
\definecolor{mycolor1}{rgb}{1.00000,0.00000,1.00000}%
\definecolor{mycolor2}{rgb}{0.00000,1.00000,1.00000}%
\begin{tikzpicture}

\begin{axis}[%
xmin=4,
xmax=10,
xtick={4,5,6,7,8,9,10},
xlabel={$\insertions+\intOrder\deletions$},
compat=newest,
xmajorgrids,
ymode=log,
ymin=1e-08,
ymax=1,
yminorticks=true,
label style={anchor=near ticklabel, font=\footnotesize},
label style={inner sep=0}, 
ylabel={Decoding failure probability $P_{f}$},
ymajorgrids,
yminorgrids,
tick label style={font=\scriptsize},
legend style={at={(0.01,0.989)},anchor=north west,legend cell align=left,align=left,draw=white!15!black, font=\footnotesize},
mystyle]

\addplot [color=red!80!black,solid,
mark=x,mark options={solid}
]
  table[row sep=crcr]{%
5 3.96365494946758e-7\\
6 0.0000107018683635625\\
7 0.000288950445816186\\
8 0.00780166203703704\\ 
9 0.210644875000000\\ 
};
\addlegendentry{Upper bound on $P_f$ \eqref{eq:upper_bound_Pfail_LILRS}};

\addplot [color=cyan!80!black,dotted,
mark=o,mark options={solid}
]
  table[row sep=crcr]{%
5 1.24399719086618e-7\\  
6 3.35879241533867e-6\\ 
7 0.0000906873952141442\\ 
8 0.00244855967078189\\ 
9 0.0661111111111111\\
};
\addlegendentry{Heuristic upper bound on $P_f$ \eqref{eq:heur_upper_bound_Pfail_LILRS}};


\addplot [color=blue,
dashed,
mark=asterisk,
mark options={solid}
]
  table[row sep=crcr]{%
6 1.908361e-06\\
7 5.726977e-05\\
8 1.290956e-03\\ 
9 3.883495e-02\\ 
};
\addlegendentry{Simulation $\deletions=1$};

\end{axis}
\end{tikzpicture}%
  \caption{Result of a Monte Carlo simulation of the code $\liftedIntLinRS{\vecbeta,\a,\shots=2,\intOrder=3;\nTransmitVec=(3,3),k=3}$ over $\F_{3^3}$ transmitted over a multishot operator channel with overall $\deletions=1$ deletions and $\insertions=2,3,4,5,6$ insertions.}
  \label{fig:simLILRS}
\end{figure}
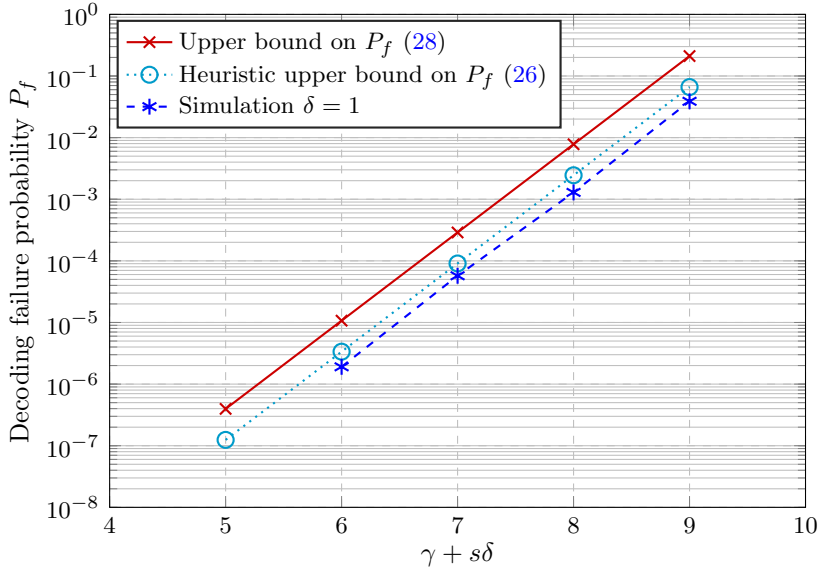
For the same parameters a (non-interleaved) lifted linearized Reed--Solomon code~\cite{martinez2019reliable} (i.e. $\intOrder=1$) can only correct $\insertions$ insertions and $\deletions$ deletions up to $\insertions+\deletions<4$.

\section{Conclusion}\label{sec:conclusion}
\acresetall

\subsection{Summary}\label{subsec:generality_remarks}
 We considered lifted $\intOrder$-interleaved linearized Reed--Solomon (\acs{LILRS})\acused{LILRS} codes for error-control in noncoherent multishot network coding and showed, that the relative overhead due to lifting can be reduced significantly compared to the construction by Mart{\'\i}nez-Pe{\~n}as--Kschischang.
 We proposed two decoding schemes for the multishot operator channel that are capable of correcting insertions and deletions beyond the unique decoding region in the sum-subspace metric.

 We proposed an efficient interpolation-based decoding scheme for \ac{LILRS} codes, which can be used as a list decoder or as a probabilistic unique decoder and can correct a total number of $\insertions$ insertions and $\deletions$ deletions up to $\insertions + \intOrder\deletions < \intOrder(\nTransmit-k+1)$ and $\insertions + \intOrder\deletions \leq \intOrder(\nTransmit-k)$, respectively, where $\intOrder$ is the interleaving order, $\nTransmit$ the sum of the dimensions of the transmitted spaces and $k$ the dimension of the code.
 We derived a {\LOlike} decoder for \ac{LILRS} codes, which provides arguments to upper bound on the decoding failure probability for the interpolation-based probabilistic unique decoder.

 We showed how to construct and decode lifted $\intOrder$-\emph{interleaved} linearized Reed--Solomon codes for error control in random linear multishot network coding.
 Compared to the construction by Mart{\'\i}nez-Pe{\~n}as--Kschischang, interleaving allows to increase the decoding region significantly (especially w.r.t.\ the number of insertions) and decreases the overhead due to the lifting (i.e., increases the code rate), at the cost of an increased packet size.

 Up to our knowledge, the proposed decoding schemes are the first being able to correct errors beyond the unique decoding region in the sum-subspace metric efficiently.
 The tightness of the upper bounds on the decoding failure probability of the proposed decoding schemes for \ac{LILRS} codes were validated via Monte Carlo simulations.

\subsection{Remarks on Generality}

In this paper, we considered codes constructed by skew polynomials with zero derivations, i.e. polynomials from~$\SkewPolyringZeroDer$, only.
The main reason for this is that for operations in $\SkewPolyringZeroDer$ we can give complexity bounds, which are of interest in the implementation point of view.
However, the complexity analysis has to be performed w.r.t. this setup (computational complexity may be larger).

We considered decoding of \emph{homogeneous} \ac{LILRS} codes, respectively, i.e. interleaved codes where the component codes have the same code dimension.
We consider these simpler code classes in order to not further complicate the quite involved notation.
The decoding schemes proposed in this paper can be generalized to \emph{heterogeneous} interleaved codes, where each component code may have a different dimension, in a straight-forward manner like e.g. in~\cite{wachter2013decoding,bartz2017algebraic}.
Denote by $k_1,\dots,k_\intOrder$ the dimensions of the component codes and define $\bar{k}\defeq\frac{1}{\intOrder}\sum_{l=1}^{\intOrder}k_l$. 
The resulting decoding regions are then $\insertions+\intOrder\deletions<\intOrder(\nTransmit-\bar{k}+1)$ for list decoding and $\insertions+\intOrder\deletions\leq\intOrder(\nTransmit-\bar{k})$ for probabilistic unique decoding.

\subsection{Outlook \& Future Work}

For future work it would be interesting to see of the considered concepts applied to interleaved \ac{LILRS} codes that are based on the construction of \ac{LRS} codes over smaller fields.

So far, no results on the list-decodability of random sum-subspace-metric codes, like e.g. for single-shot subspace codes~\cite{ding2014list}, are available.
Once such results are available it would be interesting to compare the list-decodability of random sum-subspace-metric codes with constructive results proposed in this paper. 


\section{Acknowledgments}
H. Bartz acknowledges the financial support by the Federal Ministry of Education and Research of Germany in the programme of “Souverän. Digital. Vernetzt.” Joint project 6G-RIC, project identification number: 16KISK022.


\begin{appendices}

\section{Proofs from Section~\ref{subsec:LO_dec_LILRS}}

\subsection{Proof of Lemma~\ref{lem:properties_LODecMat_LILRS}}\label{app:proof_properties_LODecMat_LILRS}

\begin{proof}
  First observe, that every received space $\rxSpaceShot{i}$ can be represented by a matrix $\widetilde{\U}^{(i)}$ of the form
     \begin{equation}
       \renewcommand{\arraystretch}{1.5}
       \widetilde{\U}^{(i)}=
       \left(
       \begin{array}{c|ccc}
        \vecxi_1^{(i)\top} & \vec{x}_1^{(i)\top} & \dots & \vec{x}_\intOrder^{(i)\top} 
        \\ \hline
        \vecxi_2^{(i)\top} & \widetilde{\vec{u}}_1^{(i)\top} & \dots & \widetilde{\vec{u}}_\intOrder^{(i)\top} 
        \\ \hline
        \vec{0} & \hat{\vec{e}}_1^{(i)\top} & \dots & \hat{\vec{e}}_\intOrder^{(i)\top}
       \end{array}
       \right)\in\Fqm^{\nReceiveShot{i}\times (\intOrder+1)}
     \end{equation}
     such that first $\nTransmitShot{i}-\deletionsShot{i}$ rows of $\widetilde{\U}^{(i)}$ form a basis for the non-corrupted received space $\rxSpaceShot{i}\cap\txSpaceShot{i}$ whereas the last $\rankErrShot{i}+\deviationsShot{i}=\insertionsShot{i}$ rows correspond to the insertions for all $i=1,\dots,\shots$, i.e. we have that
     \begin{align}
       \Rowspace{\left(
       \begin{array}{c|ccc}
        \vecxi_1^{(i)\top} & \vec{x}_1^{(i)\top} & \dots & \vec{x}_\intOrder^{(i)\top}
       \end{array}
       \right)}
       &=
      \rxSpaceShot{i}\cap\txSpaceShot{i},\label{eq:non_corr_comp_space}
      \\[5pt]
      \renewcommand{\arraystretch}{1.5}
      \Rowspace{\left(
       \begin{array}{c|ccc}
        \vecxi_2^{(i)\top} & \widetilde{\vec{u}}_1^{(i)\top} & \dots & \widetilde{\vec{u}}_\intOrder^{(i)\top} 
        \\ \hline
        \vec{0} & \hat{\vec{e}}_1^{(i)\top} & \dots & \hat{\vec{e}}_\intOrder^{(i)\top}
       \end{array}
       \right)}
       &=
      \errSpaceShot{i}\label{eq:err_comp_space}
     \end{align}
     for all $i=1,\dots,\shots$ where $\vecxi_1^{(i)},\vec{x}_j^{(i)}\in\Fqm^{\nTransmitShot{i}-\deletionsShot{i}}$, $\vecxi_2^{(i)},\widetilde{\vec{u}}_j^{(i)}\in\Fqm^{\rankErrShot{i}}$ and $\hat{\E}^{(i)}=
      \begin{pmatrix}
        \hat{\vec{e}}_1^{(i)\top} & \dots & \hat{\vec{e}}_\intOrder^{(i)\top}
       \end{pmatrix}
       \in\Fqm^{\deviationsShot{i}\times \intOrder}$.

     Since the rows of the matrices $\U^{(i)}$ also form a basis of $\rxSpaceShot{i}$, there exist invertible matrices $\W^{(i)}\in\Fq^{\nReceiveShot{i}}$ such that
     \begin{equation}
       \left(\W^{(i)}\right)^\top\U^{(i)}=\widetilde{\U}^{(i)},\quad\forall i=1,\dots,\shots.
     \end{equation}
     Since $\rxSpaceShot{i}\cap\txSpaceShot{i}$ is an $\nTransmitShot{i}-\deletionsShot{i}$-dimensional subspace $\txSpaceShot{i}$, we have by the definition of \ac{LILRS} codes that $\rk_{q}(\vecxi_1^{(i)})=\nTransmitShot{i}-\deletionsShot{i}$ for all $i=1,\dots,\shots$.
     We also have that $\rk_{q}(\hat{\E}^{(i)})=\deviationsShot{i}$ for all $i=1,\dots,\shots$ since otherwise the matrix $\widetilde{\U}^{(i)}$ would not have full $\Fq$-rank $\nReceiveShot{i}$, which contradicts that it is a basis for $\rxSpaceShot{i}$ that is obtained from $\U^{(i)}$ via $\Fq$-elementary row operations.
     Since $\txSpaceShot{i}\cap\errSpaceShot{i}=\{\0\}$, we have that $\widetilde{\vec{u}}_l^{(i)}\neq\opev{f_1}{\vecxi_2^{(i)}}{a_i}$ for all $i=1,\dots,\shots$ and $l=1,\dots,\intOrder$, since otherwise there would exist a vector such that, for some $j$, we have
     \begin{equation*}
       \Rowspace{(\xi_{2,j}^{(i)}\mid\opev{f_1}{\xi_{2,j}^{(i)}}{a_i}\dots\opev{f_1}{\xi_{2,j}^{(i)}}{a_i})}\in\txSpaceShot{i}
     \end{equation*}
     which contradicts that $\txSpaceShot{i}\cap\errSpaceShot{i}=\{\0\}$.
     We now show that
     \begin{align*}
       \rk_{q}(\vecxi_2^{(i)})=\rankErrShot{i},
       \quad\text{and}\quad
       \rk_{q}((\vecxi_1^{(i)},\vecxi_2^{(i)}))=\nTransmitShot{i}-\deletionsShot{i}+\rankErrShot{i},
       \quad\forall i=1,\dots,\shots.
     \end{align*} 
     Suppose that $\rk_{q}(\vecxi_2^{(i)})<\rankErrShot{i}$.
     Then there exists an $\Fq$-linear combination of the rows of the matrix in~\eqref{eq:err_comp_space} such that the first entry of the vector becomes zero.
     In this case it would belong to the lower part of the matrix having a leading zero.
     If $\rk_{q}((\vecxi_1^{(i)},\vecxi_2^{(i)}))<\nTransmitShot{i}-\deletionsShot{i}+\rankErrShot{i}$, then there exist an element $\xi_{2,j}$ such that $\Rowspace{\xi_{2,j}}\subseteq\Rowspace{\vecxi_{1}^\top}$.
     In this case we can subtract the corresponding $\Fq$-linear combination from the row having $\xi_{2,j}$ as first entry, which again yields a row with leading zero element.

     Using the matrices $\widetilde{\U}^{(i)}$ for all $i=1,\dots,\shots$ and performing $\Fqm$-linear row operations, we can set up a matrix of the form
     \begin{equation}
       \bar{\LODecMat}=\left(\LODecMat^{(1)},\dots,\LODecMat^{(\shots)}\right)\defeq
       \left(
       \begin{array}{c}
        \Lambda_{\nTransmit-\deletions-1}(\bar{\vecxi})_{\vec{a}}
        \\ 
        \Lambda_{\nTransmit-\deletions-k}(\bar{\u}_{1})_{\vec{a}} 
        \\ 
        \vdots 
        \\ 
        \Lambda_{\nTransmit-\deletions-k}(\bar{\u}_{\intOrder})_{\vec{a}}
       \end{array}
      \right)
      \in \Fqm^{((\intOrder+1)(\nTransmit-\deletions)-\intOrder k -1) \times \nReceive}
     \end{equation}
     where
     \begin{align*}
      \bar{\vecxi}&=\left((\vecxi_1^{(1)}\mid\vecxi_2^{(1)}\mid \0)\mid\dots\mid(\vecxi_1^{(\shots)}\mid\vecxi_2^{(\shots)}\mid \0)\right),
      \\ 
      \bar{\u}_{l}&=\left((\0\mid\widetilde{\e}_l^{(1)}\mid\hat{\e}_l^{(1)})\mid\dots\mid(\0\mid\widetilde{\e}_l^{(\shots)}\mid\hat{\e}_l^{(\shots)})\right), \quad\forall l=1,\dots,\intOrder.
     \end{align*}
     The component matrices are then of the form
     \begin{equation*}
       \bar{\LODecMat}^{(i)}=
       \begin{pmatrix}
        \opVandermonde{\nTransmit-\deletions-1}{\bar{\vecxi}_1^{(i)}}{a_i} & \opVandermonde{\nTransmit-\deletions-1}{\bar{\vecxi}_2^{(i)}}{a_i} & \0
        \\ 
        \0 & \opVandermonde{\nTransmit-\deletions-k}{\widetilde{\e}_1^{(i)}}{a_i} & \opVandermonde{\nTransmit-\deletions-k}{\hat{\e}_1^{(i)}}{a_i}
        \\
        \vdots & \vdots & \vdots
        \\ 
        \0 & \opVandermonde{\nTransmit-\deletions-k}{\widetilde{\e}_\intOrder^{(i)}}{a_i} & \opVandermonde{\nTransmit-\deletions-k}{\hat{\e}_\intOrder^{(i)}}{a_i}
       \end{pmatrix}
       \in \Fqm^{((\intOrder+1)(\nTransmit-\deletions)-\intOrder k -1) \times \nReceiveShot{i}}
     \end{equation*}
     where
     $\x^{(i)}_l-\opev{f_l}{\vecxi_1^{(i)}}{a_i}=\0$ (yielding the zero matrices in the leftmost block) and 
     $\widetilde{\e}_l^{(i)}=\widetilde{\u}_l^{(i)}-\opev{f_l}{\vecxi_2^{(i)}}{a_i}$ for all $l=1,\dots,\intOrder$ and $i=1,\dots,\shots$. 
     Define the matrices
     \begin{equation}\label{eq:def_E_tilde_i}
      \widetilde{\E}^{(i)}=
      \begin{pmatrix}
       \widetilde{\e}_1^{(i)}
       \\ 
       \vdots
       \\ 
       \widetilde{\e}_\intOrder^{(i)}
      \end{pmatrix}\in\Fqm^{\intOrder\times \rankErr^{(i)}},
      \qquad
      \hat{\E}^{(i)}=
      \begin{pmatrix}
       \hat{\e}_1^{(i)}
       \\ 
       \vdots
       \\ 
       \hat{\e}_\intOrder^{(i)}
      \end{pmatrix}\in\Fqm^{\intOrder\times \deviationsShot{i}}
     \end{equation}
     for all $i=1,\dots,\shots$.
     Suppose that the $\Fq$-rank of $\widetilde{\E}^{(i)}$ is less than $\rankErr^{(i)}$.
     Then there exist a nontrivial $\Fq$-linear combination of the columns of the matrix 
     \begin{equation*}
       \begin{pmatrix}
        \opVandermonde{\nTransmit-\deletions-1}{\bar{\vecxi}_2^{(i)}}{a_i}
        \\ 
        \opVandermonde{\nTransmit-\deletions-k}{\widetilde{\e}_1^{(i)}}{a_i}
        \\
        \vdots 
        \\ 
        \opVandermonde{\nTransmit-\deletions-k}{\widetilde{\e}_\intOrder^{(i)}}{a_i}
       \end{pmatrix}
       \in \Fqm^{((\intOrder+1)(\nTransmit-\deletions)-\intOrder k -1) \times \rankErr^{(i)}}
     \end{equation*}
     such that the first (upper) $(\nTransmit-\deletions-1)$ rows are nonzero and the last (lower) $\intOrder(\nTransmit-\deletions-k)$ rows are all zero.
     This contradicts that $\txSpaceShot{i}\cap\errSpaceShot{i}=\{\0\}$ since each $\bar{\LODecMat}^{(i)}$ is obtained by subtracting the evaluations of each $f_l$ at the corresponding values in $\vecxi_1^{(i)}$ and $\vecxi_2^{(i)}$ (and the corresponding row-operator powers thereof).
     Now assume that $\rk_q(\widetilde{\E}^{(i)}\mid\hat{\E}^{(i)})<\insertionsShot{i}$ for some $i\in\{1,\dots,\shots\}$.
     Then there exist nontrivial $\Fq$-linear combinations of the $\rankErr^{(i)}+\deviationsShot{i}=\insertionsShot{i}$ rightmost columns of $\bar{\LODecMat}^{(i)}$ such that the $\intOrder(\nTransmit-\deletions-k)$ rows are all zero and the upper $\nTransmit-\deletions-1$ rows are nonzero. 
     This contradicts that $\txSpaceShot{i}\cap\errSpaceShot{i}=\{\0\}$ by the same argument as above.
     
     Since $\bar{\LODecMat}$ is obtained from $\LODecMat$ via $\Fqm$-elementary row operations and $\Fq$-elementary column operations, we have that
     \begin{equation}
       \rk_{q^m}(\LODecMat)=\rk_{q^m}(\bar{\LODecMat}).
     \end{equation}
\end{proof}

\subsection{Proof of Lemma~\ref{lem:right_kernel_prop_LILRS}}\label{app:proof_right_kernel_prop_LILRS}
\begin{proof}
\begin{itemize}
  \item[--] Ad~\ref{itm:LO_LILRS_lemma_rank_L}):
  By column permutations, i.e. by rearranging the columns of $\bar{\LODecMat}$ according to the vectors
     \begin{align*}
       \vecxi_1&\defeq\left(\vecxi_1^{(1)} \mid \vecxi_1^{(2)} \mid \dots \mid \vecxi_1^{(\shots)}\right)\in\Fqm^{\nTransmit-\deletions},
       \quad&
       \vecxi_2&\defeq\left(\vecxi_2^{(1)} \mid \vecxi_2^{(2)} \mid \dots \mid \vecxi_2^{(\shots)}\right)\in\Fqm^{\rankErr},
       \\ 
       \widetilde{\e}_l&\defeq\left(\widetilde{\e}_l^{(1)} \mid \widetilde{\e}_l^{(2)} \mid \dots \mid \widetilde{\e}_l^{(\shots)}\right)\in\Fqm^{\rankErr},
       &
       \hat{\vec{e}}_l&\defeq\left(\hat{\vec{e}}_l^{(1)} \mid \hat{\vec{e}}_l^{(2)} \mid \dots \mid \hat{\vec{e}}_l^{(\shots)}\right)\in\Fqm^{\deviations},
     \end{align*}
     we obtain another (equivalent) matrix
     \begin{equation}
     \widetilde{\LODecMat}
     \defeq
      \left(
       \begin{array}{c|c|c}
        \Lambda_{\nTransmit-\deletions-1}(\vecxi_1)_{\vec{a}} & \Lambda_{\nTransmit-\deletions-1}(\vecxi_2)_{\vec{a}} & \vec{0}
        \\ 
        \vec{0} & \Lambda_{\nTransmit-\deletions-k}(\widetilde{\E})_{\vec{a}} & \Lambda_{\nTransmit-\deletions-k}(\hat{\E})_{\vec{a}}
       \end{array}
      \right)
     \end{equation}
     where 
     \begin{equation*}
       \widetilde{\E}
       =
       \left(\widetilde{\E}^{(1)}\mid\dots\mid\widetilde{\E}^{(\shots)}\right)\in\Fqm^{\intOrder\times\rankErr}
       \quad\text{and}\quad 
       \hat{\E}=
       \left(\hat{\E}^{(1)}\mid\dots\mid\hat{\E}^{(\shots)}\right)\in\Fqm^{\intOrder\times\deviations}
     \end{equation*}
     such that $\rk_{q^m}(\LODecMat)=\rk_{q^m}(\bar{\LODecMat})=\rk_{q^m}(\widetilde{\LODecMat})$.
     By Lemma~\ref{lem:properties_LODecMat_LILRS} we have that
     \begin{equation}
       \SumRankWeight(\widetilde{\E})=\rankErr,\quad\SumRankWeight(\hat{\E})=\deviations
     \end{equation}
     and
     \begin{equation}
       \SumRankWeight\left((\shot{\widetilde{\E}}{1}\mid\shot{\hat{\E}}{1})\mid\dots\mid (\shot{\widetilde{\E}}{\shots}\mid\shot{\hat{\E}}{\shots})\right)=t+\deviations=\insertions.
     \end{equation}
    Defining the matrix
    \begin{equation*}
      \Z\defeq
      \begin{pmatrix}
       \Lambda_{\nTransmit-\deletions-k}(\widetilde{\E})_{\vec{a}} & \Lambda_{\nTransmit-\deletions-k}(\hat{\E})_{\vec{a}}
      \end{pmatrix}
      \in\Fqm^{\intOrder(\nTransmit-\deletions-k)\times\insertions}
    \end{equation*}
    we can write $\widetilde{\LODecMat}$ as a block matrix of the form
    \begin{equation*}
      \widetilde{\LODecMat}=
      \begin{pmatrix}
        \Lambda_{\nTransmit-\deletions-1}(\vecxi_1)_{\vec{a}} & \star \\
        \0 & \Z
      \end{pmatrix},
    \end{equation*}
    where $\star$ denotes an arbitrary matrix, and thus we have that $\rk_{q^m}(\widetilde{\LODecMat})=\rk_{q^m}(\Lambda_{\nTransmit-\deletions-1}(\vecxi_1)_{\vec{a}})+\rk_{q^m}(\Z)$.
    By Lemma~\ref{lem:properties_LODecMat_LILRS} the vector $\vecxi_1\in\Fqm^{\nTransmit-\deletions}$ has sum-rank $\nTransmit-\deletions$.
    Recall that we have 
    \begin{equation*}
      \rk_{q^m}\left(\Lambda_{\nTransmit-\deletions-1}(\vecxi_1)_{\vec{a}}\right)=\min\{\nTransmit-\deletions-1,\nTransmit-\deletions\}=\nTransmit-\deletions-1.
    \end{equation*}
    The matrix $\Z$ can be obtained from $\bar{\Z}$ via column permutations, which implies that $\rk_{q^m}(\Z)=\rk_{q^m}(\bar{\Z})$.
    Since the $\Fqm$-rank of $\bar{\Z}$ (and thus also of $\Z$) equals $\insertions$, we conclude by the upper block triangular structure of $\widetilde{\LODecMat}$ that 
    \begin{equation*}
      \rk_{q^m}(\widetilde{\LODecMat})=\rk_{q^m}(\LODecMat)=\nTransmit-\deletions-1+\insertions=\nReceive-1.
    \end{equation*}

  \item[--] Ad~\ref{itm:LO_LILRS_lemma_rank_of_kernel_element}):
    Let $\bar{\h}=\left(\h^{(1)}\mid\h^{(2)}\mid\dots\mid\h^{(\shots)}\right)\in\rker(\bar{\LODecMat})\setminus\{\0\}$ be a nonzero element in the right kernel of $\bar{\LODecMat}$.
    If the $\Fqm$-rank of $\bar{\Z}$ equals $\insertions$, we have that
    \begin{equation*}
     \rk_{q^m}(\bar{\Z}^{(i)})=\insertionsShot{i}
    \end{equation*}
    for all $i=1,\dots,\shots$.
    Thus, for any $\bar{\h}\in\rker(\bar{\LODecMat})\setminus\{\0\}$ we must have that the last (rightmost) $\insertionsShot{i}$ entries of each $\bar{\h}^{(i)}$ must be zero, which implies that 
    \begin{equation}\label{eq:min_rank_h_i_bar}
      \rk_{q}(\bar{\h}^{(i)})\leq\nTransmitShot{i}-\deletionsShot{i},\quad \forall i=1,\dots,\shots.
    \end{equation}

    Let $\mycode{C}$ be the $[\nTransmit-\deletions+t,\nTransmit-\deletions-1]$ code generated by $(\Lambda_{\nTransmit-\deletions-1}(\vecxi_1)_{\vec{a}} \mid \Lambda_{\nTransmit-\deletions-1}(\vecxi_2)_{\vec{a}})$ and let $\mycode{C}'$ be the $[\nReceive,\nTransmit-\deletions-1]$ code generated by $(\Lambda_{\nTransmit-\deletions-1}(\vecxi_1)_{\vec{a}} \mid \Lambda_{\nTransmit-\deletions-1}(\vecxi_2)_{\vec{a}} \mid \vec{0})$.
    Since $(\Lambda_{\nTransmit-\deletions-1}(\vecxi_1)_{\vec{a}} \mid \Lambda_{\nTransmit-\deletions-1}(\vecxi_2)_{\vec{a}})$ is a generator matrix of an $[\nTransmit-\deletions+t,\nTransmit-\deletions-1]$ \ac{LRS} code (up to column permutations), we have that the dual code $\mycode{C}^\perp$ is equivalent to an $\ac{LRS}$ code with minimum sum-rank distance $\nTransmit-\deletions$. 
    The dual code of $\mycode{C}'$ is an $[\nReceive, \insertions+1]$ code $\mycode{C}'^\perp$ that can be decomposed into
     \begin{align*}
       \mycode{C}_1'^\perp&=\{(\vec{c}\mid \vec{0}):\vec{c}\in\mycode{C}^\perp\}
       \\ 
       \mycode{C}_2'^\perp&=\{(\vec{0}\mid \vec{\widetilde{\vec{c}}}):\vec{\widetilde{\vec{c}}}\in\Fqm^{\deviations}\}
     \end{align*}
     such that $\mycode{C}'^\perp=\mycode{C}_1'^\perp+\mycode{C}_2'^\perp$ and $\mycode{C}_1'^\perp\cap\mycode{C}_2'^\perp=\{\0\}$.
     Therefore, we have that the minimum distance of the dual code $\mycode{C}'^\perp$ is $\nTransmit-\deletions$.

     The first $\nTransmit-\deletions-1$ rows of $\bar{\LODecMat}$ span code $\bar{\mycode{C}}$ which is equivalent to $\mycode{C}'$ and thus we have that the minimum sum-rank distance of the dual code $\bar{\mycode{C}}^\perp$ equals $\nTransmit-\deletions$.
     Since any nonzero element $\bar{\h}$ in the right kernel of $\bar{\LODecMat}$ is contained in $\bar{\mycode{C}}^\perp$, we have that
     \begin{equation}\label{eq:min_weight_h_bar}
       \SumRankWeight(\bar{\h})\geq\nTransmit-\deletions.
     \end{equation}
     Therefore,~\eqref{eq:min_rank_h_i_bar} and~\eqref{eq:min_weight_h_bar} are satisfied only if 
     \begin{equation}
      \rk_{q}(\bar{\h}^{(i)})=\nTransmitShot{i}-\deletionsShot{i},\quad \forall i=1,\dots,\shots.
     \end{equation}
     This proves the statement, since by Lemma~\ref{lem:properties_LODecMat_LILRS} we have that $\rker(\LODecMat)=\rker(\bar{\LODecMat}\W^{-1})$ and thus
     \begin{equation*}
       \bar{\h}\in\rker(\bar{\LODecMat})
       \quad\Longleftrightarrow\quad
       \bar{\h}\left(\W^{-1}\right)^\top\in\rker(\LODecMat).
     \end{equation*}

  \item[--] Ad~\ref{itm:LO_LILRS_lemma_Ti_matrices}):
    Expand $\h^{(i)}$ into an $m \times \nReceiveShot{i}$ matrix over $\Fq$, which has $\Fq$-rank $\nTransmitShot{i}-\deletionsShot{i}$ by \ref{itm:LO_LILRS_lemma_rank_of_kernel_element}). 
    Then, we can perform elementary column operations on this matrix to bring it into reduced column echelon form, where the $\nTransmitShot{i}-\deletionsShot{i}$ nonzero columns are the leftmost ones. 
    The matrix $\T^{(i)}\in\Fq^{\nReceiveShot{i}\times\nReceiveShot{i}}$ is then chosen to be the matrix that, by multiplication from the right, performs the used sequence of elementary column operations. 
    Note that the $\nTransmitShot{i}-\deletionsShot{i}$ nonzero entries of $\h^{(i)}\T^{(i)}$ are linearly independent over $\Fq$.

  \item[--] Ad~\ref{itm:LO_LILRS_lemma_EDi_zero}): 
    The matrices $\T^{(i)} \in \Fq^{\nReceiveShot{i} \times \nReceiveShot{i}}$ are chosen such that the rightmost $\insertionsShot{i}$ positions of $\h^{(i)} \T^{(i)}$ are zero.
    Recall from~\ref{itm:LO_LILRS_lemma_rank_of_kernel_element}), that if $\rk_{q^m}(\bar{\Z})=\insertions$, we have that for any nonzero element $\bar{\h}=(\bar{\h}^{(1)}\mid\dots\mid\bar{\h}^{(\shots)})$ in the right kernel of $\bar{\LODecMat}$ the $\insertionsShot{i}$ rightmost positions of each $\bar{\h}^{(i)}$ are zero.
    Define the matrices $\D^{(i)}=\left(\T^{(i)-1}\right)^\top$ for all $i=1,\dots,\shots$.
    Then we have that 
    \begin{equation}
      \h\T\in\rker(\LODecMat\cdot\diag(\D^{(1)},\dots,\D^{(\shots)})).
    \end{equation}
    We will now show that the column span of the matrices
    \begin{equation}
      \left[\bar{\LODecMat}^{(i)}\right]_{\{\nTransmitShot{i}-\deletionsShot{i}+1,\dots,\nReceiveShot{i}\}}
      \quad\text{and}\quad
      \left[\LODecMat^{(i)}\D^{(i)}\right]_{\{\nTransmitShot{i}-\deletionsShot{i}+1,\dots,\nReceiveShot{i}\}}
    \end{equation}
    must be the same.
    Since $\rk_{q^m}(\bar{\Z})=\insertions$ implies $\rk_{q^m}(\bar{\Z}^{(i)})=\insertionsShot{i}$ for all $i=1,\dots,\shots$, the $\intOrder(\nTransmit-\deletions-k)$ last (lower) rows of $\LODecMat^{(i)}\D^{(i)}$ have $\Fqm$-rank $\insertionsShot{i}$ for all $i=1,\dots,\shots$.
    Now assume that the column span of the matrices above are not the same for some $i\in\{1,\dots,\shots\}$.
    Then there exist $\Fqm$-linearly independent columns for some index $j\in\{1,\dots,\nTransmitShot{i}-\deletionsShot{i}\}$, which contradicts that $\h\T$ (which has the $\insertionsShot{i}$ rightmost entries per block equal to zero) is in the right kernel of $\LODecMat\cdot\diag(\D^{(1)},\dots,\D^{(\shots)})$.
    Therefore, we must have that the column space is the same.

    Since the $\insertionsShot{i}$ rightmost columns of each $\bar{\LODecMat}^{(i)}$ correspond to the insertions, we have that the rightmost $\insertionsShot{i}$ columns of $\LODecMat^{(i)}\D^{(i)}$ also correspond to the insertions.
    Since $\LODecMat$ is set up using the matrices $\U^{(i)}$ for all $i=1,\dots,\shots$, we have that the last (lower) $\insertionsShot{i}$ rows of $\left(\T^{(i)}\right)^{-1}\U^{(i)}$ span the error space $\errSpaceShot{i}$. 
    This implies that the first $\nTransmitShot{i}-\deletionsShot{i}$ rows of $\left(\T^{(i)}\right)^{-1}\U^{(i)}$ form a basis of the noncorrupted space $\rxSpaceShot{i}\cap\txSpaceShot{i}$, since by the definition of the operator channel we have that $\txSpaceShot{i}\cap\errSpaceShot{i}=\{\0\}$ for all $i=1,\dots,\shots$.

  \item[--] Ad~\ref{itm:LO_LILRS_reconstruct_polynomials}):
    By~\ref{itm:LO_LILRS_lemma_EDi_zero}) the first $\nTransmitShot{i}-\deletionsShot{i}$ rows of  
     \begin{equation}
      \hat{\U}^{(i)}=\left(\T^{(i)}\right)^{-1}\U^{(i)}
      =
       \begin{pmatrix}
        \hat{\vecxi}^{(i)\top} & \hat{\u}_1^{(i)\top} & \dots & \hat{\u}_\intOrder^{(i)\top} 
       \end{pmatrix}
       \in\Fqm^{\nReceiveShot{i}\times (\intOrder+1)}.
     \end{equation}
     form a basis for the noncorrupted subspace $\rxSpaceShot{i}\cap\txSpaceShot{i}$, i.e. we have that
     \begin{equation*}
       \hat{u}^{(i)}_{l,\mu}=\opev{f_l}{\hat{\xi}^{(i)}}{a_i},
       \quad\forall\mu=1,\dots,\nTransmitShot{i}-\deletionsShot{i},\,\forall l=1,\dots,\intOrder,\,\forall i=1,\dots,\shots.
     \end{equation*}
     Therefore, the message polynomials $f_1,\dots,f_\intOrder \in \SkewPolyringZeroDer_{<k}$ can be reconstructed via Lagrange interpolation (cf. Lemma~\ref{lem:int_poly_op_ev}).
\end{itemize}
\end{proof}

\section{Proofs from Section~\ref{subsec:complementaryLILRS}}
\subsection{Proof of Proposition~\ref{prop:complementary_LILRS_decoding}}\label{app:proof_complementary_LILRS_decoding}
\begin{proof}
The proof proceed similar as the proof of~\cite[Theorem~4]{bartz2018efficient}.
At the output of the operator channel we have
 \begin{equation*}
    \rxSpaceVec^\perp=\delOp{N-\nTransmit-\deletions}(\txSpaceVec^\perp)\oplus\errSpaceVec_\insertions
  \end{equation*}
  where $\sumDim(\errSpaceVec_\insertions)=\insertions$. 
  Let $\insertionsVec$ and $\deletionsVec$ denote the partition of the insertions and deletions, respectively.
  
  For all $i=1,\dots,\shots$ we have that 
  \begin{align*}
    \rxSpaceShot{i}=\left((\rxSpaceShot{i})^\perp\right)^\perp
    &=\left(\delOp{N_i-\nTransmitShot{i}-\deletionsShot{i}}\left({(\txSpaceShot{i})}^\perp\right)\oplus\errSpaceShot{i}_{\insertionsShot{i}}\right)^\perp
    \\ 
    &=\left(\delOp{N_i-\nTransmitShot{i}-\deletionsShot{i}}\left({(\txSpaceShot{i})}^\perp\right)\right)^\perp\cap(\errSpaceShot{i}_{\insertionsShot{i}})^\perp.
  \end{align*}
  By~\cite[Lemma~14]{bartz2018efficient} we have that
  \begin{equation}\label{eq:compl_erasure_op}
    \left(\delOp{N_i-\nTransmitShot{i}-\deletionsShot{i}}\left({(\txSpaceShot{i})}^\perp\right)\right)^\perp=\txSpaceShot{i}\oplus\errSpaceShot{i}_{\deletionsShot{i}}
  \end{equation}
  where $\dim(\errSpaceShot{i}_{\deletionsShot{i}})=\deletionsShot{i}$ and therefore 
  \begin{align}
    \rxSpaceShot{i}&=\left(\txSpaceShot{i}\oplus\errSpaceShot{i}_{\deletionsShot{i}}\right)\cap(\errSpaceShot{i}_{\insertionsShot{i}})^\perp \notag
    \\ 
    &=\left(\txSpaceShot{i}\cap(\errSpaceShot{i}_{\insertionsShot{i}})^\perp\right)\oplus\errSpaceShot{i}_{\deletionsShot{i}} \label{eq:compl_erasure_op_2}
  \end{align}
  for all $i=1,\dots,\shots$ since $\errSpaceShot{i}_{\deletionsShot{i}}\cap(\errSpaceShot{i}_{\insertionsShot{i}})^\perp=\errSpaceShot{i}_{\deletionsShot{i}}$.
  From the operator channel we have that
  \begin{equation*}
     (\txSpaceShot{i})^\perp\cap\errSpaceShot{i}_{\insertionsShot{i}}=0
     \quad\Longrightarrow\quad
     (\txSpaceShot{i})^\perp\subseteq(\errSpaceShot{i}_{\insertionsShot{i}})^\perp
  \end{equation*} 
  since $(\txSpaceShot{i})^\perp\oplus\txSpaceShot{i}=\Fq^{N_i}$ and $\dim((\txSpaceShot{i})^\perp)\leq \dim((\errSpaceShot{i}_{\insertionsShot{i}})^\perp)$.
  From~\eqref{eq:compl_erasure_op} we get $\txSpaceShot{i}\cap\errSpaceShot{i}_{\deletionsShot{i}}=0$ and therefore
  \begin{equation*}
    \errSpaceShot{i}_{\deletionsShot{i}}\subseteq(\txSpaceShot{i})^\perp\subseteq(\errSpaceShot{i}_{\insertionsShot{i}})^\perp
  \end{equation*}
  for all $i=1,\dots,\shots$.
  We have that $\dim((\errSpaceShot{i}_{\insertionsShot{i}})^\perp)=N_i-\insertionsShot{i}$ which therefore deletes exactly $\insertionsShot{i}$ dimensions from $\txSpaceShot{i}$ since $\dim(\rxSpaceShot{i})=N_i-\dim((\rxSpaceShot{i})^\perp)=\dim(\txSpaceShot{i})-\insertionsShot{i}+\deletionsShot{i}$ for all $i=1,\dots,\shots$.
  Hence, we can write~\eqref{eq:compl_erasure_op_2} as
  \begin{equation*}
    \rxSpaceShot{i}=\delOp{\nTransmitShot{i}-\insertionsShot{i}}(\txSpaceShot{i})\oplus\errSpaceShot{i}_{\deletionsShot{i}},\quad\forall i=1,\dots,\shots
  \end{equation*}
  which proves the statement.
\end{proof}

\section{Efficient Implementation of the Multishot Operator Channel}\label{app:draw_uniform_errors}

For implementing the multishot operator channel we need to draw tuples that contain subspaces and have a particular constant sum-dimension uniformly at random from the set of all such tuples efficiently.
In order to accomplish this subspace-tuple drawing algorithm efficiently we adapted the enumerative coding approach from~\cite{puchinger2020generic} for the case where the dimension of the ambient space as well as the dimensions of the transmitted spaces are equal in each shot, i.e. for $\NbarVec=(\Nbar, \Nbar,\dots,\Nbar)$ with $\Nbar \defeq \frac{N}{\shots}$ and $\nTransmitBarVec = (\nTransmitBar, \nTransmitBar, \dots, \nTransmitBar)$. 

Suppose that we want to draw tuples with fixed sum-dimension $\insertions$ uniformly at random from the set
\begin{equation}
    \{\errSpaceVec \in \ProjspaceAny{\NbarVec} : \sumDim(\errSpaceVec) = \insertions\}.
\end{equation}

Define the set $\compSet{\Nbar, \insertions, \shots} \defeq \left\{\insertionsVec \in \{0,\dots, \min\{\Nbar, \insertions\}\}^\shots : \sum_{i=1}^{\shots} \insertionsShot{i} = \insertions \right\}$.

Then then the number $\numSubTuples{\Nbar,t,\shots}$ of tuples in $\ProjspaceAny{\NbarVec}$ with sum-dimension $\insertions$ is given by
\begin{equation}
    \numSubTuples{\Nbar,\insertions,\shots} = \sum_{\insertionsVec \in \compSet{\Nbar, \insertions, \shots}} \prod_{i=1}^{\shots} \quadbinom{\Nbar}{\insertionsShot{i}}.
\end{equation}
Since the cardinality of the set $\compSet{\Nbar, \insertions, \shots}$ and thus the number of terms in the sum above is large for most parameters (see~\cite{puchinger2020generic}) we may compute $\numSubTuples{\Nbar,\insertions,\shots}$ recursively as
\begin{equation}
    \numSubTuples{\Nbar,\insertions,\shots} =
    \left\{
    \begin{array}{ll}
        \displaystyle
        \quadbinom{\Nbar}{\insertions}, & \text{if } \shots=1,
        \\ 
        \displaystyle\sum_{\insertions'=\max\{0, \insertions-(\shots-1)\Nbar\}}^{\min\{\Nbar,\insertions\}} \quadbinom{\Nbar}{\insertions'} \numSubTuples{\Nbar,\insertions - \insertions',\shots-1} & \text{if } \shots > 1,
    \end{array}
    \right.
\end{equation}
with $\numSubTuples{\Nbar,\insertions,\shots} = 0$ for $\insertions > \Nbar\shots$.

For $\shots=1$ the number coincides with the number of $\insertions$-dimensional vectors spaces of $\Fq^{\Nbar}$ which is $\quadbinom{\Nbar}{\insertions}$.
The recursion for $\shots > 1$ follows by summing over the product of all possible $\insertions'$-dimensional subspaces in the first block and the number of remaining $(\shots-1)$-tuples of sum-dimension $\insertions-\insertions'$.

As in~\cite{puchinger2020generic} this recursion provides an efficient method to draw dimension partitions of $\shots$-tuples with fixed sum-dimension that are chosen uniformly at random from $\ProjspaceAny{\NbarVec}$ which is given in Algorithm~\ref{alg:draw_dim_part}.

\begin{algorithm}[ht]
    \caption{
        \algoname{DrawDimensionPartition}
    }
  \label{alg:draw_dim_part}

  \begin{algorithmic}[1]
  \Statex \textbf{Input:} 
   Parameters $q$, $\Nbar$, $\insertions$, $\shots$ 
  
  \Statex

  \Statex \textbf{Output:} 
    Dimension partition $\insertionsVec \in \compSet{\Nbar, \insertions, \shots}$ of an $\shots$-tuple $\errSpaceVec$ with sum-dimension $\insertions$ that is chosen uniformly at random from $\Grassm{\NbarVec, \nTransmitBarVec}$  
  
  \Statex

  \State $D^{(1)} \assignRand \left\{1, \dots, \numSubTuples{\Nbar, \insertions,\shots}\right\}$ \Comment{Draw index uniformly at random}
  \State $\insertions_{1} \assignDet \insertions$

  \For{$i=1,\dots,\shots$}
    \State $\insertionsShot{i} \assignDet \max\!\left\{\insertions'' \!\!\in\! \{0, \dots, \! \insertions_{i}\!\} : \sum_{\insertions'= \insertions_{i}-(\shots-i)\Nbar}^{\insertions''-1} \quadbinom{\Nbar}{\insertions'} \! \numSubTuples{\Nbar,\insertions_{i}\!\!-\!\insertions',\shots\!-\!i} \!<\! D^{(i)}  \right\}$ 

    \State $D^{(i+1)} \assignDet D^{(i)} - \sum_{\insertions'=\insertions_{i}-(\shots-i)\Nbar}^{\insertionsShot{i}-1} \quadbinom{\Nbar}{\insertions'} \! \numSubTuples{\Nbar,\insertions_{i}\!\!-\!\insertions',\shots-i}$

    \State $\insertions_{i+1} \assignDet \insertions_{i} - \insertionsShot{i}$
  \EndFor

  \State \Return $\insertionsVec=(\insertionsShot{1},\dots,\insertionsShot{\shots})$
\end{algorithmic}
\end{algorithm}

Equipped with the efficient method to draw dimension partitions for sampling $\shots$-tuples with constant sum-dimension from $\ProjspaceAny{\NbarVec}$ uniformly at random (see Algorithm~\ref{alg:draw_dim_part}) we can implement the multishot operator channel (see~\eqref{eq:def:multishot_op_channel})
\begin{equation}
    \rxSpaceVec=\delOp{\nTransmit-\deletions}(\txSpaceVec) \oplus \errSpaceVec
\end{equation}
with input alphabet $\Grassm{\NbarVec, \nTransmitBarVec}$ and an overall number of $\insertions$ insertions and $\deletions$ deletions as follows.

First, we require the partition $\deletionsVec$ of the $\deletions$ deletions.
For each shot $i$ there are $\quadbinom{\nTransmitBar}{\deletionsShot{i}}$ ways to chose an $\deletionsShot{i}$-dimensional subspace of the transmitted space $\txSpaceShot{i} \in \Grassm{\Nbar, \nTransmitBar}$.
Hence, we can use the routine $\algoname{DrawDimensionPartition}(q, \nTransmitBar, \deletions, \shots)$ in Algorithm~\ref{alg:draw_dim_part} to draw the dimension partition $\deletionsVec$ for the deletions operator $\delOp{\nTransmit-\deletions}(\txSpaceVec)$.

Recall, that for the error space $\errSpaceVec$ we have the restriction that $\txSpaceVec \cap \errSpaceVec = \0$.
Hence, we have to sample $\errSpaceVec$ uniformly at random from $\Grassm{\NbarVec-\nTransmitBarVec, \insertionsVec}$( rather than from $\Grassm{\NbarVec, \insertionsVec}$) and thus have to call $\algoname{DrawDimensionPartition}(q, \NbarVec-\nTransmitBar, \insertions, \shots)$ in order to get the corresponding insertion partition $\insertionsVec$.

Once the partitions of insertions $\insertionsVec$ and deletions $\deletionsVec$ are obtained we can apply the single-shot operator channel from~\cite{koetter2008coding} to each shot with the corresponding number of $\insertionsShot{i}$ insertions and $\deletionsShot{i}$ deletions.
The whole procedure is provided in Algorithm~\ref{alg:multi_shot_op_channel}.

\begin{algorithm}[ht]
    \caption{
        \algoname{Random Instance of Multishot Operator Channel}
    }
  \label{alg:multi_shot_op_channel}

  \begin{algorithmic}[1]
  \Statex \textbf{Input:} 
   $\txSpaceVec \in \Grassm{\NbarVec, \nTransmitBarVec}$, number of insertions $\insertions$, number of deletions $\deletions$, \shots
  
  \Statex

  \Statex \textbf{Output:} 
    Received space $\rxSpaceVec \in \ProjspaceAny{\NbarVec}$   
  
  \Statex

  \State $\insertionsVec = (\insertionsShot{1}, \dots, \insertionsShot{\shots}) \assignDet \algoname{DrawDimensionPartition}(q, \Nbar-\nTransmitBar, \insertions, \shots)$

  \State $\deletionsVec = (\deletionsShot{1}, \dots, \deletionsShot{\shots}) \assignDet \algoname{DrawDimensionPartition}(q, \nTransmitBar, \deletions, \shots)$

  \For{i=1,\dots,\shots}
    \State $\errSpaceShot{i} \assignRand \Grassm{\Nbar, \insertionsShot{i}} \setminus \txSpaceShot{i}$ \Comment{Choose random $\errSpaceShot{i}$ s.t. $\txSpaceShot{i} \cap \errSpaceShot{i} = \{\0\}$}

    \State $\rxSpaceShot{i} \assignDet \delOp{\nTransmitBar - \deletionsShot{i}}(\txSpaceShot{i}) \oplus \errSpaceShot{i}$
  \EndFor

  \State \Return $\rxSpaceVec = (\rxSpaceShot{1}, \dots, \rxSpaceShot{\shots})$
\end{algorithmic}
\end{algorithm}

\end{appendices}

\end{document}